\documentclass[pra,floatfix,twocolumn,amsfonts,amsmath,amssymb,nofootinbib,reprint]{revtex4-1}
\pdfoutput=1

\usepackage[utf8]{inputenc}
\usepackage[T1]{fontenc}
\usepackage{accents}
\usepackage[colorlinks]{hyperref}

\usepackage{amsthm}
\usepackage{mathtools}
\usepackage{bbm}
\usepackage{bm}
\usepackage[normalem]{ulem}
\usepackage{graphicx}
\usepackage{tikz}
\usetikzlibrary{arrows}

\newtheorem{proposition}{Proposition}

\newtheorem{lemma}{Lemma}
\newtheorem{definition}{Definition}

\def\bra#1{\mathinner{\langle{#1}|}}
\def\ket#1{\mathinner{|{#1}\rangle}}

\def\ketbra#1{\mathinner{|{#1}\rangle\!\langle{#1}|}}

{\catcode`\|=\active 
  \gdef\Braket#1{\left<\mathcode`\|"8000\let|\BraVert {#1}\right>}}
\def\BraVert{\egroup\,\mid@vertical\,\bgroup}
\newcommand{\oprod}[2]{| #1 \rangle\langle #2 |}

\def\dket#1{\mathinner{|{#1}\rangle\!\rangle}}
\def\dketbra#1{\mathinner{|{#1}\rangle\!\rangle\!\langle\!\langle{#1}|}}

\DeclareMathOperator{\Tr}{Tr}

\renewcommand\L{\mathcal{L}}
\renewcommand\P{\mathcal{P}}
\newcommand{\M}{\mathcal{M}}
\newcommand{\HS}{\mathcal{H}}
\newcommand{\N}{\mathcal{N}}
\newcommand{\A}{\mathcal{A}}
\newcommand{\K}{\mathcal{K}}

\newcommand{\X}{\mathcal{X}}

\newcommand{\W}{\mathcal{W}}
\newcommand{\C}{\mathcal{C}}
\newcommand{\id}{\mathbbm{1}}

\newcommand{\sep}{\text{sep}}

\newtheorem{innercustomthm}{Theorem}
\newenvironment{customthm}[1]
  {\renewcommand\theinnercustomthm{#1}\innercustomthm}
  {\endinnercustomthm}

\begin{document}

\title{On the definition and characterisation of multipartite causal (non)separability}

\author{Julian Wechs}
\affiliation{Univ.\ Grenoble Alpes, CNRS, Grenoble INP\footnote{Institute of Engineering Univ. Grenoble Alpes}, Institut N\'eel, 38000 Grenoble, France}

\author{Alastair A. Abbott}
\affiliation{Univ.\ Grenoble Alpes, CNRS, Grenoble INP\footnote{Institute of Engineering Univ. Grenoble Alpes}, Institut N\'eel, 38000 Grenoble, France}

\author{Cyril Branciard}
\affiliation{Univ.\ Grenoble Alpes, CNRS, Grenoble INP\footnote{Institute of Engineering Univ. Grenoble Alpes}, Institut N\'eel, 38000 Grenoble, France}

\date{\today}

\begin{abstract}

The concept of causal nonseparability has been recently introduced, in opposition to that of causal separability, to qualify physical processes that locally abide by the laws of quantum theory, but cannot be embedded in a well-defined global causal structure. 
While the definition is unambiguous in the bipartite case, its generalisation to the multipartite case is not so straightforward.
Two seemingly different generalisations have been proposed, one for a restricted tripartite scenario and one for the general multipartite case. Here we compare the two, showing that they are in fact inequivalent. 
We propose our own definition of causal (non)separability for the general case, which---although \emph{a priori} subtly different---turns out to be equivalent to the concept of ``extensible causal (non)separability'' introduced before, and which we argue is a more natural definition for general multipartite scenarios. 
We then derive necessary, as well as sufficient conditions to characterise causally (non)separable processes in practice. 
These allow one to devise practical tests, by generalising the tool of witnesses of causal nonseparability.

\end{abstract}

\maketitle

\section{Introduction}

The notion of a causal order between events is an essential ingredient in our understanding of the world. 
Our conventional view of causality is that events are ordered according to some global time parameter, with past events influencing future events, but not \emph{vice versa}. 
One may however wonder whether this concept is really fundamental, or whether scenarios without such an underlying background causal structure are conceivable. 
The situation is particularly interesting in quantum theory, where the properties of physical systems are not always well-defined, and where the question arises of whether the causal structure itself can be subject to quantum effects in a similar way. 
These questions are of great importance for the foundations of physics~\cite{hardy05,oreshkov12,zych17}, but they are also motivated by a more practical point of view, as new resources for quantum information processing become available when the assumption of a definite causal structure is relaxed~\cite{chiribella13}.
Recent works have demonstrated that, for instance, indefinite causal orders can enable advantages in regard to query complexity~\cite{chiribella12,colnaghi12,araujo14,facchini15}, communication complexity~\cite{feix15,guerin16} and other information processing tasks~\cite{ebler18,salek18,chiribella18}.

A particular model describing causal relations between quantum events is the so-called process matrix formalism~\cite{oreshkov12}. 
In this framework, quantum events are assumed to take place locally, but the causal order between them is not specified \emph{a priori}. The physical resource relating the local events is described by a \emph{process matrix}, which, broadly speaking, is a generalisation of a multipartite density matrix allowing also for the description of signalling scenarios, such as quantum channels.
As it turns out, some scenarios arising within this formalism are indeed incompatible with any definite causal order. The process matrices corresponding to these scenarios are called \emph{causally nonseparable}, while the process matrices describing scenarios compatible with a well-defined causal structure are called \emph{causally separable}. 

The process matrix formalism was initially introduced for two local events. 
In that bipartite case, the notion of causal (non)separability is clearly defined and well understood. 
In particular, the causal (non)separability of any bipartite process matrix can be determined using \emph{witnesses of causal nonseparability}~\cite{araujo15,branciard16a}, similar conceptually to entanglement witnesses. 
In order to comprehensively understand causal indefiniteness from a fundamental perspective, and to explore more deeply the question of how they can be harnessed as a quantum information processing resource, it is essential to clarify how the absence of a causal order can be described, characterised and certified also in multipartite scenarios. 
While the formalism of process matrices generalises rather easily to more parties~\cite{baumeler14,araujo15,oreshkov16}, the notion of causal (non)separability becomes less clear. 
In fact, several different definitions have recently been proposed to generalise the bipartite case~\cite{araujo15,oreshkov16} which, as it turns out, are not equivalent.

In this work, we clarify the definition of causal (non)separability in multipartite scenarios. 
After recalling the framework and definitions in the bipartite case, we compare the generalisations of causal (non)separability that have been proposed so far, before proposing and motivating our own definition for the multipartite case (Definition~\ref{def:our_def-CS}). 
We then provide a characterisation of multipartite causally (non)separable processes via necessary as well as sufficient conditions (Propositions~\ref{prop:CS_charact_3}, \ref{prop:CS_necessary} and~\ref{prop:CS_sufficient}), allowing us to generalise the tool of witnesses of causal nonseparability.

\section{Process matrix formalism: \\ the basics}
\label{sec:W_formalism}

The process matrix formalism is perhaps most easily understood in the bipartite case~\cite{oreshkov12}. 
To begin with, let us briefly recall its framework for this case, before turning to the generalisation to multipartite scenarios.

\subsection{Bipartite process matrices}
\label{subsec:W_formalism_N_2}

The formalism of process matrices was introduced in Ref.~\cite{oreshkov12} to study correlations between events that locally obey the laws of quantum theory, but which are not \emph{a priori} embedded into any global causal order. 
In the bipartite scenario, two parties, who we shall call Alice ($A$) and Bob ($B$), are each associated with closed laboratories.  
The parties perform an experiment during which their interactions with the ``outside world'' (and hence with each other) are restricted to opening their laboratories only once to let an incoming physical system enter, and once to send out an outgoing system.
Alice and Bob may choose local operations to perform within their laboratories, possibly depending on some external (classical) input $x$ or $y$ for $A$ and $B$, and producing (classical) measurement outcomes $a$ and $b$, respectively.
The correlations established between the parties after repeating the experiment many times are described by the conditional probability distribution $P(a,b|x,y)$.

While no assumption is made about the global causal order between the parties, we assume that the local operations performed inside the laboratories are described by standard quantum theory. 
We can therefore assign some ``incoming'' and ``outgoing'' Hilbert spaces to the parties, which we denote $\HS^{A_I}, \HS^{A_O}$ (for Alice) and  $\HS^{B_I}, \HS^{B_O}$ (for Bob), of dimensions $d_{A_I}, d_{A_O}, d_{B_I}$ and $d_{B_O}$, respectively. 
The spaces of Hermitian linear operators over these Hilbert spaces will simply be denoted by $A_I, A_O, B_I$ and $B_O$. 
For convenience we also define $A_{IO} \coloneqq A_I \otimes A_O$, $B_{IO} \coloneqq B_I \otimes B_O$, $d_{A_{IO}} \coloneqq d_{A_I} d_{A_O}$ and $d_{B_{IO}} \coloneqq d_{B_I} d_{B_O}$. 
In this paper, we will only consider finite-dimensional Hilbert spaces; 
for a generalisation of the framework to infinite-dimensional systems, see Ref.~\cite{giacomini16}.

According to quantum theory, Alice and Bob's local operations can most generally be described as quantum instruments~\cite{davies70}---that is, sets of completely positive (CP) maps that sum up to completely positive trace-preserving (CPTP) maps. 
The Choi-Jamio\l{}kowski (CJ) isomorphism~\cite{jamiolkowski72,choi75} allows us to represent these CP maps as positive semidefinite matrices $M_{a|x}^{A_{IO}}, M_{b|y}^{B_{IO}}$, and the CPTP maps as positive semidefinite matrices $M_{x}^{A_{IO}} \coloneqq \sum_a M_{a|x}^{A_{IO}}$, $M_{y}^{B_{IO}} \coloneqq \sum_b M_{b|y}^{B_{IO}}$ that satisfy $\Tr_{A_O} M_{x}^{A_{IO}} = \id^{A_I}$ and $\Tr_{B_O} M_{y}^{B_{IO}} = \id^{B_I}$. 
Here, $\Tr_X$ denotes the partial trace over the system $X$, and $\id^X$ denotes the identity operator in the space $X$ (in general, superscripts on operators, which may be omitted when clear enough, denote the system(s) they apply to).

As shown in Ref.~\cite{oreshkov12}, requiring compatibility with quantum mechanics locally and assuming the non-contextuality of the probabilities imply that the probabilities $P(a,b|x,y)$ must be bilinear in the CP maps associated with the operations of $A$ and $B$---or, equivalently, bilinear in their CJ representations. 
(Throughout this paper we will often refer to CP maps by their equivalent CJ representation and \emph{vice versa}.)
It follows that the overall process can be described by a Hermitian operator, a ``process matrix'' $W \in A_{IO} \otimes B_{IO}$~\cite{oreshkov12}, such that the correlations are obtained via the \emph{generalised Born rule}
\begin{equation} \label{eq:born_rule_bipartite}
P(a,b|x,y) = \Tr\left[ M_{a|x}^{A_{IO}} \otimes M_{b|y}^{B_{IO}} \cdot W\right] 
\end{equation}
(where $\Tr$ is now the full trace).

The framework also permits the parties to share, in addition to the process matrix, some (possibly entangled) ancillary quantum state that can be accessed via their local operations. 
The parties may thus have access also to some extra incoming Hilbert spaces $\HS^{A_{I'}}$ and $\HS^{B_{I'}}$ of arbitrary (finite) dimension, and be able to perform CP maps $M_{a|x}^{A_{II'O}} \in A_{II'O} \coloneqq A_I \otimes A_{I'} \otimes A_O$ and $M_{b|y}^{B_{II'O}} \in B_{II'O} \coloneqq B_I \otimes B_{I'} \otimes B_O$, respectively (where as before, $A_{I'}$ and $B_{I'}$ are the spaces of Hermitian linear operators over $\HS^{A_{I'}}$ and $\HS^{B_{I'}}$).
This implies that any process matrix $W \in A_{IO} \otimes B_{IO}$ can be extended to a process matrix $W \otimes \rho \in A_{II'O} \otimes B_{II'O}$, for any extra incoming spaces $A_{I'}, B_{I'}$ and any $\rho\in A_{I'} \otimes B_{I'}$~\cite{oreshkov12}.

Requiring Eq.~\eqref{eq:born_rule_bipartite} to yield valid (i.e., nonnegative and normalised) probabilities, even when the parties share arbitrary ancillary states, is equivalent to $W$ satisfying the following constraints:
\begin{equation}\label{eq:valid_bipartite}
W \ge 0, \ \ W \in \L^{\{A,B\}}, \ \ \text{and} \ \ \Tr W = d_{A_O} d_{B_O} 
\end{equation}
for some particular linear subspace $\L^{\{A,B\}}$ of $A_{IO} \otimes B_{IO}$; see Sec.~\ref{subsec:valid_W_and_Afirst} and Appendix~\ref{app:charact_valid_Ws} for an explicit characterisation~\cite{oreshkov12,araujo15}. 
In the following we will refer to a matrix satisfying the first two constraints above (i.e., without necessarily imposing the normalisation constraint $\Tr W = d_{A_O} d_{B_O}$) as a \emph{valid} process matrix, and whenever we talk about a process matrix $W$ we always implicitly assume it is valid. 
Hermitian matrices that are not valid process matrices will simply be referred to as ``matrices''.

\subsection{Bipartite causal (non)separability}
\label{sec:bipartite_CS}

One may now consider the question, whether the situation described by a process matrix can be embedded in a well-defined causal structure, with a fixed causal order between the events happening in each party's laboratory, or not.

A process matrix is said to be ``compatible with (the causal order) $A \prec B$'' (sometimes abbreviated to just ``$A \prec B$'', e.g., in superscripts) if all the correlations it generates are compatible with a causal order where $A$ acts before $B$, which is to be understood operationally: such a process matrix $W^{A\prec B}$ does not allow for any signalling from $B$ to $A$. 
More precisely, whatever the CP and CPTP maps $M_{a|x}^{A_{IO}}, M_{y^{(\prime)}}^{B_{IO}}$ of $A$ and $B$, the resulting correlations respect the no-signalling condition $P(a|x,y) = P(a|x,y')$, or $\Tr[ M_{a|x}^{A_{IO}} \otimes M_{y}^{B_{IO}} \cdot W^{A \prec B}] = \Tr[ M_{a|x}^{A_{IO}} \otimes M_{y'}^{B_{IO}} \cdot W^{A \prec B}]$ according to Eq.~\eqref{eq:born_rule_bipartite}.
This constrains $W^{A \prec B}$ to be in a linear subspace $\L^{A \prec B}\subset \L^{\{A,B\}}$ of $A_{IO} \otimes B_{IO}$; see Sec.~\ref{subsec:valid_W_and_Afirst} and Appendix~\ref{app:fixed_order} for an explicit characterisation of $\L^{A \prec B}$. 

Likewise, process matrices that do not allow signalling from $A$ to $B$ are said to be compatible with the causal order $B \prec A$, and will typically be denoted $W^{B \prec A} \in \L^{B \prec A}$. 
One can also conceive of situations where the causal order is not fixed to be the same for all experimental runs, but where there is instead a probabilistic mixture of the two possibilities. 
Such a scenario is described by a convex combination of process matrices compatible with $A \prec B$ and $B \prec A$, respectively.
Process matrices of this form remain compatible with an underlying causal framework and are the subject of the following definition, first introduced by Oreshkov, Costa and Brukner~\cite{oreshkov12}:
\begin{definition}[Bipartite causal (non)separability~\cite{oreshkov12}] \label{def:biCS}
A bipartite process matrix $W$ is said to be \emph{causally separable} if and only if it can be written as a convex combination
\begin{equation}\label{def:caus_sep_bi}
	 W \, = \, q \, W^{A \prec B} \, + \, (1{-}q) \, W^{B \prec A} \, , 
\end{equation}
with $q \in [0,1]$ and where $W^{A \prec B}$ and $W^{B \prec A}$ are two process matrices compatible with the causal orders $A \prec B$ and $B \prec A$, respectively.

A process matrix that cannot be decomposed as above is said to be \emph{causally nonseparable}.
\end{definition}

Causally separable process matrices thus describe the most general bipartite situations where one can identify a definite causal order between the parties, be it fixed for all experimental runs or subject to classical randomness. 
In contrast, if a process matrix is causally nonseparable, it is incompatible with any causal order between $A$ and $B$. 
In the bipartite case, causal (non)separability can be easily and efficiently verified; in particular, any causally nonseparable process can be detected using a \emph{witness of causal nonseparability}~\cite{araujo15,branciard16a} (see Sec.~\ref{subsec:witnesses}).

\subsection{Towards generalising to more parties}
\label{sec:W_def_multipartite}

The process matrix framework itself generalises rather easily to the multipartite case.

Let us first introduce some generalised notations. We shall consider $N$ parties denoted by $A_k$ for $k \in \{ 1, \ldots, N\} \coloneqq \N$, with corresponding inputs and outputs denoted by $x_k,$ and $a_k$, respectively. 
We define the input and output vectors $\vec x \coloneqq (x_1, \ldots, x_N)$ and $\vec a \coloneqq (a_1, \ldots, a_N)$. 
The ``incoming'' and ``outgoing'' Hilbert spaces for each party are denoted by $\HS^{A_I^k}, \HS^{A_O^k}$ (of dimensions $d_{A_I^k}, d_{A_O^k}$, respectively), while the spaces of Hermitian linear operators over these Hilbert spaces are denoted by $A_I^k, A_O^k$. 
We also define $A_{IO}^k \coloneqq A_I^k \otimes A_O^k$, and $d_{A_{IO}^k} \coloneqq d_{A_I^k} d_{A_O^k}$. 

For a subset $\K \subseteq \N$ of parties, we will denote by $\vec x_\K$ and $\vec a_\K$ the vectors of inputs and outputs restricted to the parties in $\K$, and use shorthand notations like $A_{IO}^\K \coloneqq \bigotimes_{k \in \K} A_{IO}^k$ ($= \mathbb{R}$ if $\K = \emptyset$), $\id^\K \coloneqq \bigotimes_{k \in \K} \id^{A_{IO}^k} = \id^{A_{IO}^\K}$, and $\Tr_\K$ for the trace over all (incoming and outgoing) systems of the parties in $\K$---i.e., $\Tr_{A^\K_{IO}}$ or $\Tr_{A^\K_{II'O}}$, as appropriate (see below), and with $\Tr_{\emptyset}$ the identity operation and $\Tr_{\N}$ the full trace.
For notational simplicity, we shall identify the parties' names with their labels, and singletons of parties (e.g., $\{A_k\}$) with the parties themselves (e.g., $A_k$) or the corresponding label, so that $\N = \{ 1, \ldots, N\} \equiv \{ A_1, \ldots, A_N\}$, $\N \backslash \{A_k\} \equiv \N \backslash k$, $\Tr_{\{A_k\}} \equiv \Tr_k$, etc.

The CP maps corresponding to the parties' operations are then denoted by $M_{a_k|x_k}^{A_{IO}^k}$, the corresponding CPTP maps $M_{x_{k}}^{A^{k}_{IO}}\coloneqq \sum_{a_{k}}M_{a_{k}|x_{k}}^{A^{k}_{IO}}$, and the overall process is represented by a process matrix $W \in A_{IO}^\N$.
The resulting correlations are then obtained through a generalised Born rule as before:
\begin{equation} \label{eq:born_rule_Npartite}
P(\vec a | \vec x) = \Tr\left[ M_{a_1|x_1}^{A_{IO}^1} \otimes \cdots \otimes M_{a_N|x_N}^{A_{IO}^N} \cdot W\right]. 
\end{equation}

As in the bipartite case, the parties may also share some ancillary state $\rho$ in some extra incoming spaces $A_{I'}^1 \otimes \cdots \otimes A_{I'}^N = A_{I'}^\N$, and extend their local operations to act on these spaces as well.
Requiring again the nonnegativity and normalisation of all obtainable probabilities, including for arbitrary extensions $W\otimes\rho$ of $W$, imposes validity constraints on $W$. 
In the general multipartite case, they read
\begin{equation}
W \ge 0, \ \ W \in \L^\N, \ \ \text{and} \ \ \Tr W = \prod_{k \in \N} d_{A_O^k} \label{eq:valid_Npartite}
\end{equation}
for some particular linear subspace $\L^\N$ of $A_{IO}^\N$; see Sec.~\ref{subsec:valid_W_and_Afirst} and Appendix~\ref{app:charact_valid_Ws}~\cite{oreshkov12,araujo15}. 
As for the bipartite case, in this paper a matrix will be called a (valid) process matrix whenever it satisfies the first two constraints above, without necessarily requiring that it is correctly normalised.

\medskip

The no-signalling constraints can readily be generalised to the $N$-partite case, allowing the notion of compatibility with a fixed causal order to be extended accordingly.
For instance, a process matrix is said to be compatible with the fixed causal order ${A_1 \prec A_2 \prec \cdots \prec A_N}$ if no party or group of parties can signal to other parties in their causal ``past'' (as defined by the specified causal order)---which translates into the constraint that $P(a_1,\ldots,a_k|\vec x) = P(a_1,\ldots,a_k|x_1,\ldots,x_k)$ for all $k = 1, \ldots, N-1$.
As before, this constrains such a process matrix $W^{A_1 \prec \cdots \prec A_N}$ to be in a linear subspace $\L^{A_1 \prec \cdots \prec A_N}\subset \L^\N$ of $A_{IO}^\N$; see Sec.~\ref{subsec:valid_W_and_Afirst} for an explicit characterisation of $\L^{A_1 \prec \cdots \prec A_N}$ (and Appendices~\ref{app:fixed_order}--\ref{app:allowed_forbidden_terms} for further discussions and characterisations of process matrices compatible with other fixed causal orders).

What is not so straightforward, however, is to generalise the concept of causal (non)separability, which turns out to be much more subtle for more than two parties.
In particular, additional complexity arises in the multipartite case because the causal order can be \emph{dynamical} as well as probabilistic---that is, the causal order of parties in the future can depend on operations of parties in the past~\cite{baumeler14,oreshkov16,abbott16}. 
Simply considering a convex combination of process matrices compatible with different fixed causal orders does not include scenarios with such dynamical causal orders, and is therefore too restrictive to capture all scenarios that should be considered compatible with a well-defined causal order. Perhaps more strikingly, as we shall see the possibility to extend process matrices with ancillary quantum states has nontrivial implications for the definition of causal (non)separability for more than two parties~\cite{oreshkov16}.
The main objectives of this paper are precisely to discuss how the concept of causal (non)separability should properly be generalised to the multipartite case, and to characterise causally separable and causally nonseparable process matrices.

\section{Defining multipartite causal (non)separability}
\label{sec:defs}

\subsection{Ara\'ujo \emph{et al.}'s definition}
\label{subsec:araujo_def}

The multipartite case was first considered in a restricted tripartite situation in which one party has no (or, equivalently, a trivial) outgoing system. 
This particular scenario was studied because of its relevance for a practical protocol where the causal order between two parties $A$ and $B$, which perform some unitary operations $U_A$ and $U_B$ on a target system initialised in a state $\ket{\psi}^t$, is controlled by another (two-dimensional) quantum system. If this control qubit is initialised in the state $\ket{0}^c$, the operation $U_A$ is applied before $U_B$, while for a control qubit in the state $\ket{1}^c$, $U_B$ is applied before $U_A$. If the control qubit is initialised in a superposition state $\ket{+}^c = \frac{1}{\sqrt{2}}(\ket{0}^c + \ket{1}^c)$, the overall transformation on the joint state of the target and control systems is thus
\begin{align}
& \ket{\psi}^t \otimes \ket{+}^c \notag \\
& \to \frac{1}{\sqrt{2}}(U_B U_A \ket{\psi}^t \otimes \ket{0}^c + U_A U_B \ket{\psi}^t \otimes \ket{1}^c) \, ,
\end{align}
i.e., the unitaries are applied in a ``superposition of orders''. The output state is then sent to a third party $C$ (Charlie) who can measure the control qubit, and possibly also the target system.
The protocol just described can straightforwardly be generalised to the case where $A$ and $B$'s operations are general quantum instruments instead of unitaries.
This so-called \emph{quantum switch} can be understood as a quantum supermap~\cite{chiribella08}, or higher order transformation, that maps $A$ and $B$'s local operations to the overall global transformation. It cannot be realised by inserting the local operations into a circuit with a well-defined causal order, and therefore constitutes a new resource for quantum computation that goes beyond causally ordered quantum circuits~\cite{chiribella13}. It has attracted particular interest as a consequence of being readily implementable, and indeed several implementations have been experimentally realised~\cite{procopio15,rubino17,rubino17a,goswami18,wei18}.
Consequent work has sought to clarify whether such implementations can really be seen as genuine realisations of indefinite causal orders, and Ref.~\cite{oreshkov18} gives arguments clarifying why they can be.

The quantum switch can naturally be described in the process matrix formalism~\cite{araujo15,oreshkov16} where it indeed corresponds to a tripartite process matrix for parties $A$, $B$ and $C$, where Charlie has no outgoing system and therefore cannot signal to the other parties. 
The situation is thus relatively similar to the bipartite case, since the only relevant causal orders are those where Charlie acts last, i.e., $A \prec B \prec C$ and $B \prec A \prec C$. 
This observation led Ara\'ujo \emph{et al.} to propose the following definition (as an initial, ``1-step'' generalisation of Definition~\ref{def:biCS}) for this particular scenario:

\begin{definition}[Ara\'ujo \emph{et al.}'s causal separability~\cite{araujo15}] \label{def:AB+-CS}
In a tripartite scenario where party $C$ has no outgoing system, a process matrix $W$ is said to be \emph{causally separable} if and only if it can be written as a convex combination
\begin{equation}
	 W \, = \, q \, W^{A \prec B \prec C} \, + \, (1{-}q) \, W^{B \prec A \prec C} \, , \label{def:caus_sep_tri}
\end{equation}
with $q \in [0,1]$ and where $W^{A \prec B \prec C}$ and $W^{B \prec A \prec C}$ are two process matrices compatible with the causal orders $A \prec B \prec C$ and $B \prec A \prec C$, respectively.
\end{definition}

It was shown that the process matrix describing the quantum switch is causally nonseparable as per Definition~\ref{def:AB+-CS}~\cite{araujo15}, and this definition has subsequently been used e.g.\ in Refs.~\cite{branciard16a,giacomini16,goswami18}.

\subsection{Oreshkov and Giarmatzi's definitions}

While Ara\'ujo \emph{et al.}'s definition recalled above applied only to a particular tripartite situation, Oreshkov and Giarmatzi~(OG) considered in Ref.~\cite{oreshkov16} the general multipartite case---taking into account, in particular, the possibility of dynamical causal orders. 
They defined in fact two possible generalisations of bipartite causal (non)separability, namely what they called the notions of ``causal (non)separability'' and ``extensible causal (non)separability''.

The definition they proposed for causal separability is recursive, in analogy with the definition of multipartite ``causal correlations''~\cite{oreshkov16,abbott16}---correlations that are compatible with a definite causal order. In Refs.~\cite{oreshkov16,abbott16}, these were characterised as those for which it is possible to identify, up to some probability, a party that acts first, and such that, for any behaviour of this first party, the conditional correlations shared by the remaining parties are again causal. 
Oreshkov and Giarmatzi invoked an analogous ``unraveling argument'' for causally separable processes.

More specifically, their definition is based on the concept of a ``conditional (process) matrix'', defined for a given matrix $W$ and a given CP map $M_k \coloneqq M_{a_k|x_k}^{A_{IO}^k}$ applied by a party $A_k$ as
\begin{align}
W_{|M_k} \coloneqq \Tr_k \left[M_k \otimes \id^{\N \backslash k} \, \cdot \, W\right]. \label{eq:conditional_W}
\end{align}
In general, even if $W$ is a valid process matrix, $W_{|M_k}$ thus defined may not be a valid process matrix (in which case we shall just talk about a ``conditional matrix'').
In fact, as we will see in Sec.~\ref{subsec:valid_W_and_Afirst}, a process matrix $W$ is compatible with party $A_k$ acting first (i.e., it does not allow signalling from the other parties to $A_k$) if and only if for any CP map $M_k$ the conditional matrix $W_{|M_k}$, as defined in Eq.~\eqref{eq:conditional_W}, is (up to normalisation%
\footnote{For a properly normalised process matrix $W$ compatible with $A_k$ first (i.e., which always gives $P(a_k|\vec x) = P(a_k|x_k)$) and a trace-non-increasing CP map $M_k = M_{a_k|x_k}$, one has $\Tr W_{|M_k} = P(a_k|x_k) \prod_{j \in \N\backslash k} d_{A_O^j}$, so that $W_{|M_k}$ must be divided by the factor $P(a_k|x_k)$ to also be properly normalised according to Eq.~\eqref{eq:valid_Npartite}. \label{footnote:norm_cond_process}}%
) a valid $(N{-}1)$-partite process matrix for the parties in $\N\setminus k$. 
In that case, the conditional process matrix $W_{|M_k}$ then represents the process shared by these $N{-}1$ parties, conditioned on party $A_k$ performing the CP map $M_k = M_{a_k|x_k}^{A_{IO}^k}$ (i.e., conditioned on both receiving the input $x_k$ and obtaining the outcome $a_k$).

Oreshkov and Giarmatzi then proposed the following (recursive) definition:%
\footnote{More precisely, what we present here as their definition is actually presented in Ref.~\cite{oreshkov16} (in a slightly different, but equivalent way) as a characterisation following from a more fundamental recursive definition of causally separable processes (not necessarily quantum mechanical).}

\begin{definition}[Oreshkov and Giarmatzi's causal separability~\cite{oreshkov16}] \label{def:OG-CS}
For $N=1$, any process matrix is causally separable. For $N \ge 2$, an $N$-partite process matrix $W$ is said to be \emph{causally separable} if and only if it can be decomposed as
\begin{equation}
 W = \sum_{k \in \N} q_k \, W_{(k)} , \label{eq:def:OG-CS}
\end{equation}
with $q_k\ge 0$, $\sum_k q_k = 1$, and where for each $k$, $W_{(k)}$ is a process matrix compatible with party $A_k$ acting first, and is such that for any possible CP map $M_k \in A_{IO}^k$ applied by party $A_k$, the conditional $(N{-}1)$-partite process matrix $(W_{(k)})_{|M_k} \coloneqq \Tr_k [M_k \otimes \id^{\N \backslash k} \, \cdot \, W_{(k)}]$ is itself causally separable.
\end{definition}

As outlined in the previous section, the process matrix framework allows for process matrices to be extended by providing additional ancillary states the the parties. 
Taking this into account, OG introduced a second definition of causal separability for process matrices that are causally separable even under arbitrary such extensions:

\begin{definition}[Oreshkov and Giarmatzi's extensible causal separability~\cite{oreshkov16}] \label{def:OG-ECS}
An $N$-partite process matrix $W$ is said to be \emph{extensibly causally separable} if and only if it is causally separable (as per Definition~\ref{def:OG-CS} above), and it remains so under any extension with incoming systems in an arbitrary joint quantum state---i.e., if and only if for any extension $A_{I'}^\N$ of the parties' incoming systems and any ancillary quantum state $\rho \in A_{I'}^\N$, $W \otimes \rho$ is causally separable.
\end{definition}

It is easy to see that OG's causal separability (CS) and extensible causal separability (ECS) are equivalent in the bipartite case, and, indeed, equivalent to Definition~\ref{def:biCS} given in Sec.~\ref{sec:bipartite_CS}: the process matrix $W \otimes \rho$ obtained by attaching an ancillary state $\rho$ to a causally separable process matrix $W$ of the form of Eq.~\eqref{def:caus_sep_bi} remains of the same form, with $W^{A \prec B} \otimes \rho$ ($W^{B \prec A} \otimes \rho$) compatible with $A$ acting before $B$ ($B$ before $A$), and for both terms $W^{A \prec B} \otimes \rho$ and $W^{B \prec A} \otimes \rho$, whatever operation the first party applies, the resulting conditional process matrix for the other party is single-partite, hence trivially causally separable.

However, OG's CS and ECS are not equivalent in the general multipartite case and thus indeed represent two different possible multipartite generalisations of the same bipartite concept. 
Of course ECS implies CS, but the converse is not true in general---the result of a phenomenon called ``activation of causal nonseparability'' in Ref.~\cite{oreshkov16}. 
An explicit example of a CS process that is not ECS was indeed given in~\cite{oreshkov16}, in a tripartite scenario where one party has no incoming system; we will see another example in the following subsection.

\subsection{Comparison}
\label{subsec:comparison}

We thus now have three potential generalisations of the concept of causal separability to the particular tripartite situation where one party has no outgoing system---namely, the two different definitions of causal separability (Definitions~\ref{def:AB+-CS} and~\ref{def:OG-CS}), and that of extensible causal separability (Definition~\ref{def:OG-ECS}). 
How do they relate to one another? 
Are the two definitions of causal separability indeed equivalent?
These questions are answered by the following result:

\begin{proposition} \label{prop:comp_3_defs}
In a tripartite scenario where party $C$ has no outgoing system, Ara\'ujo \emph{et al.}'s definition of causal separability (Definition~\ref{def:AB+-CS}) is equivalent to Oreshkov and Giarmatzi's definition of extensible causal separability (Definition~\ref{def:OG-ECS}), but nonequivalent to their definition of causal separability (Definition~\ref{def:OG-CS}).
\end{proposition}

The equivalence between Definitions~\ref{def:AB+-CS} and~\ref{def:OG-ECS} for this particular tripartite scenario is proved explicitly in Appendix~\ref{app:charact_CS_3_dCO1}, which we refer to for more details; we simply summarise the argument here as follows.
Clearly, any process matrix $W$ of the form of Eq.~\eqref{def:caus_sep_tri} is ECS, as any $W \otimes \rho$ is also of that form (and of the form also of Eq.~\eqref{eq:def:OG-CS}), and for any $W^{A \prec B \prec C}$ and any $M_A$, the conditional process $(W^{A \prec B \prec C})_{|M_A}$ is compatible with the order $B \prec C$ (hence it is causally separable; similarly for any $W^{B \prec A \prec C}$ and any $M_B$). 
The proof that an ECS process matrix $W$ necessarily has the form of Eq.~\eqref{def:caus_sep_tri} is based on a ``teleportation technique'' (see Lemma~\ref{lemma:teleportation} in Appendix~\ref{app:characterisation_CS_Ws}), already used in Ref.~\cite{oreshkov16}, that consists in introducing an ancillary system in a maximally entangled state $\rho$ shared by two parties, e.g.\ $A$ and $C$. 
By definition, the global process $W \otimes \rho_{A_{I'}C_{I'}}$ has a decomposition of the form~\eqref{eq:def:OG-CS}. 
It is then easy to see that the terms $W_A$ and $W_B$ compatible with parties $A$ or $B$ acting first are in fact compatible, since $C$ has no outgoing system, with the causal orders $A \prec B \prec C$ and $B \prec A \prec C$, respectively, and thus contribute to the terms $W^{A \prec B \prec C}$ and $W^{B \prec A \prec C}$ in Eq.~\eqref{def:caus_sep_tri}. 
For the term $W_C$ compatible with $C$ acting first, letting $C$ project his systems $C_{II'} \coloneqq C_I \otimes C_{I'}$ onto the maximally entangled state effectively ``teleports'' his system to $A$. 
By definition, the conditional bipartite process then shared by $A$ and $B$ must be causally separable, and must therefore have a decomposition of the form~\eqref{def:caus_sep_bi}, which also leads to a decomposition of the form~\eqref{def:caus_sep_tri} for $W_C$. 

\medskip

In order to prove the nonequivalence between Ara\'ujo \emph{et al.} and OG's definitions of causal separability, we will now show that OG's CS and ECS are nonequivalent---i.e., that there can be ``activation of causal nonseparability'' (according to OG's terminology)---in the scenario where party $C$ has no outgoing system.
Note that this scenario differs from that in which OG already gave an example of activation of causal nonseparability: they indeed considered a tripartite case where $C$ has no \emph{incoming} system, rather than no \emph{outgoing} system.

Consider for that the following process matrix:
\begin{align}
\hspace{-2mm} W^\text{act.} \! &\coloneqq \frac18 \Big[\id(\id\id{-}\hat{\textsc{z}}\hat{\textsc{z}})\id\id + \frac{\sqrt{3}}{4} \, \id(\hat{\textsc{x}}\hat{\textsc{x}}{+}\hat{\textsc{y}}\hat{\textsc{y}})(\hat{\textsc{z}}\id{+}\id \hat{\textsc{z}}) \notag \\
& \qquad + \frac12 \hat{\textsc{z}}(\hat{\textsc{z}}\id{-}\id \hat{\textsc{z}})\id\id + \frac14 \hat{\textsc{x}}(\hat{\textsc{x}}\hat{\textsc{y}}{-}\hat{\textsc{y}}\hat{\textsc{x}})(\hat{\textsc{z}}\id{-}\id \hat{\textsc{z}}) \Big] , \! \! \label{eq:def_W_activ}
\end{align}
where the subsystems are written, for convenience, in the order $C_IA_IB_IA_OB_O$ (i.e., $W^\text{act.} \in C_I \otimes A_I \otimes B_I \otimes A_O \otimes B_O$). 
Here, as in the other examples presented in this paper, $\hat{\textsc{x}}, \hat{\textsc{y}}, \hat{\textsc{z}}$ denote the Pauli matrices, $\id$ denotes the $2 \times 2$ identity matrix and tensor products between all matrices are implicit.

We note first that $W^\text{act.}$ is compatible with Charlie acting first---i.e., with the order $C \prec \{A,B\}$.%
\footnote{As $C$ has no outgoing system, $W^\text{act.}$ is also compatible with $C$ acting last (see Appendix~\ref{app:fixed_order_part_cases}). But to prove that $W^\text{act.}$ is CS (according to OG's definition) as we do below we need to consider $C$ acting first.}
(Indeed, it satisfies Eq.~\eqref{eq:constr_L_k_first} given later, for $A_k=C$.)
Any CP map applied by Charlie---i.e., since $C$ has no outgoing system, any element of a positive-operator valued measure (POVM) in his qubit incoming space $C_I$---can be written as $M_{\vec c} = \id + \vec c \cdot \vec \sigma$, where $\vec \sigma \coloneqq (\hat{\textsc{x}}, \hat{\textsc{y}}, \hat{\textsc{z}})$ and $\vec c \coloneqq (c_{\textsc{x}}, c_{\textsc{y}}, c_{\textsc{z}})$ is a 3-dimensional real vector with $|\vec c| \leq 1$, so that $M_{\vec c} \ge 0$ (and where we ignore the trace-nonincreasing constraint, and indeed the overall normalisation of $M_{\vec c}$, since it is irrelevant for our argument).
The resulting conditional matrix for parties $A$ and $B$ (as defined in Eq.~\eqref{eq:conditional_W}) is then
\begin{align}
(W^\text{act.})_{|M_{\vec c}} \coloneqq & \, \Tr_{C_I} [M_{\vec c}^{C_I} \otimes \id^{A_IB_IA_OB_O} \cdot W^\text{act.}] \notag \\
=& \, \frac12 W_{|M_{\vec c}}^{A \prec B} + \frac12 W_{|M_{\vec c}}^{B \prec A} \label{eq:WMc_decomp}
\end{align}
with (written in the order $A_IB_IA_OB_O$)
\begin{align}
\!\!\!\! W_{|M_{\vec c}}^{A \prec B} & \coloneqq \frac14 \Big[(\id\id{-}\hat{\textsc{z}}\hat{\textsc{z}})\id\id + \frac{\sqrt{3}}{2} (\hat{\textsc{x}}\hat{\textsc{x}}{+}\hat{\textsc{y}}\hat{\textsc{y}})\hat{\textsc{z}}\id \notag \\
& \qquad + \frac{c_{\textsc{z}}}{2} (\hat{\textsc{z}}\id-\id \hat{\textsc{z}})\id\id + \frac{c_{\textsc{x}}}{2} (\hat{\textsc{x}}\hat{\textsc{y}}{-}\hat{\textsc{y}}\hat{\textsc{x}})\hat{\textsc{z}}\id \Big]
\end{align}
and with $W_{|M_{\vec c}}^{B \prec A}$ of a similar form, obtained from $W_{|M_{\vec c}}^{A \prec B}$ by changing $\hat{\textsc{z}}^{A_O}\id^{B_O}$ to $\id^{A_O}\hat{\textsc{z}}^{B_O}$ and $c_{\textsc{x}}$ to $-c_{\textsc{x}}$.

Note that $W_{|M_{\vec c}}^{A \prec B}$ and $W_{|M_{\vec c}}^{B \prec A}$ are valid, causally ordered process matrices, compatible with $A \prec B$ and $B \prec A$, respectively (their eigenvalues are found to be $0, \frac12 \big(1 \pm \sqrt{\frac{3+c_{\textsc{x}}^2+c_{\textsc{z}}^2}{4}}\big) \geq 0$ for $|\vec c| \leq 1$, and they satisfy the appropriate form of Eq.~\eqref{eq:constr_order_A1__AN} given later). 
From Eq.~\eqref{eq:WMc_decomp} and the definition of causal separability in the bipartite case (Definition~\ref{def:biCS}), we conclude that for any CP map (i.e.\ here, any POVM element) $M_{\vec c}$ applied by Charlie, $(W^\text{act.})_{|M_{\vec c}}$ is a (bipartite) causally separable process. 
Therefore, according to OG's Definition~\ref{def:OG-CS}, $W^\text{act.}$ is a tripartite CS process (with a single term in the decomposition~\eqref{eq:def:OG-CS}, corresponding to $C$ first).

A crucial feature of the decomposition~\eqref{eq:WMc_decomp} is that $(W^\text{act.})_{|M_{\vec c}}$, $W_{|M_{\vec c}}^{A \prec B}$ and $W_{|M_{\vec c}}^{B \prec A}$ all depend on Charlie's operation $M_{\vec c}$.
Even though any valid process matrix in $C_IA_IB_IA_OB_O$ (including $W_{|M_{\vec c}}^{A \prec B}$ and $W_{|M_{\vec c}}^{B \prec A}$) is compatible with $C$ acting last (since $C$ has no outgoing system), the decomposition~\eqref{eq:WMc_decomp} still does not allow us to obtain a decomposition of the form $W^\text{act.} = \frac12 W^{A \prec B \prec C} + \frac12 W^{B \prec A \prec C}$ for $W^\text{act.}$ (or even with different weights $q$, $1{-}q$), as in Eq.~\eqref{def:caus_sep_tri}.
Indeed, such a decomposition for $W^\text{act.}$, with $W^{A \prec B \prec C}$ and $W^{B \prec A \prec C}$ valid process matrices compatible with the indicated causal order, does not exist. 
This can be shown using Ara\'ujo \emph{et al.}'s technique of ``witnesses of causal nonseparability''~\cite{araujo15,branciard16a}: one can construct a witness for $W^\text{act.}$, and we give one explicitly in Appendix~\ref{app:witness_Wactiv}.

Since, as stated above, the existence of such a decomposition (as in Definition~\ref{def:AB+-CS}) would be equivalent in the scenario considered here to OG's ECS (Definition~\ref{def:OG-ECS}), this implies that although $W^\text{act.}$ is CS according to OG's Definition (see above), it is not ECS. 
This provides an explicit example of ``activation of causal nonseparability'' in that scenario.

\medskip

Hence, OG's CS does not reduce (contrary to OG's ECS) to Ara\'ujo \emph{et al.}'s definition of causal separability in this particular scenario. 
Definitions~\ref{def:AB+-CS} and~\ref{def:OG-CS} of causal separability are therefore inconsistent. 
Our aim now is to rectify this inconsistency.

\subsection{Our choice of definition}

To fix this, we now propose our own definition of multipartite causal separability, which indeed resolves the inconsistency pointed out above, and which we argue is a more natural definition for general multipartite scenarios.  
Similarly to OG, we choose a recursive definition, based on the concept of a conditional process matrix and very much in the spirit of the recursive definitions that have been given for multipartite causal correlations~\cite{abbott16,oreshkov16}. 
For a process matrix to be compatible with a definite causal order, there should, in any run of the experiment, be a designated party that acts first (which party this is can be determined probabilistically, just like in the bipartite case) and the conditional process matrix for the remaining parties, which depends on the action of the first party, should again be causally separable for any CP map that the first party applies.

For several reasons, we consider it important to allow extensions with extra incoming systems, similar to OG's \emph{extensible} causal separability.
Firstly, the whole process matrix framework is constructed so as to allow for shared ancillary systems between the parties. 
For consistency, we should thus take into account such extensions with shared incoming quantum states when defining causal (non)separability.
Indeed, entanglement is a very different resource from causal nonseparability: entangled systems do not by themselves allow signalling between parties, and should be able to be distributed between parties prior to an experiment without ``activating'' causal nonseparability.
(Note, however, that entanglement can still play a crucial role in causal nonseparability, as e.g., in the quantum switch, where the control and target systems can end up being entangled after the parties' operations.)
While a ``resource theory'' for causal nonseparability has not yet been developed, it is reasonable to expect that providing additional shared (entangled) incoming states should be a free operation in such an approach.
These considerations lead us to propose the following definition.

\begin{definition}[$N$-partite causal separability] \label{def:our_def-CS}
For $N{=}1$, any process matrix is causally separable. 
For $N \ge 2$, an $N$-partite process matrix $W$ is said to be \emph{causally separable} if and only if, for any extension $A_{I'}^\N$ of the parties' incoming systems and any ancillary quantum state $\rho \in A_{I'}^\N$, $W \otimes \rho$ can be decomposed as 
\begin{equation}
 W \otimes \rho = \sum_{k \in \N} q_k \, W_{(k)}^\rho , \label{eq:our_def-CS}
\end{equation}
with $q_k\ge 0$, $\sum_k q_k = 1$, and where for each $k$, $W_{(k)}^\rho \in A_{II'O}^\N$ is a process matrix compatible with party $A_k$ acting first, and is such that for any CP map $M_k \in A_{II'O}^k$ applied by party $A_k$, the conditional $(N{-}1)$-partite process matrix%
\footnote{Note that compared to Eq.~\eqref{eq:conditional_W}, we take here $A_{II'O}^k \coloneqq A_{IO}^k \otimes A_{I'}^k$, $M_k \coloneqq M_{a_k|x_k}^{A_{II'O}^k}$, $\Tr_k \coloneqq \Tr_{A_{II'O}^k}$ and $\id^{\N \backslash k} \coloneqq \bigotimes_{j \in \N \backslash k} \id^{A_{II'O}^j}$ in the definition of the conditional matrix.}
$(W_{(k)}^\rho)_{|M_k} \coloneqq \Tr_k [M_k \otimes \id^{\N \backslash k} \, \cdot \, W_{(k)}^\rho]$ is itself causally separable.
\end{definition}

Note that there is a subtle difference between our definition here and that of OG's ECS (Definition~\ref{def:OG-ECS}).
We indeed require all conditional process matrices appearing at all levels of the recursive decomposition to remain causally separable under extension with arbitrary ancillary states, while OG impose this \emph{a priori} only for the original process matrix. 
In fact, although \emph{prima facie} different, these definitions turn out to be equivalent; the proof of this is given in Appendix~\ref{app:equiv_OGECS_our_def}.

From Definition~\ref{def:our_def-CS} we recover the natural, intuitive definition of Ara\'ujo \emph{et al.}~\cite{araujo15} in the particular tripartite case where one party has a trivial outgoing system---a case of practical relevance, as the quantum switch is the first example of a causally nonseparable process that has been demonstrated and studied in laboratory experiments~\cite{procopio15,rubino17,goswami18}. 
One can also readily verify that process matrices that are causally separable by Definition~\ref{def:our_def-CS} cannot generate noncausal correlations (as defined in Refs.~\cite{oreshkov16,abbott16}); an explicit proof is given in Appendix~\ref{app:causal_correl}.  

From now on, whenever we talk about causal (non)separability we will refer to our Definition~\ref{def:our_def-CS}.

\section{Characterising multipartite causal (non)separability}
\label{sec:charact}

With the definition of causal (non)separability given above, we now turn to addressing the question of how to characterise causally separable process matrices in terms of simple conditions and how to demonstrate multipartite causal nonseparability in practice. 

For that we will start by reviewing the characterisations of valid process matrices and of process matrices compatible with fixed causal orders, before recalling the characterisations of causally separable process matrices in the bipartite and tripartite cases, where we will give conditions for causal separability that are both necessary and sufficient.
We will then present a generalisation to the $N$-partite case which, for $N\ge 4$, gives two conditions, one necessary and one sufficient, whose coincidence remains an open question.

In this section we will not concern ourselves with the normalisation of process matrices (which can always be imposed later).
Our characterisations will then be given in terms of linear subspaces of matrices (e.g., the spaces $\L^\N$ and $\mathcal{L}^{A_1\prec\cdots\prec A_N}$ introduced already in Sec.~\ref{sec:W_formalism});
when adding the requirement of positive semidefiniteness, the corresponding sets of (nonnormalised) process matrices will thus be closed convex cones of positive semidefinite matrices.
This will allow the conditions we give to be checked efficiently with semidefinite programming (SDP) techniques.
In particular, by generalising the techniques used for the bipartite and restricted tripartite cases in Refs.~\cite{araujo15,branciard16a}, we will extend the idea of witnesses of causal nonseparability to the multipartite case and show how multipartite witnesses can be constructed efficiently, allowing this causal nonseparability to be verified experimentally by having each party perform appropriately chosen measurements~\cite{rubino17,goswami18}.

Following Ref.~\cite{araujo15}, we adopt the following notation, which will be used heavily throughout the rest of the paper:
\begin{equation}\label{eq:trade_and_pad}
	\begin{gathered}
		{}_{X}W \coloneqq (\Tr_X W) \otimes \frac{\id^X}{d_X}\,, \quad {}_1W \coloneqq W, \\
		{}_{[\sum_X \alpha_X X]}W \coloneqq \sum_X \alpha_X \, {}_{X\!} W,
	\end{gathered}
\end{equation}
with $d_X$ the dimension of the Hilbert space of system $X$ (note that $W \to {}_{X}W$ defines a CPTP map).
In particular, constraints of the form ${}_{[1-X]}W = 0$ (which will appear regularly) therefore mean that $W$ is of the form $W = \Omega \otimes \frac{\id^X}{d_X}$ (with $\Omega = \Tr_X W$).

\subsection{Valid process matrices and compatibility with a fixed causal order}
\label{subsec:valid_W_and_Afirst}

Recall from Sec.~\ref{sec:W_formalism} that the conditions for a process matrix $W$ to be valid arise from requiring that the generalised Born rule~\eqref{eq:born_rule_Npartite} should give valid probability distributions, even when the parties share arbitrary ancillary systems.
The fact that these probabilities should be nonnegative imposes that $W$ must be positive semidefinite, while the requirement that these probabilities must sum to 1 implies that any valid (but, once again, not necessarily normalised) $W$ must be in a subspace $\L^\N$ of $A_{IO}^\N$~\cite{oreshkov12,araujo15}. 
In Appendix~\ref{app:charact_valid_Ws} we recall the proof (following Ref.~\cite{araujo15}) that this subspace can be characterised as follows:
\begin{align}
& W \in \L^\N \Leftrightarrow \, \forall \ \X \subsetneq \N, \X \neq \emptyset, \ \Tr_{\N \backslash \X}  W \in \L^\X \notag \\
& \hspace{20mm} \text{ and } {}_{ \prod_{i \in \N}[1-A_O^i]} W = 0 \label{eq:constr_LN_rec} \\[1mm]
& \hspace{3mm} \Leftrightarrow \, \forall \ \X \subseteq \N, \X \neq \emptyset, \ {}_{ \prod_{i \in \X}[1-A_O^i]A_{IO}^{\N \backslash \X} } W = 0 \,. \label{eq:constr_LN}
\end{align}
Written in the form of Eq.~\eqref{eq:constr_LN_rec}, the validity constraint for $W$ says that all reduced matrices $\Tr_{\N \backslash \X}W$ shared by the parties of any strict subset $\X$ of $\N$ (obtained after tracing out the parties that are not in $\X$) must be valid, and that $W$ must further satisfy the additional constraint that ${}_{ \prod_{i \in \N}[1-A_O^i]} W = 0$. 
The form of Eq.~\eqref{eq:constr_LN} expresses explicitly all the (linearly independent) constraints that these recursive validity conditions imply on $W$.%
\footnote{Note that the constraint in Eq.~\eqref{eq:constr_LN} can also be written as ${}_{\prod_{i \in \X}[1-A_O^i]} (\Tr_{\N \backslash \X}W) = 0$. 
In this paper we generically use the form of Eq.~\eqref{eq:constr_LN} for ease of notation; it may be useful, however, to keep in mind that this type of constraint is in fact a constraint on the reduced matrix $\Tr_{\N \backslash \X}W$ shared by the parties in $\X$, as written more explicitly in Eq.~\eqref{eq:constr_LN_rec}.}
Denoting by $\P$ the convex cone of positive semidefinite matrices, the set of valid process matrices is then the convex cone
\begin{equation}
\W = {\cal P} \cap \L^\N \,.
\end{equation}

\medskip

In order to discuss the causal separability of process matrices, it is necessary to also characterise the subspaces of such matrices that are compatible with certain fixed causal relations between (subsets of) parties.
Such causal relations, as for the particular cases of fixed causal orders discussed in the previous sections, are understood via the notion of signalling: if a (group of) parties is in the causal future of some others, then there is no way for them to signal to those earlier parties.

We first consider the case of process matrices that are compatible with a given party $A_k$ acting first:%
\footnote{Note that a process matrix can be compatible with several different causal relations between parties. For example, if a matrix $W$ does not allow any party to signal to another, then it is compatible with any party or group of parties acting first.}
regardless of the operation performed by the other parties $A_{k'}$ (for all $k' \neq k$), the marginal probability distribution for $A_k$ obtained from~\eqref{eq:born_rule_Npartite} must not depend on the CPTP maps $M_{x_{k'}}^{A^{k'}_{IO}}$ chosen by those other parties.
As already mentioned in the previous section and shown in Appendix~\ref{app:fixed_order}, a given process matrix $W$ satisfies this condition if and only if, whatever CP map $M_k$ is applied by $A_k$, the conditional process matrix $W_{|M_k}$, as defined in Eq.~\eqref{eq:conditional_W}, is a valid $(N{-}1)$-partite process matrix for the remaining parties in $\N\backslash A_k$.

We can in fact ignore here the assumption that $M_k \ge 0$, and the above constraint is equivalent to imposing that $W_{|M_k} \in \L^{\N \backslash A_k}$ for any $M_k \in A_{IO}^k$.
Such a constraint defines a linear subspace of $A_{IO}^\N$.
Taking its intersection with the subspace $\L^\N$, we denote the linear subspace of valid process matrices compatible with party $A_k$ first by $\L^{A_k \prec (\N \backslash A_k)}$. 
We find, using Eq.~\eqref{eq:constr_LN} above (and after removing redundant constraints; see Eq.~\eqref{eq:constr_L_k_first_app} in Appendix~\ref{app:fixed_order}):
\begin{align}\label{eq:constr_L_k_first}
& W \in \L^{A_k \prec (\N \backslash A_k)} \notag \\[1mm]
& \Leftrightarrow \ W \in \L^\N \quad \text{and} \quad \forall \ M_k \in A_{IO}^k, \ W_{|M_k} \in \L^{\N \backslash A_k} \notag \\[1mm]
& \Leftrightarrow \ {}_{[1-A_O^k]A_{IO}^{\N \backslash k}} W = 0 \quad \text{and}  \notag \\
& \qquad \forall \ \X \subseteq \N \backslash k, \X \neq 0, {}_{\prod_{i \in \X}[1-A_O^i]A_{IO}^{\N \backslash k \backslash \X} } W = 0.
\end{align}

In Appendix~\ref{app:fixed_order} we also derive constraints for more general causal orders of the form $\K_1 \prec \K_2 \prec \cdots \prec \K_K$, for various disjoint subsets $\K_i$ of $\N$. 
Of particular interest is the specific case in which each $\K_i$ is a singleton, which gives constraints on a process matrix $W$ being compatible with a fixed causal order such as $A_1 \prec A_2 \prec \cdots \prec A_N$.
Such a $W$ must be compatible with $A_1$ acting first (and must therefore satisfy Eq.~\eqref{eq:constr_L_k_first} for $k=1$---in particular, the constraint on its third line); 
then, whatever CP map $M_1$ party $A_1$ applies, the resulting conditional process matrix $W_{|M_1}$ must then be a valid $(N{-}1)$-partite process matrix, compatible with party $A_2$ acting first (and must therefore satisfy Eq.~\eqref{eq:constr_L_k_first} for $k=2$---in particular, its third line---with $\N$ replaced by $\N \backslash \{1\}$); etc.
By iterating this argument (up until the party $A_N$), we find that the linear subspace $\L^{A_1 \prec \cdots \prec A_N}$ of process matrices compatible with the causal order $A_1 \prec \cdots \prec A_N$ is characterised by (cf.\ Eq.~\eqref{eq:constr_fixed_order_N})~\cite{gutoski06,chiribella09,araujo15}
\begin{align}
& W \in \L^{A_1 \prec \cdots \prec A_N} \notag \\[1mm]
& \Leftrightarrow \ \forall \, k = 1,\ldots,N, \ {}_{[1-A_O^k] A_{IO}^{(>k)}} W = 0 \,, \label{eq:constr_order_A1__AN}
\end{align}
with $A_{IO}^{(>k)} = A_{IO}^{\{k+1,\ldots,N\}}$ (with $A^{(>N)}_{IO}=A^\emptyset_{IO}=1$).

\subsection{Bipartite and tripartite causally (non)separable process matrices}
\label{subsec:charact_2_3}

In the bipartite scenario, the above characterisation of the subspaces $\L^{A\prec B}$ and $\L^{B\prec A}$ allows us, from Definition~\ref{def:biCS}, to give the following explicit characterisation of causally separable process matrices.

\begin{proposition}[Characterisation of bipartite causally separable process matrices] \label{prop:CS_charact_2}
A matrix $W \in A_{IO} \otimes B_{IO}$ is a valid bipartite causally separable process matrix if and only if it can be decomposed as
\begin{equation}
W = W_{(A,B)} + W_{(B,A)} \label{eq:decomp_CS_W_2}
\end{equation}
where, for each permutation $(X,Y)$ of the two parties $A$ and $B$, $W_{(X,Y)}$ is a positive semidefinite matrix satisfying
\begin{equation}
{}_{[1-X_O]Y_{IO}} W_{(X,Y)} = 0 , \ {}_{[1-Y_O]}W_{(X,Y)} = 0 \label{eq:decomp_CS_W_2_constr}
\end{equation}
(i.e., $W_{(X,Y)}$ is a valid process matrix compatible with the causal order $X \prec Y$).
\end{proposition}

Note that, in contrast to Eq.~\eqref{def:caus_sep_bi} in Definition~\ref{def:biCS}, we did not write the weights $q$ and $1-q$ explicitly in Eq.~\eqref{eq:decomp_CS_W_2}.
Instead, for convenience and consistency with the characterisations of tripartite and $N$-partite causally separable processes which will follow, we decomposed $W$ in terms of nonnormalised process matrices, writing $W_{(A,B)} = q \, W^{A \prec B}$ and $W_{(B,A)} = (1{-}q) \, W^{B \prec A}$.

\medskip

As we discussed in Sec.~\ref{sec:defs}, the tripartite case of causal separability was already studied by Oreshkov and Giarmatzi under the name ``extensible causal separability'' in Ref.~\cite{oreshkov16}.
In their Proposition~3.3 they provided a characterisation of tripartite (extensible) causal separability, albeit describing the constraints in a different way.
In our approach, this characterisation can be expressed as follows:

\begin{widetext}
\begin{proposition}[Characterisation of tripartite causally separable process matrices] \label{prop:CS_charact_3}
A matrix $W \in A_{IO} \otimes B_{IO} \otimes C_{IO}$ is a valid tripartite causally separable process matrix (as per Definition~\ref{def:our_def-CS}) if and only if it can be decomposed as
\begin{equation}
\begin{array}{rcccccc}
W &=& W_{(A)} &+& W_{(B)} &+& W_{(C)} \\[1mm]
&=& \overbrace{W_{(A,B,C)} + W_{(A,C,B)}} &+& \overbrace{W_{(B,A,C)} + W_{(B,C,A)}} &+& \overbrace{W_{(C,A,B)} + W_{(C,B,A)}}
\end{array} \label{eq:decomp_CS_W_3}
\end{equation}
where, for each permutation of the three parties $(X,Y,Z)$, $W_{(X,Y,Z)}$ and $W_{(X)} \coloneqq W_{(X,Y,Z)} + W_{(X,Z,Y)}$ are positive semidefinite matrices satisfying
\begin{gather}
 {}_{[1-X_O]Y_{IO}Z_{IO}} W_{(X)} = 0 \,, \label{eq:decomp_CS_W_3_constrWX} \\[1mm]
 {}_{[1-Y_O]Z_{IO}} W_{(X,Y,Z)} = 0 \,, \quad {}_{[1-Z_O]}W_{(X,Y,Z)} = 0 \,. \label{eq:decomp_CS_W_3_constrWXYZ}
\end{gather}
\end{proposition}
\end{widetext}

The proof of this characterisation was sketched in Ref.~\cite{oreshkov16} using a somewhat different terminology to what we employ; in particular, they express causal constraints in terms of restrictions of what terms are ``allowed'' in a Hilbert-Schmidt basis decomposition of a matrix (see Appendix~\ref{app:allowed_forbidden_terms}).
We give a more detailed proof in Appendix~\ref{app:charact_CS_3_general}, which is again based on a ``teleportation technique'' (cf.\ Lemma~\ref{lemma:teleportation} in Appendix~\ref{app:characterisation_CS_Ws}), similar in spirit to the one briefly sketched in Sec.~\ref{subsec:comparison}. 

Let us break down and analyse the terms appearing in the decomposition~\eqref{eq:decomp_CS_W_3} to understand better this characterisation.

From the constraints in Eq.~\eqref{eq:decomp_CS_W_3_constrWXYZ} it follows that, that for each party $X$, the matrix $W_{(X)} ( = W_{(X,Y,Z)} + W_{(X,Z,Y)})$ satisfies ${}_{[1-Y_O]Z_{IO}}W_{(X)} = {}_{[1-Z_O]Y_{IO}}W_{(X)} = {}_{[1-Y_O][1-Z_O]}W_{(X)} = 0$. 
Together with Eq.~\eqref{eq:decomp_CS_W_3_constrWX} and the fact that $W_{(X)}$ is positive semidefinite, this implies that $W_{(X)}$ is a valid tripartite process matrix compatible with party $X$ acting first (since it satisfies Eq.~\eqref{eq:constr_L_k_first} for $A_k=X$). 
$W$ is thus decomposed in Eq.~\eqref{eq:decomp_CS_W_3} as a sum of 3 valid process matrices, which ensures in particular that it is itself a valid process matrix.

On the other hand, the matrices $W_{(X,Y,Z)}$ in the decomposition~\eqref{eq:decomp_CS_W_3} are not necessarily valid process matrices. 
Nevertheless, the constraints~\eqref{eq:decomp_CS_W_3_constrWXYZ} imply that whatever the CP map $M_X$ applied by the first party $X$, the conditional process matrix $(W_{(X,Y,Z)})_{|M_X} \coloneqq \Tr_X[M_X \otimes \id^{YZ} \cdot W_{(X,Y,Z)}]$ is a valid bipartite process matrix, compatible with the causal order $Y \prec Z$ (indeed, it satisfies Eq.~\eqref{eq:constr_order_A1__AN} for this causal order: 
e.g., ${}_{[1-Y_O]Z_{IO}}[(W_{(X,Y,Z)})_{|M_X}] = ({}_{[1-Y_O]Z_{IO}}W_{(X,Y,Z)})_{|M_X} = 0$).

The fact that the matrices $W_{(X,Y,Z)}$ are not necessarily valid process matrices, and thus that Eq.~\eqref{eq:decomp_CS_W_3} does not simply decompose $W$ into a combination of process matrices compatible with fixed causal orders, is a consequence of the possibility of dynamical (but still well-defined, albeit not \emph{fixed}) causal orders (recall the discussion at the end of Sec.~\ref{sec:W_def_multipartite}).
In Sec.~\ref{subsec:examples} we will consider in more detail a concrete example of a process matrix allowing for such dynamical causal orders.

\subsection{General multipartite causally (non)separable process matrices}
\label{subsec:charact_N}

As we will see below, it is possible to generalise the decomposition of Proposition~\ref{prop:CS_charact_3} to the case of $N$-partite causal separability. While the generalisation clearly provides a sufficient condition for causal separability, it turns out that the proof that it is also a necessary condition does not readily generalise. 
Indeed, the proof for the tripartite case in Appendix~\ref{app:charact_CS_3_general} relies on the fact that each term $W_{(X)}$ in Eq.~\eqref{eq:decomp_CS_W_3} is the sum of only two ``base'' terms, something that is not true in the natural generalisation of this decomposition.
(To understand this better, we encourage the interested reader to look at the subtleties of that proof.)

For the general multipartite case, we therefore provide the following, separate, necessary and sufficient conditions.
Since these arise from different considerations, we will present and discuss these individually. Indeed, although these coincide in the bipartite and tripartite cases, it remains an open question whether this is the case in general (or if one is both necessary and sufficient but not the other, or if neither are).

\subsubsection{Necessary condition}

The necessary condition we present here is based on the teleportation technique and is a generalisation of the use of this approach in the proof of the tripartite characterisation.
The teleportation technique is more formally described in Lemma~\ref{lemma:teleportation} in Appendix~\ref{app:characterisation_CS_Ws}, but we briefly outline how it leads to the necessary condition to help understand the condition itself.
The idea is to consider, in Eq.~\eqref{eq:our_def-CS} of Definition~\ref{def:our_def-CS}, a specific shared incoming ancillary state, as well as specific operations $M_k$ applied by the parties $A_k$, for which there is a straightforward relation between the forms of the respective $N$-partite process matrices in which $A_k$ acts first, and the corresponding $(N{-}1)$-partite conditional process matrices that we obtain after $A_k$ has operated. As the latter are by definition causally separable (and satisfy thus the necessary conditions for $(N{-}1)$-partite causal separability), this allows us to infer necessary conditions for the causal separability of the original $N$-partite process matrix. 

More precisely, we provide, as ancillary incoming systems, a maximally entangled state between every pair of parties, defining an overall ancillary state $\rho$.
If $W$ is a causally separable process matrix, then, by definition, $W \otimes \rho$ can be decomposed into a sum of process matrices $W_{(k)}^\rho$ compatible with a given party $A_k$ acting first (cf.\ Eq.~\eqref{eq:our_def-CS} in Definition~\ref{def:our_def-CS}); furthermore, as $\rho$ is pure, one can write $W_{(k)}^\rho = W_{(k)} \otimes \rho$ with $W_{(k)}$ itself being compatible with $A_k$ first.
For each such process matrix $W_{(k)}$ the party $A_k$ can then ``teleport'' the part of $W_{(k)}$ on their systems $A^k_{IO}$ to another party $A_{k'}$ by applying an appropriate CP map $M_k$.
The effect is that the resulting $(N{-}1)$-partite conditional process matrix $(W_{(k)}^\rho)_{|M_k}$ formally has the same form as $W_{(k)}$ (tensored with what is left over of the, now reduced, ancillary state $\rho$), 
except that the systems $A_{IO}^k$ are instead attributed (``teleported'') to the ancillary incoming system $A_{I'}^{k'}$ of $A_{k'}$.
From the definition of causal separability, $(W_{(k)}^\rho)_{|M_k}$ must itself be causally separable, so the necessary condition can be recursively applied to this $(N{-}1)$-partite process matrix until the base case of $N=3$, given by Proposition~\ref{prop:CS_charact_3}, is reached.

We give the full details of the proof of the necessary condition in Appendix~\ref{app:charact_CS_N_necessary}. 
However, in order to state more formally the condition itself, let us introduce the following notation.
For a given matrix $W \in A_{IO}^\N$, we denote by $W^{A_{IO}^k \to A_{I'}^{k'}} \in A_{IO}^{\N \backslash k} \otimes A_{I'}^{k'}$ the same matrix, where the systems $A_{IO}^k$ are attributed to some other system $A_{I'}^{k'}$ (of the same dimension as $A_{IO}^k$).
More formally, 
\begin{equation}
W^{A_{IO}^k \to A_{I'}^{k'}} \!\coloneqq \sum_{i,j} \Tr_k \! \big[ \!\ket{i}\!\!\bra{j}^{A_{IO}^k} \!\otimes \id^{\N \backslash k} \cdot W \big] \!\otimes \ket{j}\!\!\bra{i}^{A_{I'}^{k'}} \!\!, \label{eq:W_A_telep_to_A'}
\end{equation}
where $\{\ket{i}\}$ is an orthonormal basis of $\HS^{A_I^k} \otimes \HS^{A_O^k}$.

We then obtain the following recursive necessary condition:

\begin{proposition}[Necessary condition for general multipartite causal separability] \label{prop:CS_necessary}
An $N$-partite causally separable process matrix $W \in A_{IO}^\N$ (as per Definition~\ref{def:our_def-CS}) must necessarily have a decomposition of the form
\begin{equation}\label{eq:decomp_CN_W}
	W = \sum_{k \in \N} W_{(k)} 
\end{equation}
where each $W_{(k)}$ is a valid process matrix compatible with party $A_k$ acting first, and such that for each $k' \neq k$, $W_{(k)}^{A_{IO}^k \to A_{I'}^{k'}}$ is an $(N{-}1)$-partite causally separable process matrix.

Hence, any constraints satisfied by $(N{-}1)$-partite causally separable process matrices must also be satisfied by $W_{(k)}$ after re-attributing the system $A_{I'}^{k'}$ back to $A_{IO}^k$---i.e., after formally replacing $A_I^{k'}$ by $A_I^{k'}A_{I'}^{k'}$ and then $A_{I'}^{k'}$ by $A_{IO}^k$ in the constraints written using the notation defined in Eq.~\eqref{eq:trade_and_pad}.
\end{proposition}

The decomposition of Eq.~\eqref{eq:decomp_CN_W} follows from that of Eq.~\eqref{eq:our_def-CS} in our definition of causal separability, for the appropriate choice of ancillary state and CP maps, as described above (see Appendix~\ref{app:charact_CS_N_necessary}).

To further clarify this condition, let us illustrate, in the fourpartite case (with parties $A,B,C,D$), how one can use it to obtain explicit constraints on causally separable process matrices. 
Proposition~\ref{prop:CS_necessary} implies that a fourpartite causally separable process matrix $W$ must be decomposable as
\begin{equation} \label{eq:decomp_CS_WA_4}
W = W_{(A)} + W_{(B)} + W_{(C)} + W_{(D)}, 
\end{equation}
with each $W_{(X)}$ (for $X = A,B,C,D$) being a valid process matrix compatible with party $X$ acting first---hence satisfying Eq.~\eqref{eq:constr_L_k_first} for $A_k=X$.%
\footnote{Note that the existence, for all $Y$, of a decomposition of the form of Eq.~\eqref{eq:explicit_NC_4partite_2} satisfying Eq.~\eqref{eq:explicit_NC_4partite} implies all the constraints of Eq.~\eqref{eq:constr_L_k_first}, except for the third line (i.e., ${}_{[1-X_O]Y_{IO}Z_{IO}T_{IO}}W=0$).}
For each $X$ and every other party $Y \neq X$, the recursive constraint that $W_{(X)}^{X_{IO} \to Y_{I'}}$ is a tripartite causally separable process matrix further implies, according to Proposition~\ref{prop:CS_charact_3} (for the 3 parties $Y,Z,T \neq X$) and after re-attributing the system $Y_{I'}$ to $X_{IO}$ (i.e., replacing $Y_{IO}$ by $Y_{I'}Y_{IO}$ and then $Y_{I'}$ by $X_{IO}$ in the constraints), that there must exist a decomposition of $W_{(X)}$ of the form%
\footnote{Here the superscripts $[X \to Y]$ are simply labels to indicate that, for each matrix $W_{(X)}$, there are potentially different decompositions of the form~\eqref{eq:explicit_NC_4partite_2} for each $Y \neq X$. (The sufficient condition below will in fact precisely be obtained by assuming that these decompositions do not depend on $Y$.)}
\begin{align}\label{eq:explicit_NC_4partite_2}
	W_{(X)} =&\, W_{(X,Y)}^{^{\scriptscriptstyle [X \to Y]}} + W_{(X,Z)}^{^{\scriptscriptstyle [X \to Y]}} + W_{(X,T)}^{^{\scriptscriptstyle [X \to Y]}} \notag \\
	 =&\, W_{(X,Y,Z,T)}^{^{\scriptscriptstyle [X \to Y]}} + W_{(X,Y,T,Z)}^{^{\scriptscriptstyle [X \to Y]}} + W_{(X,Z,Y,T)}^{^{\scriptscriptstyle [X \to Y]}} \notag \\
	& + W_{(X,Z,T,Y)}^{^{\scriptscriptstyle [X \to Y]}} + W_{(X,T,Y,Z)}^{^{\scriptscriptstyle [X \to Y]}} + W_{(X,T,Z,Y)}^{^{\scriptscriptstyle [X \to Y]}}
\end{align}
where each term appearing in the decomposition is positive semidefinite, $W_{(X,Y)}^{^{\scriptscriptstyle [X \to Y]}} = W_{(X,Y,Z,T)}^{^{\scriptscriptstyle [X \to Y]}} + W_{(X,Y,T,Z)}^{^{\scriptscriptstyle [X \to Y]}}$, etc., and with (for all $X \neq Y \neq Z \neq T$)
\begin{align}\label{eq:explicit_NC_4partite}
	& {}_{[1-Y_O]Z_{IO}T_{IO}}W_{(X,Y)}^{^{\scriptscriptstyle [X \to Y]}} = {}_{[1-Z_O]X_{IO}Y_{IO}T_{IO}}W_{(X,Z)}^{^{\scriptscriptstyle [X \to Y]}} = 0, \notag \\
	& {}_{[1-Z_O]T_{IO}}W_{(X,Y,Z,T)}^{^{\scriptscriptstyle [X \to Y]}} = {}_{[1-T_O]}W_{(X,Y,Z,T)}^{^{\scriptscriptstyle [X \to Y]}} = 0, \notag \\
	& {}_{[1-Y_O]T_{IO}}W_{(X,Z,Y,T)}^{^{\scriptscriptstyle [X \to Y]}} = {}_{[1-T_O]}W_{(X,Z,Y,T)}^{^{\scriptscriptstyle [X \to Y]}} = 0, \notag \\
	& {}_{[1-T_O]X_{IO}Y_{IO}}W_{(X,Z,T,Y)}^{^{\scriptscriptstyle [X \to Y]}} = {}_{[1-Y_O]}W_{(X,Z,T,Y)}^{^{\scriptscriptstyle [X \to Y]}} = 0. 
\end{align}

\medskip

Finally, we remark that the constraints obtained by considering teleporting each party $X$'s system to just a single other party $Y$ (i.e., by just demanding the existence of a decomposition of the above form for some other party $Y$, rather than for \emph{all} other parties $Y \neq X$) yields conditions that are still necessary for the causal separability of $W$, but which are generally weaker than those given in Proposition~\ref{prop:CS_necessary}.
Indeed, in Appendix~\ref{app:gap_nec_suff} we give an example of a fourpartite process matrix which satisfies those weaker conditions but not all of those given above.

\subsubsection{Sufficient condition}

A sufficient condition for causal separability can be obtained by considering a stricter form of the recursive decomposition~\eqref{eq:decomp_CN_W} in Proposition~\ref{prop:CS_necessary}.
In particular, we demand that $W$ has a decomposition into $W_{(k)}$ compatible with $A_k$ acting first and such that each $W_{(k)}$ itself recursively satisfies the sufficient constraints for an $(N{-}1)$-partite process matrix without $A^k_{IO}$ being traced out.
One can easily verify that the decomposition~\eqref{eq:decomp_CS_W_3} in the tripartite case is a generalisation of this kind from the bipartite case.
In the fourpartite case described explicitly above, this means that for each party $X$ there should be a single decomposition of the form~\eqref{eq:explicit_NC_4partite_2} (i.e., no longer dependent on $Y$) such that the constraints~\eqref{eq:explicit_NC_4partite} are satisfied without tracing out $X_{IO}$ on the first and fourth lines.
The fact that, unlike in the necessary conditions, we only consider a single (recursive) decomposition of each $W_{(k)}$ means that we can give a more explicit formulation for the sufficient condition.

Before stating the sufficient condition, let us introduce some more notations.
Let $\Pi$ denote the set of permutations (generically denoted by $\pi$) of $\N$. 
For an ordered subset $(k_1,\ldots,k_n)$ of $\N$ with $n$ elements (with $1 \le n \le N$, $k_i \neq k_j$ for $i \neq j$), let $\Pi_{(k_1,\ldots,k_n)}$ be the set of permutations of $\N$ for which the element $k_1$ is first, $k_2$ is second, $\ldots$, and $k_n$ is $n^{\text{th}}$---i.e., $\Pi_{(k_1,\ldots,k_n)} = \{\pi \in \Pi \mid \pi(1) = k_1, \ldots, \pi(n) = k_n\}$.
With these notations, we have the following sufficient condition, that directly generalises the decomposition of Proposition~\ref{prop:CS_charact_3}.

\begin{proposition}[Sufficient condition for general multipartite causal separability] \label{prop:CS_sufficient}
If a matrix $W \in A_{IO}^\N$ can be decomposed as a sum of $N!$ positive semidefinite operators $W_\pi \ge 0$ in the form
\begin{equation}
W = \sum_{\pi \in \Pi} W_\pi, \label{eq:decomp_CS_W}
\end{equation}
such that for any ordered subset of parties $(k_1, \ldots, k_n)$ of $\N$ (with $1 \le n \le N$, $k_i \neq k_j$ for $i \neq j$), the partial sum
\begin{equation}
W_{(k_1, \ldots, k_n)} \coloneqq \sum_{\pi \in \Pi_{(k_1, \ldots, k_n)}} W_\pi \label{eq:decomp_CS_WK}
\end{equation}
satisfies
\begin{equation}
{}_{[1-A_O^{k_n}]A_{IO}^{\N \backslash \{k_1,\ldots,k_n\}}}W_{(k_1, \ldots, k_n)} = 0, \label{eq:decomp_CS_suff_cond}
\end{equation}
then $W$ is a valid causally separable process matrix (as per Definition~\ref{def:our_def-CS}).
\end{proposition}

This decomposition was also suggested independently by Oreshkov as a possible generalisation of Proposition~\ref{prop:CS_charact_3}~\cite{oreshkov16a} (although following the approach of Refs.~\cite{oreshkov12,oreshkov16}, Oreshkov expressed it differently, namely in terms of allowed terms in a Hilbert-Schmidt basis decomposition of the matrices $W_{(k_1, \ldots, k_n)}$; cf.\ Appendix~\ref{app:allowed_forbidden_terms}).
The proof that the condition above is indeed sufficient is given in Appendix~\ref{app:charact_CS_N_sufficient}. 
In order to understand it better, it is nonetheless worth discussing the form of the decomposition and the terms appearing within in a little more detail.

Firstly, note that one can easily show by induction (see Appendix~\ref{app:charact_CS_N_sufficient}), that if Eq.~\eqref{eq:decomp_CS_suff_cond} is satisfied for all $(k_1, \ldots, k_n)$, then one also has, for all $(k_1, \ldots, k_n)$ with $1 \le n < N$, that
\begin{align}
& \forall \ \X \subseteq \N \backslash \{k_1, \ldots, k_n\}, \X \neq \emptyset, \notag \\
& {}_{\prod_{i \in \X}[1-A_O^i]A_{IO}^{\N \backslash \{k_1, \ldots, k_n\} \backslash \X}}W_{(k_1, \ldots, k_n)} = 0. \label{eq:decomp_CS_suff_cond_2}
\end{align}
Note also that since all $W_\pi \ge 0$, all $W_{(k_1, \ldots, k_n)} \ge 0$ as well.

For $n=1$, Eqs.~\eqref{eq:decomp_CS_suff_cond} and~\eqref{eq:decomp_CS_suff_cond_2} imply that each matrix $W_{(k_1)} \, (\ge 0)$ is a valid process matrix compatible with party $A_{k_1}$ acting first; indeed, Eq.~\eqref{eq:constr_L_k_first} is satisfied for $A_k=A_{k_1}$. 
As $W = \sum_{k_1} W_{(k_1)}$ according to Eqs.~\eqref{eq:decomp_CS_W}--\eqref{eq:decomp_CS_WK}, this ensures in particular that $W$ is indeed a valid process matrix.

Note, however, that in general the matrices $W_{(k_1, \ldots, k_n)}$ for $n>1$ are \emph{not} valid processes matrices compatible with the causal order $A_{k_1} \prec \cdots \prec A_{k_n}$.
Indeed, as we already observed in the tripartite case, $W_{(k_1, \ldots, k_n)}$ may not generally be a valid process matrix at all.
Nevertheless, comparing with Eq.~\eqref{eq:constr_L_k_first}, one can see that Eqs.~\eqref{eq:decomp_CS_suff_cond} and~\eqref{eq:decomp_CS_suff_cond_2} imply that whatever the CP maps $M_{k_1}, \ldots, M_{k_{n-1}}$ applied by the $n{-}1$ parties $A_{k_1}, \ldots, A_{k_{n-1}}$, the conditional matrix $(W_{(k_1, \ldots, k_n)})_{|M_{k_1} \otimes \cdots \otimes M_{k_{n-1}}} \coloneqq \Tr_{k_1, \ldots, k_{n-1}} [M_{k_1} \otimes \cdots \otimes M_{k_{n-1}} \otimes \id^{\N \backslash \{k_1,\ldots,k_{n-1}\}} \cdot W_{(k_1, \ldots, k_n)}]$ is a valid $(N{-}n{+}1)$-partite process matrix, compatible with party $A_{k_n}$ acting first.

\medskip

As we have noted already, the condition of Proposition~\ref{prop:CS_sufficient} coincides, in the bipartite and tripartite cases, with those given in Propositions~\ref{prop:CS_charact_2} and~\ref{prop:CS_charact_3}, respectively.
Indeed, for these cases, the necessary and sufficient conditions given here coincided.
For four-or-more parties it remains an open question whether this is also the case.
We performed several numerical searches for process matrices satisfying the necessary but not sufficient conditions (see Appendix~\ref{app:numerical_search}) and failed to find any such examples, although the complexity of the numerical searches means that we caution against interpreting this as evidence that the conditions coincide in general.
In Appendix~\ref{app:particular_4partite}, however, we show that they do coincide in the specific fourpartite case with $d_{D_O} = 1$.
This is a rather restricted scenario (where any process matrix is compatible with $D$ acting last), but nonetheless includes cases of interest such as the fourpartite variant of the quantum switch we discuss at the end of this section.

Finally, we note that the decomposition in Proposition~\ref{prop:CS_sufficient} has consequences beyond the definition of causal separability meriting additional interest: as we show elsewhere~\cite{wechs18}, it characterises precisely (i.e., providing a necessary and sufficient condition for) quantum circuits with classical control of causal order.

\subsection{Witnesses of causal nonseparability}
\label{subsec:witnesses}

While the previous characterisations provide mathematical descriptions of causally (non)separable process matrices, an important problem is the ability to detect and certify causal nonseparability in practice.
One approach that has been explored extensively is to show the violation of causal inequalities~\cite{oreshkov12,baumeler14,branciard16,abbott16,abbott17}, which is indeed only possible (within the process matrix formalism) with causally nonseparable process matrices (see Appendix~\ref{app:causal_correl}), and provides a device-independent certificate of noncausality.
However, certain causally nonseparable process matrices are known not to violate any such inequalities---this is, e.g., the case for the quantum switch~\cite{araujo15,oreshkov16}.

Another approach, first introduced in Ref.~\cite{araujo15} for the bipartite and restricted tripartite scenarios, and further studied in Ref.~\cite{branciard16a}, is to construct \emph{witnesses of causal nonseparability}---or ``\emph{causal witnesses}'' for short.
Here, we outline this approach before describing how the conditions given in the previous subsections allow us to construct causal witnesses for general multipartite scenarios. 
This will permit a full analysis of the examples in the following section, as well as the verification of certain results already claimed in previous sections.
While the overall approach of causal witnesses---and their formulation as efficiently solvable semidefinite programming (SDP) problems---in the general case mirrors that of the specific scenarios previously studied~\cite{araujo15,branciard16a}, the validity of the generalisation rests on certain technical details which we prove in Appendix~\ref{app:witnesses_SDP}.

A causal witness is defined as a Hermitian operator $S$ such that
\begin{equation}\label{eq:witnessDef}
	\Tr[S\cdot W^\sep]\ge 0
\end{equation}
for all $W^\sep\in\W^\sep$, where $\W^\sep\subset \W$ is the set of causally separable process matrices.
For any causally nonseparable $W^\text{ns}$, it is known that there exists a causal witness $S$ such that $\Tr[S\cdot W^\text{ns}]<0$~\cite{araujo15,branciard16a}.
Given a process, a causal witness $S$ can be ``measured'' by having each party implement suitably chosen operations or measurements, providing a now device-dependent test of causal nonseparability.
This approach has been used, e.g., to verify experimentally the causal nonseparability of two different implementations of the quantum switch~\cite{rubino17,goswami18}.

Propositions~\ref{prop:CS_charact_2}, \ref{prop:CS_charact_3}, \ref{prop:CS_necessary} and~\ref{prop:CS_sufficient} allow for the characterisation of the convex cone $\W^\text{sep}$ of causally separable processes---or, for the latter two propositions, outer and inner approximations $\W_+^\text{sep}$ and $\W_-^\text{sep}$ thereof---in terms of Minkowski sums and intersections of linear subspaces and of the cone of positive semidefinite operators ${\cal P}$.
The set of causal witnesses is then precisely the dual cone of $\W^\text{sep}$, $\mathcal{S} = (\W^\text{sep})^*$~\cite{araujo15,branciard16a}.
A characterisation of $\mathcal{S}$ can, in general, be obtained from the description of $\W^\text{sep}$ by using the following duality relations for any two nonempty closed convex cones $\C_1$ and $\C_2$~\cite{rockafellar70}:
\begin{equation}\label{eq:duality_relations}
(\C_1 + \C_2)^* = \C_1^* \cap \C_2^*, \quad (\C_1 \cap \C_2)^* = \C_1^* + \C_2^*
\end{equation}
(where $\C_1 + \C_2 = \{c_1+c_2 \mid c_1 \in \C_1, c_2 \in \C_2\}$ is the Minkowski sum of the two cones $\C_1$ and $\C_2$; note that all the cones we shall consider will be nonempty, closed and convex).

Since these cones are convex, the construction of causal witnesses (or of explicit decompositions of causally separable process matrices) can be efficiently performed with SDP, as first described in Ref.~\cite{araujo15}; we will follow here the slightly different approach of~\cite{branciard16a}.
The question of whether a given $W$ is causally separable can be reformulated as the optimisation problem of how much white noise can be added to a process matrix before it becomes causally separable.
Let $\id^\circ=\id^{A_{IO}^\N}/\prod_{k \in \N} d_{A_I^k}$ be the ``white noise'' process matrix (which corresponds to each each party just receiving a fully mixed state $\id^{A_I^k}/d_{A_I^k}$, and is causally separable), and consider the noisy process matrix
\begin{equation}
	W(r)=\frac{1}{1+r}(W+r \id^\circ).
\end{equation}
Since the normalisation is irrelevant to membership of $\mathcal{W}^\text{sep}$, determining whether $W$ is causally separable can be thus phrased as the SDP optimisation problem
\begin{equation}\label{eq:primalSDP}
	\begin{split}
		& \qquad\quad \min\ r\\
		& \text{s.t.}\quad W+r\id^\circ\in\mathcal{W}^\text{sep},
	\end{split}
\end{equation}
which can be efficiently solved using standard software by writing $\mathcal{W}^\text{sep}$ in terms of SDP constraints (see~\cite{branciard16a}, the examples below and Appendix~\ref{app:witnesses_SDP} for further details).
The solution to this problem, $r^*$, gives the \emph{random robustness} $\max(r^*,0)$ of $W$, and a value $r^*>0$ implies that $W$ is causally nonseparable~\cite{araujo15,branciard16a}.

Eq.~\eqref{eq:primalSDP} is known as the primal problem, and is related to the dual problem
\begin{equation}\label{eq:dualSDP}
	\begin{split}
		& \qquad\quad \min\ \Tr[S\cdot W]\\
		& \text{s.t.}\quad S\in \mathcal{S}\quad\text{and}\quad\Tr[S\cdot \id^\circ]=1,
	\end{split}
\end{equation}
defined over the dual cone $\mathcal{S}$ of $\mathcal{W}^\text{sep}$~\cite{araujo15,branciard16a}.
The optimal solution $S^*$ is a witness of the causal nonseparability of $W$ whenever $\Tr[S^*\cdot W]<0$.
The Strong Duality Theorem for SDP problems moreover relates these two problems, stating that their solutions satisfy
\begin{equation}\label{eq:rSoptReln}
	r^* = -\Tr[S^*\cdot W].
\end{equation}
In Appendix~\ref{app:witnesses_SDP} we show that~\eqref{eq:dualSDP} is indeed the dual of~\eqref{eq:primalSDP} and that the Strong Duality Theorem is indeed applicable for arbitrary scenarios, as well as giving some further details.
This implies in particular that the witness $S^*$ thus obtained is optimal when $W$ is subject to white noise, in the sense that it witnesses the causal nonseparability of all noisy process matrices $W(r)$ with $r$ sufficiently small ($r < r^*$) so as for $W(r)$ to remain causally nonseparable.

For more than 3 parties, the witnesses in the set $\mathcal{S}_+ = (\W_+^\text{sep})^*$ obtained from the cone $\W_+^\text{sep} \supseteq \W^\text{sep}$ arising from the necessary condition of Proposition~\ref{prop:CS_necessary} are also valid witnesses of $\W^\text{sep}$ since $\mathcal{S}_+\subseteq \mathcal{S}$.
On the other hand, by solving the primal SDP problem over the cone $\W_-^\text{sep}$ arising from the sufficient condition in Proposition~\ref{prop:CS_sufficient}, one can show the causal separability of any $W\in \W_-^\text{sep}\subseteq W^\text{sep}$ (through the construction of an explicit causally separable decomposition for $W$ of the form given in Proposition~\ref{prop:CS_sufficient}).
Recalling the claim that such process matrices correspond precisely to quantum circuits with classical control of orders~\cite{wechs18}, the dual cone $\mathcal{S}_-$ is thus the set of ``witnesses for no classical control of causal order'' (which can thus be found by solving the dual SDP problem). 

In Appendix~\ref{app:witnesses_SDP} we give some concrete characterisations of the cones $\mathcal{W}^\sep$ and $\mathcal{S}$ for different scenarios in terms of SDP constraints, and which are relevant for the examples that we shall now give in the following section.

\subsection{Examples}
\label{subsec:examples}

In the bipartite scenario and restricted tripartite scenario in which $C$ has no outgoing system, several examples of causally nonseparable process matrices have previously been formulated and studied in detail~\cite{oreshkov12,araujo15,oreshkov16}.
The characterisations of the cones of causal witnesses that we give in Appendix~\ref{app:witness_examples} for these bipartite and restricted tripartite scenarios (see Eqs.~\eqref{eq:Scone_N2} and~\eqref{eq:Scone_N3_no_CO}) are equivalent to those given in Refs.~\cite{araujo15,branciard16a}, and can readily be used to verify the causal nonseparability of these examples, following the approach just outlined.

As mentioned already in Sec.~\ref{subsec:araujo_def}, the quantum switch is a particularly interesting example of a causally nonseparable process matrix in the second of these scenarios.
In that same scenario we have in fact also already looked at another explicit example: the process matrix $W^\text{act.}$~\eqref{eq:def_W_activ} introduced in Sec.~\ref{subsec:comparison} to show the ``activation of causal nonseparability'' under OG's definition of causal separability.
An explicit witness from the cone~\eqref{eq:Scone_N3_no_CO} is given in Appendix~\ref{app:witness_Wactiv}, Eq.~\eqref{eq:def_S_activ}, which could thus have been equally well found with the approach of Refs.~\cite{araujo15,branciard16a}.

\medskip

Another example of ``activation of causal nonseparability'' under OG's terminology was given in Ref.~\cite{oreshkov16} in the different tripartite case in which one party, say now $A$, has only a nontrivial outgoing system, and can thus always be seen as acting first.
A witness for this example can be found by solving the dual SDP problem~\eqref{eq:dualSDP} using now the cone of witnesses~\eqref{eq:Scone_N3_no_AI} corresponding to this restricted tripartite scenario.

\medskip

Of more novel interest is the fourpartite scenario, in which causal separability has not previously been characterised.
A particularly interesting and simple example here is a fourpartite version of the quantum switch, in which a party $A$(lice) has no incoming system ($d_{A_I}=1$) and always acts first, while another party $D$(orothy) has no outgoing system ($d_{D_O}=1$) and always acts last.
Let us describe more precisely this version of the quantum switch.

The switch is composed of two qubits: a control qubit and a target qubit.
Initially, Alice prepares the control qubit in some state of her choosing (in general as a function of her input $x$). (Note that it is here that the fourpartite switch described here differs from the tripartite one, where the control qubit is in a fixed superposition.)
The target qubit, initially prepared (externally to the 4 parties) in some state $\ket{\psi}$, is then sent to Bob and Charlie, who act in an order that depends on the state of the control qubit: if it is $\ket{0}$ then Bob acts before Charlie ($B\prec C$), while if it is $\ket{1}$ then Charlie acts before Bob ($C\prec B$).
If it is in a superposition, then Bob and Charlie can instead be seen to act in a superposition of different orders.
Finally, both qubits are sent to Dorothy who can perform a measurement on them (for simplicity, we will consider that $D$ simply ignores the target qubit and thus will trace it out, as this will not change the discussion that follows).%
\footnote{We note that the quantum switch was also described as a fourpartite process in Ref.~\cite{rubino17}, with one party acting first, and one acting last. 
However, in that reference the first party was controlling the target qubit, rather than the control qubit as we consider here. 
In that case (with the first party controlling the target qubit), the random robustness is increased to $2.767$. 
One could also have here a first party that controls both the target and control qubits (as in Ref.~\cite{oreshkov18}), which further increases the tolerable white noise to $4.686$; for simplicity we do not consider this possibility, as our goal here is just to illustrate the role of the control qubit.
\\
Note also that Rubino \emph{et al.}~\cite{rubino17} used yet another definition of causal nonseparability, different from the ones discussed in Sec.~\ref{sec:defs}, which did not allow for dynamical causal orders. 
As argued before and discussed in Refs.~\cite{oreshkov16,abbott16}, such a definition is however too restrictive to really characterise processes that are compatible with a well-defined causal order, as one would like the notion of causal separability to do. 
Nevertheless, it turns out that the witness constructed and experimentally tested in Ref.~\cite{rubino17} is not only a witness for fixed (nondynamical) causal orders, but also witnesses causal nonseparability as per our Definition~\ref{def:our_def-CS}.
}
Labelling the relevant incoming and outgoing systems (where the superscripts indicate control and target qubits) $A_O^c$, $B_I^t$, $B_O^t$, $C_I^t$, $C_O^t$, $D_I^t$, $D_I^c$, the process matrix for the quantum switch can be written~\cite{araujo15,oreshkov16,oreshkov18}
\begin{align}\label{eq:4partiteswitch}
& W^{\text{switch}} = \Tr_{D_I^t}\ketbra{w} \quad \text{with} \notag \\[1mm]
& \ket{w} = \ket{0}^{A_O^c} \ket{\psi}^{B_I^t} \dket{\id}^{B_O^tC_I^t} \dket{\id}^{C_O^tD_I^t} \ket{0}^{D_I^c} \notag \\
& \hspace{9mm}+ \ket{1}^{A_O^c} \ket{\psi}^{C_I^t} \dket{\id}^{C_O^tB_I^t} \dket{\id}^{B_O^tD_I^t} \ket{1}^{D_I^c},
\end{align}
where $\dket{\id} \coloneqq \ket{0}\!\ket{0} + \ket{1}\!\ket{1}$ is the pure CJ representation (in the computational basis $\{\ket{0}, \ket{1}\}$) of an identity qubit channel.
Note that, while Alice has control over the causal order of the other parties, this switch differs from a classical dynamical control of causal order in that she has coherent quantum control over the control qubit (and thus the causal orders).

In this particular restricted fourpartite scenario, our necessary and sufficient conditions for the causal separability of a process matrix $W$ coincide and reduce to the existence of a decomposition of the form $W = W_{(A,B,C,D)} + W_{(A,C,B,D)}$ with $W_{(A,B,C,D)}, W_{(A,C,B,D)} \ge 0$ (which need not be valid process matrices) satisfying ${}_{[1-B_O]C_{IO}D_{I}}W_{(A,B,C,D)} = {}_{[1-C_O]D_{I}}W_{(A,B,C,D)} = {}_{[1-C_O]B_{IO}D_{I}}W_{(A,C,B,D)} = {}_{[1-B_O]D_{I}}W_{(A,C,B,D)} = 0$ (and with ${}_{[1-A_O]B_{IO}C_{IO}D_{I}}W = 0$ to ensure, with the previous constraints, that $W$ is valid); see Proposition~\ref{prop:CS_charact_4_no_AI_no_DO} in Appendix~\ref{app:particular_4partite}.
These conditions thus characterise precisely the cone $\W^\sep$ in the scenario considered here, and the dual cone of causal witnesses $\mathcal{S}$ is then readily obtainable (see Eq.~\eqref{eq:Scone_N4_no_AI_no_DO} for the explicit characterisation).

The causal nonseparability of $W^\text{switch}$ can thus be verified by solving the dual SDP problem~\eqref{eq:dualSDP} and thereby obtaining a witness of its causal nonseparability.
Doing so, we find that (up to numerical precision) the random robustness of $W^\text{switch}$ of $2.343$ (note that this does not depend on the choice of initial state of the target qubit, so in solving the SDP problem numerically we can take, e.g., $\ket{\psi}=\ket{0}$).
In experimental efforts to measure a witness and verify the causal nonseparability of a process matrix, one may only have access to a restricted set of operations for the parties.
Many natural such constraints can also be imposed as SDP constraints, as described in Ref.~\cite{branciard16a}, allowing one to find implementable causal witnesses.
A particularly natural such constraint is to restrict $B$ and $C$'s operations to unitary operations (as in the experimental implementation of the tripartite switch in Refs.~\cite{procopio15,goswami18}); we find that the tolerable white noise on $W^\text{switch}$ to witness its causal nonseparability is reduced, under such a restriction, to $0.746$.

It is important to note that if we trace out the last party from $W^\text{switch}$ (i.e., $D_I^c$ in addition to $D_I^t$), we obtain
\begin{align}
\Tr_D W^{\text{switch}} &= \ketbra{0}^{A_O^c} \ketbra{\psi}^{B_I^t} \dketbra{\id}^{B_O^tC_I^t} \id^{C_O^t} \notag \\
& \ \ + \ketbra{1}^{A_O^c} \ketbra{\psi}^{C_I^t} \dketbra{\id}^{C_O^tB_I^t} \id^{B_O^t}, \label{eq:TrD_Wswicth}
\end{align}
which is causally separable since it is of the form of Eq.~\eqref{eq:decomp_CS_W_3} with just the first two terms being nonzero: $\Tr_D W^{\text{switch}} = W_{(A)} = W_{(A,B,C)} + W_{(A,C,B)}$, with $W_{(A,B,C)}$, $W_{(A,C,B)}$ (defined as the first and second terms in Eq.~\eqref{eq:TrD_Wswicth} above, respectively) and $W_{(A)}$ satisfying the constraints of Eqs.~\eqref{eq:decomp_CS_W_3_constrWX}--\eqref{eq:decomp_CS_W_3_constrWXYZ}.
This was also the case with the original tripartite version of the quantum switch (in which the control qubit is in the fixed state $\frac{1}{\sqrt{2}}(\ket{0}+\ket{1})$).
There, one is left with a simple probabilistic mixture of channels in two different directions after tracing out the last party~\cite{araujo15,oreshkov16}.
In contrast here, Eq.~\eqref{eq:TrD_Wswicth} is not compatible with any probabilistic mixture of fixed causal orders: indeed, $W_{(A,B,C)}$ and $W_{(A,C,B)}$ are not valid process matrices, as ${}_{[1-A_O]B_{IO}C_{IO}} W_{(A,B,C)} = - \, {}_{[1-A_O]B_{IO}C_{IO}} W_{(A,C,B)} \neq 0$ (these terms cancel in the sum $W_{(A,B,C)} + W_{(A,C,B)}$, so that ${}_{[1-A_O]B_{IO}C_{IO}} W_{(A)} = 0$ as required for $W_{(A)}$ to be a valid process matrix).
Rather, $\Tr_D W^{\text{switch}}$ is a ``classical switch'' in which $A$ can incoherently control the causal order between $B$ and $C$, which thus allows for dynamical causal orders.

\section{Discussion}

In this paper we studied the question of how to generalise the concept of causal (non)separability to the multipartite case. 
We reviewed several definitions that had been proposed for multipartite scenarios in previous works, namely the definition of causal separability introduced by Ara\'ujo \emph{et al.}~\cite{araujo15} for a particular tripartite situation, and Oreshkov and Giarmatzi's definitions of causal separability (CS) and extensible causal separability (ECS)~\cite{oreshkov16} for the general multipartite case. 
We established the equivalence between Ara\'ujo \emph{et al.}'s (restricted) definition of causal separability and Oreshkov and Giarmatzi's definition of ECS in the particular tripartite situation considered by Ara\'ujo \emph{et al.}, thus linking two \emph{a priori} different definitions for that case. 
Moreover, by showing that ECS and CS are different in that scenario, we found that the two definitions of causal separability proposed by Ara\'ujo \emph{et al.}~\cite{araujo15} and by Oreshkov and Giarmatzi~\cite{oreshkov16} were inconsistent, a problem that thus needed to be addressed.

We proposed a new general definition of $N$-partite causal nonseparability, similar in spirit to the recursive definitions that have been proposed for multipartite causal correlations~\cite{oreshkov16,abbott16}, and more consistent with the fact that the process matrix framework always allows for parties to share additional ancillary systems. 
Our definition thus avoids some unwanted features of the definition of CS in Ref.~\cite{oreshkov16}, such as the ``activation'' of causal nonseparability by shared entanglement. 
Moreover, we showed that our definition, although \emph{a priori} different, in fact reduces to the notion of ECS proposed in~\cite{oreshkov16}, which also reduces to the definition of Ara\'ujo \emph{et al.}~\cite{araujo15} in the particular restricted scenario considered there. 

\medskip

We then focused on characterising causally separable process matrices, giving (in the general multipartite case) two conditions---one necessary and one sufficient (Propositions~\ref{prop:CS_necessary} and~\ref{prop:CS_sufficient}, respectively)---for a given process matrix to be causally separable. 
These conditions allowed us to characterise the corresponding sets of process matrices through SDP constraints, and to generalise the tool of witnesses for causal nonseparability to the multipartite case.
In the bipartite and tripartite cases, our necessary and sufficient conditions coincide and reduce to those previously described~\cite{oreshkov12,araujo15,oreshkov16}.
The principal open question raised by this work is whether this also holds in the general $N$-partite case with $N \ge 4$, or whether one of the two is both necessary and sufficient (or if one could derive yet another distinct condition, that would we both necessary and sufficient).

As we show elsewhere, our sufficient condition characterises precisely the processes that can be realised as a quantum circuit with classical control of causal order~\cite{wechs18}. 
If that condition is in fact also necessary, this would thus confirm the conjecture of Oreshkov and Giarmatzi, that causally separable process matrices (or ``extensibly causally separable processes'' using their terminology) are those realisable by such ``classically controlled quantum circuits''~\cite{oreshkov16}. 
This would provide more solid founding for our understanding of the notion of causal separability, which would then indeed correspond to our intuition (quantum circuits with possibly dynamical causal orders that are classically controlled). 
Furthermore, the proof in Ref.~\cite{wechs18} would also provide a general explicit construction to realise any given causally separable process matrix in practice.

However, the forms of our necessary and sufficient condition, and the fact that the proof for the necessity of the conditions in the tripartite case does not generalise straightforwardly to more parties, indeed leave open the possibility that our sufficient condition may turn out to not be necessary. 
If this is the case, it would mean that there exist causally separable process matrices that are not realisable as classically controlled quantum circuits---and which we would not currently know how to realise experimentally. 
It would certainly be interesting to understand what kind of situations such process matrices correspond to---and if (and how) they can be realised quantum mechanically. 
This question is reminiscent of the open problem of whether process matrices that allow for the violation of causal inequalities are realisable with ``standard'' quantum mechanics. 
Here the question would concern even less extreme situations: causally separable process matrices.

\medskip

Another question that arises naturally in the multipartite case is whether a given phenomenon is \emph{genuinely} multipartite, in the sense that its occurrence truly requires the coordinated action of a certain number of parties. 
It would be important for our understanding of multipartite process matrices to define a notion of ``genuinely multipartite causal nonseparability'', similar to the concept of ``genuinely multipartite noncausality'' for correlations~\cite{abbott17} and analogous to the notions of genuinely multipartite entanglement~\cite{Acin01} and nonlocality~\cite{svetlichny87,gallego12,bancal13}. 
It would then also be interesting to study whether the definition can be refined to give a hierarchy of degrees of causal nonseparability, similar to the approach in Ref.~\cite{abbott17} for correlations, and whether the characterisation of the corresponding process matrices and the construction of ``witnesses of  genuinely multipartite causal nonseparability'' are still possible with SDP techniques.
These questions are left for further research.

\medskip

Finally, our clarification of the definition of causal separability in the $N$-partite scenario, as well as characterisations of causally separable process matrices, should be helpful in the study of causal nonseparability as a computational and information processing resource~\cite{chiribella13,chiribella12,colnaghi12,araujo14,facchini15,feix15,guerin16,ebler18,salek18,chiribella18}.
Indeed, since multipartite scenarios offer significantly richer structure, understanding these scenarios is a prerequisite to grasping the full possibilities of causal indefiniteness as a resource, and an important step towards developing a resource theory of noncausality~\cite{Jia18,Guerin18}.

\subsection*{Acknowledgements}

We acknowledge fruitful discussions with F.\ Costa, C.\ Giarmatzi and O.\ Oreshkov, and financial support from the `Retour Post-Doctorants' (ANR-13-PDOC-0026) and the `Investissements d'avenir' (ANR-15-IDEX-02) programs of the French National Research Agency.

\appendix

\section{Characterisation of process matrices}
\label{app:characterisation_Ws}

In this first appendix we show how the valid process matrices, as well as those compatible with a given causal order, can be characterised.
We then discuss some properties of process matrices, and alternative characterisations.

Recall that in the construction of the process matrix framework (as in Sec.~\ref{subsec:W_formalism_N_2}), the Choi-Jamio\l{}kowski (CJ) isomorphism~\cite{jamiolkowski72,choi75} is used to represent the parties' operations. Different versions of the CJ isomorphism exist; following Ref.~\cite{oreshkov12}, one may for instance define the CJ representation of a linear map $\M: A_I \to A_O$ as
\begin{align}
M \coloneqq & \left[{\cal I} \otimes \M (\dketbra{\id})\right]^\text{T} \notag \\[-1mm]
= & \Big[\sum_{m,m'} \oprod{m}{m'} \otimes \M(\oprod{m}{m'})\Big]^\text{T} \in A_I \otimes A_O,
\end{align}
where ${\cal I}$ is the identity channel, $\{\ket{m}\}_m$ is a fixed orthonormal basis of $\HS^{A_I}$, $\dket{\id} \coloneqq \sum_m \ket{m} \otimes \ket{m}$, and $\text{T}$ denotes transposition in the chosen basis $\{\ket{m}\}_m$ of $\HS^{A_I}$ and some fixed basis of $\HS^{A_O}$.
From its CJ representation $M$ it is easy to recover the map $\M$, using $\M (\rho) = \big[\Tr_{A_I} [M \cdot \rho \otimes \id^{A_O}]\big]^\text{T}$; see for instance Appendix~A in Ref.~\cite{araujo15} for more details.
Referring to this isomorphism, in the following we always identify linear maps with their CJ representation; recall in particular that a linear map is completely positive if and only if its CJ representation is positive semidefinite, and it is trace preserving if and only if its CJ representation satisfies $\Tr_{A_O} M = \id^{A_I}$.

\subsection{Valid process matrices}
\label{app:charact_valid_Ws}

A given matrix $W \in A_{IO}^\N$ defines a valid $N$-partite process matrix if and only if it generates nonnegative and normalised probabilities $P(\vec a | \vec x)$ through the generalised Born rule of Eq.~\eqref{eq:born_rule_Npartite}---including in the case where an ancillary quantum state $\rho$ in some extension $A_{I'}^\N$ of the parties' incoming spaces is attached to $W$ (and thus shared among the $N$ parties), and the parties' operations are allowed to act on their joint incoming systems $A_{II'}^k \coloneqq A_I^k \otimes A_{I'}^k$.

The constraint that the probabilities in Eq.~\eqref{eq:born_rule_Npartite} are nonnegative for any set of CP maps---i.e., any positive semidefinite matrices $M_{a_k|x_k}^{A_{IO}^k}$---and that this remains the case when attaching any ancillary quantum state $\rho \in A_{I'}^\N$ (and for any $M_{a_k|x_k}^{A_{II'O}^k} \ge 0$), translates into the constraint that $W$ must be positive semidefinite~\cite{oreshkov12,araujo15}.

As for normalisation, the constraint is that Eq.~\eqref{eq:born_rule_Npartite} must give a total probability equal to $1$ for any set of CPTP maps---i.e., any positive semidefinite matrices $M_{x_k}^{A_{IO}^k}$ satisfying ${}_{A_O^k}M_{x_k}^{A_{IO}^k} = \frac{\id}{d_{A_O^k}}$, using the ``trace-out-and-replace'' notation ${}_X\cdot$ defined in Eq.~\eqref{eq:trade_and_pad}. It is easy to see that the constraint of positive semidefiniteness does not play any role here; furthermore, note that for any matrix $M \in A_{IO}$, $M' \coloneqq \frac{\id}{d_{A_O}} + {}_{[1-A_O]}M$ satisfies ${}_{A_O}M' = \frac{\id}{d_{A_O}}$ and that any $M \in A_{IO}$ satisfying ${}_{A_O}M = \frac{\id}{d_{A_O}}$ is of the form $M = \frac{\id}{d_{A_O}} + {}_{[1-A_O]}M$. The normalisation constraint thus translates into the constraint that
\begin{equation}
\Tr\big[({\textstyle \frac{\id}{d_{A_O^1}}} + {}_{[1-A_O^1]}M_1) \otimes \cdots \otimes ({\textstyle \frac{\id}{d_{A_O^N}}} + {}_{[1-A_O^N]}M_N) \cdot W\big] = 1 \label{eq:constr_valid_W}
\end{equation}
for any set of matrices $M_1 \in A_{IO}^1, \ \ldots, \ M_N \in A_{IO}^N$.
Expanding this constraint, and using the fact that the map ${}_{[1-A_O^k]}\cdot$ is self adjoint with respect to the trace (Hilbert–Schmidt) inner product (i.e., $\Tr[{}_{[1-A_O^k]} M \cdot W] = \Tr[M \cdot {}_{[1-A_O^k]} W]$), one finds that this is equivalent to
\begin{align}
& \Tr W = {\textstyle \prod_{k \in \N}} \, d_{A_O^k} \qquad \textrm{and} \label{eq:constr_norm_W} \\
& \forall \ \X \subseteq \N, \X \neq \emptyset, \ {}_{\prod_{i \in \X}[1-A_O^i]}\Tr_{\N \backslash \X}W = 0. \label{eq:constr_LN_v0}
\end{align}
Note that for simplicity we did not explicitly attach an ancillary state $\rho$ to $W$ here; doing so would have led to the same conclusion. Full details for this whole argument can be found in Appendix~B of Ref.~\cite{araujo15}.

We shall in general ignore the normalisation constraint~\eqref{eq:constr_norm_W} when talking about valid process matrices. The $2^N-1$ linear constraints of Eq.~\eqref{eq:constr_LN_v0} define a linear subspace of $A_{IO}^\N$, which we denote by $\L^\N$, the subspace of valid process matrices: explicitly (noting that the constraints ${}_{\prod_{i \in \X}[1-A_O^i]}\Tr_{\N \backslash \X}W = 0$ are equivalent to ${}_{\prod_{i \in \X}[1-A_O^i]A_{IO}^{\N \backslash \X}}W = 0$),
\begin{align}
& W \in \L^\N \notag \\
& \Leftrightarrow \, \forall \ \X \subseteq \N, \X \neq \emptyset, \ {}_{\prod_{i \in \X}[1-A_O^i]A_{IO}^{\N \backslash \X}}W = 0, \label{eq:constr_LN_app}
\end{align}
as in Eq.~\eqref{eq:constr_LN} of the main text.
It is furthermore straightforward to check that this is equivalent to the following recursive characterisation, as in Eq.~\eqref{eq:constr_LN_rec}:
\begin{align}
& W \in \L^\N \Leftrightarrow \, \forall \ \X \subsetneq \N, \X \neq \emptyset, \ \Tr_{\N \backslash \X}  W \in \L^\X \notag \\
& \hspace{20mm} \text{ and } \ {}_{ \prod_{i \in \N}[1-A_O^i]} W = 0 \,. \label{eq:constr_LN_rec_app}
\end{align}

Summing up, the set of (nonnormalised) valid process matrices is the convex cone $\W = {\cal P} \cap \L^\N$, where ${\cal P}$ is the cone of positive semidefinite matrices.

\subsection{Compatibility with fixed causal orders}
\label{app:fixed_order}

Let us now analyse the constraints imposed on process matrices by requiring that they are compatible with a given causal order.

\subsubsection{Causal order between two subsets of parties}

Consider two nonempty disjoint subsets of parties $\K_1, \K_2 \subsetneq \N$. We say that the correlation $P(\vec a | \vec x)$ is compatible with the causal order $\K_1 \prec \K_2$ if and only if there is no signalling from the parties in $\K_2$ to the parties in $\K_1$---i.e., the marginal probability distribution for the outputs of parties in $\K_1$ does not depend on the inputs of parties in $\K_2$: $P(\vec a_{\K_1} | \vec x) = P(\vec a_{\K_1} | \vec x_{\N \backslash \K_2})$ for all $\vec x, \vec a_{\K_1}$.
We then say that a (valid) process matrix $W$ is compatible with the causal order $\K_1 \prec \K_2$ if and only if it only generates correlations (through the generalised Born rule~\eqref{eq:born_rule_Npartite}, possibly allowing for extensions of $W$ with some ancillary state $\rho$) that are compatible with $\K_1 \prec \K_2$.

Formally, this means (ignoring again for simplicity the possibility of attaching an ancillary state $\rho$; as before, the same reasoning also goes through if we allow for this possibility) that whatever the CP maps $M_{a_{k_1}|x_{k_1}}^{A_{IO}^{k_1}}$ applied by the parties in $\K_1$ and whatever the CPTP maps $M_{x_{k_2}}^{A_{IO}^{k_2}}$ and $M_{x_{k_3}}^{A_{IO}^{k_3}}$ (such that ${}_{A_O^{k_{2/3}}}M_{x_{k_{2/3}}}^{A_{IO}^{k_{2/3}}} = \frac{\id}{d_{A_O^{k_{2/3}}}}$) applied by the parties in $\K_2$ and in $\overline{\K_{12}} \coloneqq \N \backslash (\K_1 \cup \K_2)$ (which may be empty), respectively, one must have
\begin{align}
& \Tr[ \bigotimes_{k_1 \in \K_1} M_{a_{k_1}|x_{k_1}}^{A_{IO}^{k_1}} \bigotimes_{k_2 \in \K_2} M_{x_{k_2}}^{A_{IO}^{k_2}} \bigotimes_{k_3 \in \overline{\K_{12}}} M_{x_{k_3}}^{A_{IO}^{k_3}} \cdot W] \notag \\[1mm]
& = \Tr[ \bigotimes_{k_1 \in \K_1} M_{a_{k_1}|x_{k_1}}^{A_{IO}^{k_1}} \bigotimes_{k_2 \in \K_2} \frac{\id}{d_{A_O^{k_2}}} \bigotimes_{k_3 \in \overline{\K_{12}}} M_{x_{k_3}}^{A_{IO}^{k_3}} \cdot W]
\end{align}
(i.e., the probabilities should be the same if the parties in $\K_2$ apply the CPTP maps $\frac{\id}{d_{A_O^{k_2}}}$ instead of $M_{x_{k_2}}^{A_{IO}^{k_2}}$).
As in the previous subsection, the constraint of positive semidefiniteness of the CJ matrices $M_k$ does not play any role here, and we can equivalently write the above constraint as
\begin{align}
& \Tr[ \bigotimes_{k_1 \in \K_1} M_{k_1} \bigotimes_{k_2 \in \K_2} ({\textstyle \frac{\id}{d_{A_O^{k_2}}}} + {}_{[1-A_O^{k_2}]}M_{k_2}) \notag \\[-2mm]
& \hspace{35mm} \bigotimes_{k_3 \in \overline{\K_{12}}} ({\textstyle \frac{\id}{d_{A_O^{k_3}}}} + {}_{[1-A_O^{k_3}]}M_{k_3}) \cdot W] \notag \\[1mm]
& = \Tr[ \bigotimes_{k_1 \in \K_1} \!\! M_{k_1} \!\bigotimes_{k_2 \in \K_2} \!{\textstyle \frac{\id}{d_{A_O^{k_2}}}} \!\bigotimes_{k_3 \in \overline{\K_{12}}} \!({\textstyle \frac{\id}{d_{A_O^{k_3}}}} + {}_{[1-A_O^{k_3}]}M_{k_3}) \cdot W]
\end{align}
for any matrices $M_k \in A_{IO}^k$.
Expanding this constraint in a similar way as above (or as it was done in more details in Appendix~B of Ref.~\cite{araujo15}), we find that it is equivalent to the following $2^{N-|\K_1|}-2^{N-|\K_1|-|\K_2|}$ linear constraints:
\begin{align}
& \forall \ \X_2 \subseteq \N \backslash \K_1, \X_2 \cap \K_2 \neq \emptyset, \notag \\
& \qquad {}_{\prod_{i \in \X_2}[1-A_O^i]A_{IO}^{\N \backslash \K_1 \backslash \X_2}}W = 0. \label{eq:constr_W_K1_K2_K12}
\end{align}
Combining these conditions with those from Eq.~\eqref{eq:constr_LN_app} to ensure that $W$ is a valid process matrix, and removing redundant constraints,%
\footnote{One can easily see that the constraints from Eq.~\eqref{eq:constr_LN_app} for which $\X \cap \K_2 \neq \emptyset$ are already implied by those of Eq.~\eqref{eq:constr_W_K1_K2_K12}: indeed, defining $\X_1 \coloneqq \X \cap \K_1$ and $\X_2 \coloneqq \X \cap (\N \backslash \K_1) \subseteq \N \backslash \K_1$, in such a case one has $\X_2 \cap \K_2 \neq \emptyset$ and ${}_{\prod_{i \in \X}[1-A_O^i]A_{IO}^{\N \backslash \X}}W = {}_{\prod_{i \in \X_1}[1-A_O^i]A_{IO}^{\K_1 \backslash \X_1}}\big(\,{}_{\prod_{i \in \X_2}[1-A_O^i]A_{IO}^{\N \backslash \K_1 \backslash \X_2}}W\big) = 0$ according to Eq.~\eqref{eq:constr_W_K1_K2_K12}.
Only the constraints from Eq.~\eqref{eq:constr_LN_app} for which $\X \cap \K_2 = \emptyset$, i.e., $\X \subseteq \N \backslash \K_2$ as in the second line of Eq.~\eqref{eq:constr_LK1K2_K12} (where $\X$ was renamed $\X_1$), are nonredundant. \label{footnote:remove_ctsr}}
one can then characterise the subspace $\L^{\K_1 \prec \K_2}$ of (valid) process matrices compatible with the causal order $\K_1 \prec \K_2$ through the following $2^{N-|\K_2|}-1+2^{N-|\K_1|}-2^{N-|\K_1|-|\K_2|}$ constraints:
\begin{align}
& W \in \L^{\K_1 \prec \K_2} \notag \\[1mm]
& \Leftrightarrow \ \forall \ \X_1 \subseteq \N \backslash \K_2, \X_1 \neq \emptyset, \ {}_{\prod_{i \in \X_1}[1-A_O^i]A_{IO}^{\N \backslash \X_1}}W = 0 \notag \\
& \qquad \text{and} \ \forall \ \X_2 \subseteq \N \backslash \K_1, \X_2 \cap \K_2 \neq \emptyset, \notag \\[-1mm]
& \hspace{32mm} {}_{\prod_{i \in \X_2}[1-A_O^i]A_{IO}^{\N \backslash \K_1 \backslash \X_2}}W = 0. \label{eq:constr_LK1K2_K12}
\end{align}

\medskip

Let us assume now that $\K_1$ and $\K_2$ cover the full set $\N$, so that $\K_1 \cup \K_2 = \N$. The characterisation above then simplifies to the following $2^{|\K_1|}+2^{|\K_2|}-2$ constraints:
\begin{align}
& \hspace{-3mm} [\text{For }\K_1 \cup \K_2 = \N:] \notag \\
& W \in \L^{\K_1 \prec \K_2} \notag \\[1mm]
& \Leftrightarrow \ \forall \, \X_1 \subseteq \K_1, \X_1 \neq \emptyset, {}_{\prod_{i \in \X_1}[1-A_O^i]A_{IO}^{\K_1 \backslash \X_1}A_{IO}^{\K_2}}W = 0 \notag \\
& \quad \ \ \text{and} \ \forall \, \X_2 \subseteq \K_2, \X_2 \neq \emptyset, {}_{\prod_{i \in \X_2}[1-A_O^i]A_{IO}^{\K_2 \backslash \X_2}}W = 0. \label{eq:constr_LK1K2}
\end{align}
Comparing these constraints with Eq.~\eqref{eq:constr_LN_app}, one can see that the third line of Eq.~\eqref{eq:constr_LK1K2} is equivalent to imposing that the reduced process $\Tr_{\K_2} W$ is in $\L^{\K_1}$, the subspace of valid $|\K_1|$-partite process matrices for parties in $\K_1$; while the fourth line is equivalent (using the fact that $W=0$ if and only if $\Tr_{\K_1}[M_{\K_1} \otimes \id^{\K_2} \cdot W] = 0$ for all $M_{\K_1}$) to imposing that whatever CP maps $M_{\K_1} \in A_{IO}^{\K_1}$ applied by the parties in $\K_1$, the conditional matrix $W_{|M_{\K_1}} \coloneqq \Tr_{\K_1} [M_{\K_1} \otimes \id^{\K_2} \cdot W]$ must be in the subspace $\L^{\K_2}$ of valid process matrices for the parties in $\K_2$ (note that $M_{\K_1}$ may or may not be of a product form $\bigotimes_{k_1 \in \K_1} M_{k_1}$ here, and that its complete positiveness is in fact irrelevant).%
\footnote{Although the constraints in the fourth line of Eq.~\eqref{eq:constr_LK1K2} are written exactly as those that would define $\L^{\K_2}$, we emphasise that they apply here to some matrix $W \in A_{IO}^\N$, rather than to $W \in A_{IO}^{\K_2}$ as in the definition of $\L^{\K_2}$. This is why they must of course not directly be interpreted as implying that $W \in \L^{\K_2}$, but $W_{|M_{\K_1}} \in \L^{\K_2}$ for all $M_{\K_1}$, as in Eq.~\eqref{eq:constr_LK1K2_v2}.}
We thus equivalently have the following characterisation:
\begin{align}
& \hspace{-2mm} [\text{For }\K_1 \cup \K_2 = \N:] \notag \\
& W \in \L^{\K_1 \prec \K_2} \, \Leftrightarrow \, \Tr_{\K_2}W \in \L^{\K_1} \notag \\ & \hspace{25mm} \text{and} \ \forall \, M_{\K_1} \!\in\! A_{IO}^{\K_1}, \, W_{|M_{\K_1}} \!\in\! \L^{\K_2}. \label{eq:constr_LK1K2_v2}
\end{align}
These constraints are indeed quite intuitive: they simply correspond to the fact that for a process matrix correlation $P(\vec a|\vec x)$ to be compatible with the causal order $\K_1 \prec \K_2$ (with $\K_1 \cup \K_2 = \N$), the probability distributions $P(\vec a_{\K_1}|\vec x_{\K_1})$ and $P(\vec a_{\K_2}|\vec x,\vec a_{\K_1}) \coloneqq P(\vec a|\vec x)/P(\vec a_{\K_1}|\vec x_{\K_1})$ can be calculated from $\Tr_{\K_2}W$ and $W_{|M_{\K_1}}$, and must be well-defined.

In particular, for $\K_1 = \{A_k\}$ (a singleton of just one party coming first) and $\K_2 = \N \backslash A_k$, Eq.~\eqref{eq:constr_LK1K2} becomes
\begin{align}
& W \in \L^{A_k \prec (\N \backslash A_k)} \notag \\[1mm]
& \Leftrightarrow \ {}_{[1-A_O^k]A_{IO}^{\N \backslash k}} W = 0 \quad \text{and} \notag \\
& \qquad \forall \ \X \subseteq \N \backslash k, \X \neq 0, {}_{\prod_{i \in \X}[1-A_O^i]A_{IO}^{\N \backslash k \backslash \X} } W = 0, \label{eq:constr_L_k_first_app}
\end{align}
as in Eq.~\eqref{eq:constr_L_k_first} of the main text.
For $\K_1 = \N \backslash A_k$ and $\K_2 = \{A_k\}$ (a singleton of just one party coming last), Eq.~\eqref{eq:constr_LK1K2} becomes
\begin{align}
& W \in \L^{(\N \backslash A_k) \prec A_k} \notag \\[1mm]
& \Leftrightarrow \ \forall \, \X \subseteq \N \backslash k, \X \neq \emptyset, {}_{\prod_{i \in \X}[1-A_O^i]A_{IO}^{\N \backslash k \backslash \X}A_{IO}^k}W = 0 \notag \\
& \qquad \text{and} \ {}_{[1-A_O^{k}]}W = 0. \label{eq:constr_L_Alast}
\end{align}

\subsubsection{Causal order between several subsets of parties}

Consider now $K$ disjoint subsets $\K_i$ of $\N$.
Generalising the idea of causal order between two subsets of parties, we say that the correlation $P(\vec a|\vec x)$ is compatible with the causal order $\K_1 \prec \K_2 \prec \cdots \prec \K_K$ if and only if there is no signalling from ``future parties'' to ``past parties''---i.e., if for any $k = 1, \ldots, K-1$, the outputs of parties in $\K_{(\le k)} \coloneqq \bigcup_{i=1}^k \K_i$ do not depend on the inputs of the parties in $\K_{(> k)} \coloneqq \bigcup_{j=k+1}^K \K_j$: $P(\vec a_{\K_{(\le k)}} | \vec x) = P(\vec a_{\K_{(\le k)}} | \vec x_{\N \backslash \K_{(> k)}})$ for all $\vec x, \vec a_{\K_{(\le k)}}$, or equivalently, the correlation is compatible with the causal order $\K_{(\le k)} \prec \K_{(> k)}$ for all $k = 1, \ldots, K-1$.

As before, we then say that a process matrix $W$ is compatible with the causal order $\K_1 \prec \cdots \prec \K_K$ if and only if it only generates correlations that are compatible with that order. Similarly to $\L^{\K_1 \prec \K_2}$, we define the subspace
\begin{equation}
\L^{\K_1 \prec \cdots \prec \K_K} \coloneqq \bigcap_{k=1}^{K-1} \L^{\K_{(\le k)} \prec \K_{(> k)}}
\end{equation}
of (valid) process matrices compatible with the causal order $\K_1 \prec \cdots \prec \K_K$.

In the case where the subsets $\K_i$ define a full partition of $\N$ (i.e., where $\bigcup_{i=1}^K \K_i = \N$), we easily obtain, from Eqs.~\eqref{eq:constr_LK1K2} and~\eqref{eq:constr_LK1K2_v2} (after removing redundant constraints as in Footnote~\ref{footnote:remove_ctsr}), that
\begin{align}
& \hspace{-3mm} [\text{For }{\textstyle \bigcup_{i=1}^K} \K_i = \N:] \notag \\
& W \in \L^{\K_1 \prec \cdots \prec \K_K} \notag \\
& \Leftrightarrow \ \forall \, k = 1,\ldots,K, \ \forall \, M_{\K_{(< k)}} \in A_{IO}^{\K_{(< k)}}, \notag \\
& \hspace{30mm} \Tr_{\K_{(>k)}}W_{|M_{\K_{(< k)}}} \in \L^{\K_k} \notag \\[1mm]
& \Leftrightarrow \ \forall \, k = 1,\ldots,K, \ \forall \, \X_k \subseteq \K_k, \X_k \neq \emptyset, \notag \\
& \hspace{20mm} {}_{\prod_{i \in \X_k}[1-A_O^i]A_{IO}^{\K_k \backslash \X_k}A_{IO}^{\K_{(>k)}}}W = 0 \label{eq:constr_LK1_KN}
\end{align}
with $\K_{(< k)} \coloneqq \bigcup_{i=1}^{k-1} \K_i$ for $k = 2, \ldots, K$, $\K_{(< 1)} = \K_{(>K)} = \emptyset$, $W_{|M_{\K_{(< 1)}}} = W$.

In particular, in the case where all subsets $\K_k$ are singletons---i.e., $\K_k = \{A_{\pi(k)}\}$ for some permutation $\pi$ of $\N$---we find that the subspace $\L^{A_{\pi(1)} \prec \cdots \prec A_{\pi(N)}}$ of process matrices compatible with the causal order $A_{\pi(1)} \prec A_{\pi(2)} \prec \cdots \prec A_{\pi(N)}$ is characterised by~\cite{gutoski06,chiribella09,araujo15}
\begin{align}
& W \in \L^{A_{\pi(1)} \prec \cdots \prec A_{\pi(N)}} \notag \\
& \Leftrightarrow \ \forall \, k = 1,\ldots,N, \ {}_{[1-A_O^{{\pi(k)}}]A_{IO}^{\pi(>k)}}W = 0 \label{eq:constr_fixed_order_N}
\end{align}
with $\pi(>\!k) = \{\pi(k{+}1), \ldots, \pi(N)\}$ for $k = 1,\ldots,N{-}1$, $\pi(>\!N) = \emptyset$ (as in Eq.~\eqref{eq:constr_order_A1__AN} when $\pi$ is the identity permutation).

\subsubsection{Particular cases with $d_{A_I^f} = 1$ or $d_{A_O^\ell} = 1$}
\label{app:fixed_order_part_cases}

Suppose there exists a party $A_f$ which has a trivial incoming space, i.e., such that $d_{A_I^f} = 1$.
The constraints of Eq.~\eqref{eq:constr_LN_app} can be written, depending on whether $A_f \in \X$ or $A_f \notin \X$ and renaming $\X \backslash A_f \to \X$ in the former case, in the forms ${}_{[1-A_O^f]\prod_{i \in \X}[1-A_O^i]A_{IO}^{\N \backslash A_f \backslash \X}}W = 0$ for all $\X \subseteq \N \backslash A_f$ and ${}_{A_{O}^{A_f}\prod_{i \in \X}[1-A_O^i]A_{IO}^{\N \backslash A_f \backslash \X}}W = 0$ for all $\X \subseteq \N \backslash A_f, \X \neq \emptyset$, respectively.
Summing up these two constraints in the case where $\X \neq \emptyset$ (and keeping the first one for the case where $\X = \emptyset$), we find that $\L^\N$ is characterised by the same constraints as those characterising $\L^{A_f \prec (\N \backslash A_f)}$ in Eq.~\eqref{eq:constr_L_k_first_app}, namely
\begin{align}
& \hspace{-3mm} [\text{For }d_{A_I^f} = 1:] \notag \\
& W \in \L^\N \Leftrightarrow W \in \L^{A_f \prec (\N \backslash A_f)} \notag \\[1mm]
& \Leftrightarrow \, {}_{[1-A_O^f]A_{IO}^{\N \backslash A_f}}W = 0 \ \ \text{and} \notag \\
& \hspace{5mm} \forall \, \X \subseteq \N \backslash A_f, \X \neq \emptyset, \ {}_{\prod_{i \in \X}[1-A_O^i]A_{IO}^{\N \backslash A_f \backslash \X}}W = 0. \label{eq:LN_dAIf_1}
\end{align}
Hence, in that case any valid process matrix is compatible with party $A_f$ acting first. This corresponds indeed to the natural intuition that, because party $A_f$ does not receive any physical system from anyone, they do not need to wait for any other party to act before them.

In the case where several parties in $\A_f \coloneqq \{A_{f_1}, \ldots, A_{f_n}\}$ have trivial incoming spaces (such that $d_{A_I^{f_j}} = 1$ for all $j = 1, \ldots, n$), Eq.~\eqref{eq:LN_dAIf_1} easily generalises to%
\footnote{We use here, in particular, the fact that ${}_{[1-A_O^{f_j}]A_{IO}^{\N \backslash \A_f}}W = 0\ \forall \, j$ is equivalent to ${}_{[1-A_O^{\A_f}]A_{IO}^{\N \backslash \A_f}}W = 0$. Note that if all incoming spaces are trivial, i.e., $\A_f = \N$, then ${}_{[1-A_O^\N]}W = 0$ implies that the only valid process matrices are those proportional to the identity operator $\id^{A_O^\N}$. \label{footnote:all_trivial_AI}}
\begin{align}
& \hspace{-1mm} [\text{For }d_{A_I^{f_j}} = 1, \, \forall \, A_{f_j} \in \A_f:] \notag \\
& W \in \L^\N \Leftrightarrow W \in \L^{A_{f_j} \prec (\N \backslash A_{f_j})} \ \forall \, j \notag \\[1mm]
& \Leftrightarrow \, {}_{[1-A_O^{\A_f}]A_{IO}^{\N \backslash \A_f}}W = 0 \quad \text{and} \notag \\
& \hspace{5mm} \forall \, \X \subseteq \N \backslash \A_f, \X \neq \emptyset, \ {}_{\prod_{i \in \X}[1-A_O^i]A_{IO}^{\N \backslash \A_f \backslash \X}}W = 0. \label{eq:LN_dAIf_1_gen}
\end{align}

\medskip

Instead of $d_{A_I^f} = 1$, suppose now that there exists a party $A_\ell$ which has a trivial outgoing space, i.e., such that $d_{A_O^\ell} = 1$, and consider the causal order $(\N \backslash A_\ell) \prec A_\ell$. Note that in this case the map $W \to {}_{A_O^\ell}W$ reduces to the identity, so that any constraint of the form ${}_{[1-A_O^\ell]\ldots}W = 0$ is trivially satisfied. The nontrivial constraints from Eqs.~\eqref{eq:constr_LN_app} and~\eqref{eq:constr_L_Alast} then reduce to the same set of constraints, namely,
\begin{align}
& [\text{For }d_{A_O^\ell} = 1:] \notag \\
& \ W \in \L^\N \Leftrightarrow W \in \L^{(\N \backslash A_\ell) \prec A_\ell} \notag \\[1mm]
& \, \Leftrightarrow \, \forall \ \X \subseteq \N \backslash A_\ell, \X \neq \emptyset, \ {}_{\prod_{i \in \X}[1-A_O^i]A_{IO}^{\N \backslash A_\ell \backslash \X}A_I^\ell}W = 0.
\end{align}
Hence, similarly to the previous case, here any valid process matrix is compatible with party $A_\ell$ acting last. This is again rather intuitive: as party $A_\ell$ sends no physical system out and cannot signal to anyone, then they can always come last---see, e.g., the motivation for only considering fixed orders with party $C$ coming last in Ara\'ujo \emph{et al.}'s definition of causal separability~\cite{araujo15}.

If now several parties in $\A_\ell \coloneqq \{A_{\ell_1}, \ldots, A_{\ell_n}\}$ have trivial outgoing spaces (such that $d_{A_O^{\ell_j}} = 1$ for all $j = 1, \ldots, n$), then one can easily check that any process matrix is compatible with all those parties acting last, with any causal order among them: for any permutation $\pi$ of $\{1,\ldots,n\}$,
\begin{align}
& [\text{For }d_{A_O^{\ell_j}} = 1, \, \forall \, A_{\ell_j} \in \A_\ell:] \notag \\
& \ W \in \L^\N \Leftrightarrow W \in \L^{(\N \backslash \A_\ell) \prec A_{\ell_{\pi(1)}} \prec \cdots \prec A_{\ell_{\pi(n)}}} \notag \\[1mm]
& \, \Leftrightarrow \, \forall \ \X \subseteq \N \backslash \A_\ell, \X \neq \emptyset, \, {}_{\prod_{i \in \X}[1-A_O^i]A_{IO}^{\N \backslash \A_\ell \backslash \X}A_I^{\A_\ell}}W = 0. \label{eq:LN_dAOl_1_gen}
\end{align}
It is worth emphasising that no similar property holds for several parties in $\A_f = \{A_{f_1}, \ldots, A_{f_n}\}$ having trivial incoming spaces, as considered previously: any process matrix is compatible in that case with any causal order $A_{f_j} \prec (\N \backslash A_{f_j})$ (as in Eq.~\eqref{eq:LN_dAIf_1_gen}), but not necessarily with $A_{f_1} \prec \cdots \prec A_{f_n} \prec (\N \backslash \A_f)$ (or with any other permutation of the first parties).%
\footnote{Note indeed, in a similar fashion, that while compatibility of a probability distribution $P$ with the orders $(\K_1 \cup \K_2) \prec \K_3$ and $(\K_1 \cup \K_3) \prec \K_2$ implies compatibility with $\K_1 \prec (\K_2 \cup \K_3)$, and therefore with both $\K_1 \prec \K_2 \prec \K_3$ and $\K_1 \prec \K_3 \prec \K_2$, it is not the case that compatibility with $\K_1 \prec (\K_2 \cup \K_3)$ and $\K_2 \prec (\K_1 \cup \K_3)$ necessarily implies $(\K_1 \cup \K_2) \prec \K_3$, and it therefore does also not necessarily imply compatibility with $\K_1 \prec \K_2 \prec \K_3$ or $\K_2 \prec \K_1 \prec \K_3$. As a counter-example, one can see for instance that $P(a,b,c|x,y,z) = \frac12 \delta_{a \oplus b,z} \, \delta_{c,0}$ (with binary inputs and outputs, where $\delta$ the Kronecker delta and $\oplus$ denotes addition modulo 2) is compatible with both $A \prec \{B,C\}$ and $B \prec \{A,C\}$, but not with $\{A,B\} \prec C$.}

\medskip

To finish here, note, furthermore, that if a party $A_k$ has both $d_{A_I^k} = d_{A_O^k} = 1$, then clearly one can just ignore it: in such a case, $W \in \L^\N \Leftrightarrow W \in \L^{\N \backslash A_k}$.

\subsubsection{Comment on our use of the notation $\prec$}

Let us comment briefly here on our use of the notation $\prec$. 
Recall that for two disjoint nonempty subsets $\K_1$ and $\K_2$ of $\N$, a probability distribution $P$ is said to be compatible with the causal order $\K_1 \prec \K_2$ if and only if $P(\vec a_{\K_1}|\vec x) = P(\vec a_{\K_1}|\vec x_{\N \backslash \K_2})$ for all $\vec x, \vec a_{\K_1}$. It should be noted that the relation ``compatible with $\K_1 \prec \K_2$'' thus defined is not transitive, and therefore it does not define a partial order between events. For instance, $P(a,b,c|x,y,z) \coloneqq \delta_{a,z} \, \delta_{b,0} \, \delta_{c,0}$ (with $\delta$ the Kronecker delta and $a,z$ taking at least two different values) is compatible with $A \prec B$ and $B \prec C$, but not with $A \prec C$.
This justifies why, considering more subsets, we defined the notation $\K_1 \prec \K_2 \prec \cdots \prec \K_K$ to formally mean $\K_{(\le k)} \prec \K_{(> k)}$---rather than just $\K_k \prec \K_{k+1}$---for all $k = 1, \ldots, K{-}1$.

We note also that the notation $\prec$ was used differently in Ref.~\cite{oreshkov16}, where it denoted a strict partial order (and was hence transitive). Our use of the notation $\prec$ here is consistent e.g.\ with that of Refs.~\cite{araujo15,branciard16,branciard16a,
giacomini16,feix16,abbott16,miklin17,abbott17,goswami18}, and would instead correspond to the notation $\nsucceq$ in Ref.~\cite{oreshkov16} (also used in Ref.~\cite{oreshkov12}).

\subsection{Operations on process matrices}
\label{app:operations_W_matrices}

In this section we clarify how process matrices behave in general, with respect to their validity and their compatibility with fixed causal orders, when tracing out subsystems or attaching extensions, and when tracing out, adding or grouping parties.

\subsubsection{Tracing out subsystems / Attaching extensions}
\label{app:trace_out}

Suppose that the incoming and outgoing spaces of $N$ parties can be decomposed as $A_{IO}^\N \otimes A_{I'O'}^\N$ (possibly with some trivial spaces $A_I^k$, $A_{O}^k$, $A_{I'}^k$ or $A_{O'}^k$), and consider a given matrix $W' \in A_{II'OO'}^\N$. If $W'$ is a valid process matrix, then so is $W \coloneqq \Tr_{A_{I'O'}^\N} W'$; similarly, if $W'$ is compatible with a causal order $\K_1 \prec \K_2$, then so is $W$. Both properties are quite intuitive:%
\footnote{They can be verified straightforwardly using for instance Eqs.~\eqref{eq:constr_LN_app} and~\eqref{eq:constr_LK1K2_K12}, respectively, by writing (in the first case) ${}_{\prod_{i \in \X}[1-A_O^i]A_{IO}^{\N \backslash \X}}W = \Tr_{A_{I'O'}^\N} \big[ {}_{\prod_{i \in \X}[1-A_O^i]A_{IO}^{\N \backslash \X}} W' \big] = \Tr_{A_{I'O'}^\N} \big[ {}_{\prod_{i \in \X}[1-A_{OO'}^i]A_{II'OO'}^{\N \backslash \X}} W' \big] = 0$ (and by noting that $W' \ge 0$ implies $W \ge 0$).}
clearly, ignoring some parts of the incoming and outgoing spaces cannot make a process matrix invalid, and cannot induce some signalling where there was none before.
Note, however, that the converse is in general not true: if $W = \Tr_{A_{I'O'}^\N} W'$ is a valid process matrix (or is compatible with $\K_1 \prec \K_2$), this does not guarantee in general that $W'$ is also a valid process matrix (or is compatible with $\K_1 \prec \K_2$).

There is nevertheless a case, where the validity of a process matrix $W$ ensures that a ``larger'' matrix $W'$ (defined on more subsystems) is valid: namely, when one attaches to $W$ some ancillary state $\rho$.
Indeed, in constructing the framework of process matrices, it is always assumed that one can consider some extensions of the incoming spaces of each party, and distribute (possibly entangled) ancillary states shared by all parties.
Hence, by definition, if a matrix $W \in A_{IO}^\N$ is a valid process matrix, then for any quantum state (i.e., any positive semidefinite matrix, up to normalisation) $\rho$ in any extension $A_{I'}^\N$, the matrix $W' = W \otimes \rho \in A_{II'O}^\N$ defines a valid process matrix. Similarly, if $W$ is compatible with a given causal order $\K_1 \prec \K_2$ between two disjoint subsets of parties, then so is $W \otimes \rho$.

One may then wonder if instead of attaching an ancillary state $\rho \in A_{I'}^\N$ to the incoming spaces, one could attach any other positive semidefinite matrix $W^\text{ext.} \in A_{I'O'}^\N$ in some extension of both incoming and outgoing spaces.
It is clear, from the previous remarks on the partial trace of subsystems, that for a valid (nonzero) process matrix $W$, a necessary condition for $W' \coloneqq W \otimes W^\text{ext.}$ to define a valid process matrix is that $W^\text{ext.}$ itself is also a valid process matrix.%
\footnote{This implies in particular that for an extension of the outgoing systems only, $W \otimes W^\text{ext.}$ with $W^\text{ext.} \in A_{O'}^\N$ is valid only if $W^\text{ext.}$ is proportional to the identity operator: see Footnote~\ref{footnote:all_trivial_AI}.}
For two parties and more, this condition is however not sufficient (as noted also in Refs.~\cite{Jia18,Guerin18}): for instance, $W = \frac12(\id^{A_OB_I} + \hat{\textsc{z}}^{A_O} \hat{\textsc{z}}^{B_I})$ and $W^\text{ext.} = \frac12(\id^{A_{I'}B_{O'}} + \hat{\textsc{z}}^{A_{I'}} \hat{\textsc{z}}^{B_{O'}})$ are both valid, but $W \otimes W^\text{ext.}$ is not (due to the fact that $W$ and $W^\text{ext.}$ allow for some signalling in two conflicting directions).
Similarly, for a (nonzero) process matrix $W$ compatible with $\K_1 \prec \K_2$, a necessary condition for $W' \coloneqq W \otimes W^\text{ext.}$ to be a process matrix compatible with $\K_1 \prec \K_2$ is that $W^\text{ext.}$ is also a process matrix compatible with $\K_1 \prec \K_2$. As before, this condition is however not sufficient for three parties and more.

\subsubsection{Tracing out / Adding / Separating / Grouping parties}

In the previous observations we were keeping the set of parties under consideration $\N$ fixed. Let us now consider how process matrices behave when changing the set of parties.

Consider a nonempty subset $\N_0$ of $\N$. Clearly, if $W \in A_{IO}^\N$ is a valid $N$-partite process matrix, then its restriction to the parties in the subset $\N_0$, defined as $W_0 \coloneqq \Tr_{\N \backslash \N_0} W$, is a valid $|\N_0|$-partite process matrix. Similarly, if $W$ is compatible with a causal order $\K_1 \prec \K_2$, then $W_0$ is compatible with the order $(\K_1 \cap \N_0) \prec (\K_2 \cap \N_0)$.%
\footnote{Again, both properties can easily be checked by using for instance Eqs.~\eqref{eq:constr_LN_app} and~\eqref{eq:constr_LK1K2_K12}, and the fact that $W \ge 0$ implies $W_0 \ge 0$.}

Let $\N_1$ and $\N_2$ be two disjoint sets of parties. If $W_1 \in A_{IO}^{\N_1}$ and $W_2 \in A_{IO}^{\N_2}$ are two valid process matrices for the parties in $\N_1$ and $\N_2$, respectively, then so is $W = W_1 \otimes W_2 \in A_{IO}^{\N_1 \cup \N_2}$ for all parties in $\N_1 \cup \N_2$.
(Note however that if $\N_1$ and $\N_2$ are not disjoint, this may not hold any more, as in the case with $\N_1 = \N_2 = \N$ considered in the previous subsection.)
If say $W_1$ is compatible with a causal order $\K_1 \prec \K_1'$ (with $\K_1, \K_1'$ two disjoint nonempty subsets of $\N_1$), then so is $W$.
Furthermore, for any nonempty subsets $\K_1 \subseteq \N_1$ and $\K_2 \subseteq \N_2$, $W$ is compatible with both orders $\K_1 \prec \K_2$ and $\K_2 \prec \K_1$.

Suppose now that the incoming and outgoing spaces of a party, say $A_N$, can be factorised into $A_{IO}^N = A_{IO}^{N^{(1)}} \otimes A_{IO}^{N^{(2)}}$. One can then virtually ``separate'' $A_N$ into two new parties, $A_{N^{(1)}}$ and $A_{N^{(2)}}$, with incoming and outgoing spaces $A_{IO}^{N^{(1)}}$ and $A_{IO}^{N^{(2)}}$, respectively, and thus consider the new set of $N{+}1$ parties $\N'=\{A_1,\ldots,A_{N-1},A_{N^{(1)}},A_{N^{(2)}}\}$. If $W \in A_{IO}^\N$ is a valid $N$-partite process matrix, one can then verify that when considering the set $\N'$, $W \in A_{IO}^{\N'}$ is also is a valid $(N{+}1)$-partite process matrix, i.e., $W \in \L^{\N'}$. If $W$ is compatible with a causal order $\K_1 \prec \K_2$, then $W$ is also compatible with $\K_1' \prec \K_2'$, where $\K_i'$ is obtained from $\K_i$ (like $\N'$ from $\N$) by replacing $A_N$ by $A_{N^{(1)}},A_{N^{(2)}}$ (when $A_N \in K_i$).%
\footnote{Both properties are straightforward when recalling that the validity and compatibility with a fixed order constraints are obtained by imposing certain conditions for all operations $M \in A_{IO}^N$ of party $A_N$: clearly, these constraints are also satisfied if $A_N$ is separated into two parties $A_{N^{(1)}}$ and $A_{N^{(2)}}$, and $M$ takes the form $M = M^{(1)} \otimes M^{(2)} \in A_{IO}^{N^{(1)}} \otimes A_{IO}^{N^{(2)}}$ (and by noting that if $M^{(1)}$ and $M^{(2)}$ are CPTP, then so is $M$).
These properties can also be verified formally using the characterisations of Eqs.~\eqref{eq:constr_LN_app} or~\eqref{eq:constr_LK1K2_K12}, after noting in particular that ${}_{[1-A_O^{N^{(1)}}A_O^{N^{(2)}}]\cdots}W = 0$ is equivalent to ${}_{[1-A_O^{N^{(1)}}]\cdots}W = {}_{[1-A_O^{N^{(2)}}]\cdots}W = 0$.}

Conversely, for a given set of $N\ge 2$ parties $\N$, let us finally consider a set $\N'$ obtained from $\N$ by now grouping two or more of the parties, e.g.\ $\N' = \{A_1, \ldots, A_{N-2}, \{A_{N-1},A_N\}\}$, where $\{A_{N-1},A_N\}$ is considered to form a new effective party.
Then $W$ is not necessarily a valid ($N{-}1$)-partite process matrix for the parties in $\N'$. The reason for this is that valid process matrices are required to give valid probability distributions only for product operations of the parties; if two parties are grouped together and perform a joint operation, that may no longer yield valid probabilities. An explicit counterexample is for instance $W = \frac12(\id^{A_OB_I} + \hat{\textsc{z}}^{A_O} \hat{\textsc{z}}^{B_I})$, which represents a (dephasing) channel from $A$ to $B$ and is indeed a valid bipartite process matrix, but not a valid single-partite process if $A$ and $B$ are grouped together (as ${}_{[1-A_O]}W \neq 0$; e.g., the joint CPTP map $M = \frac12(\id^{A_OB_I} - \hat{\textsc{z}}^{A_O} \hat{\textsc{z}}^{B_I})$ gives $\Tr[M \cdot W] = 0$).%
\footnote{Nevertheless, from Eqs.~\eqref{eq:LN_dAIf_1_gen} and~\eqref{eq:LN_dAOl_1_gen} one can see that parties who all have trivial incoming spaces, or parties who all have trivial outgoing spaces, can be grouped together without changing the validity of the process matrix in question.}

\subsection{Allowed and forbidden terms in the Hilbert–Schmidt basis decomposition of process matrices}
\label{app:allowed_forbidden_terms}

In Refs.~\cite{oreshkov12,oreshkov16}, the constraints characterising the set of valid process matrices and the set of process matrices compatible with a fixed causal order between two complementary subsets were formulated in a different way, namely by specifying which terms can appear in the decomposition of the corresponding operators in a Hilbert-Schmidt basis. To complete this appendix, we now establish the connection between these two alternative characterisations, and we prove their equivalence.

\medskip

A Hilbert-Schmidt basis of some space of linear operators $X$ (acting on a $d_X$-dimensional Hilbert space) is given by a set of generalised Pauli matrices, i.e., a set of Hermitian operators $\{\sigma_\mu^X \}_{\mu = 0}^{d_X^2-1}$, with $\sigma_0^X = \id^X$, $\Tr[\sigma_\mu^X \sigma_\nu^X] = d_X \delta_{\mu,\nu}$ for all $\mu,\nu = 0,\dots,d_X^2 - 1$, and $\Tr[\sigma_\mu^X] = 0$ for $\mu \ge 1$. In such a basis, a process matrix $W \in A_{IO}^1 \otimes A_{IO}^2 \otimes A_{IO}^3 \otimes \cdots$ can be expanded as  
\begin{align}
W &= \hspace{-2mm} \sum_{\mu_1,\nu_1,\mu_2,\nu_2,\ldots} \hspace{-3mm} w_{\mu_1\nu_1\mu_2\nu_2\mu_3\nu_3\cdots} \, \sigma_{\mu_1}^{A_I^1} \sigma_{\nu_1}^{A_O^1}  \sigma_{\mu_2}^{A_I^2}  \sigma_{\nu_2}^{A_O^2} \sigma_{\mu_3}^{A_I^3}  \sigma_{\nu_3}^{A_O^3} \! \cdots \notag \\
& \text{with} \ w_{\mu_1\nu_1\mu_2\nu_2\mu_3\nu_3\cdots}  \in  \mathbb{R} \quad \forall \ \mu_1,\nu_1,\mu_2,\nu_2,\mu_3,\nu_3,\ldots. \label{eq:HS_decomp}
\end{align}
The approach of Refs.~\cite{oreshkov12,oreshkov16} looks at which terms $\sigma_{\mu_1}^{A_I^1} \sigma_{\nu_1}^{A_O^1}  \sigma_{\mu_2}^{A_I^2}  \sigma_{\nu_2}^{A_O^2} \sigma_{\mu_3}^{A_I^3}  \sigma_{\nu_3}^{A_O^3} \cdots$ can appear with a nonzero coefficient $w_{\mu_1\nu_1\mu_2\nu_2\mu_3\nu_3\cdots}$ (i.e., are ``allowed'') in the above decomposition.
According to Proposition~3.1 of Ref.~\cite{oreshkov16}, a Hermitian operator $W$ is in the linear subspace $\L^\N$ of valid process matrices if and only if, in addition to the identity term $\id^\N$, it contains only terms for which at least one party $A_k$ has a nontrivial operator $\sigma_{\mu_k} \neq \id$ on $A_I^k$ and the identity operator $\id$ on $A_O^k$.%
\footnote{As clarified in Ref.~\cite{milz18}, valid process matrix can indeed only contain terms that, except for the identity, do not appear in the Hilbert-Schmidt decomposition of $\bigotimes_{k \in \N} M_{x_k}^{A_{IO}^k}$, for any CPTP maps $M_{x_k}^{A_{IO}^k}$ (as otherwise it is always possible to find some CPTP maps $M_{x_k}^{A_{IO}^k}$ that give non-normalised probabilities via the generalised Born rule of Eq.~\eqref{eq:born_rule_Npartite}). Given the constraint $\Tr_{A_O^k} M_{x_k}^{A_{IO}^k} = \id^{A_I^k}$ for CPTP maps, one can see that the only forbidden terms in any $M_{x_k}^{A_{IO}^k}$ are of the form $\sigma_{\mu_k}^{A_I^k}\id^{A_O^k}$ with $\sigma_{\mu_k} \neq \id$. Thus, only terms that contain $\sigma_{\mu_k}^{A_I^k}\id^{A_O^k}$ for at least one value of $k$ cannot appear in the Hilbert-Schmidt decomposition of $\bigotimes_{k \in \N} M_{x_k}^{A_{IO}^k}$, and are thus allowed (in addition to the identity) in the decomposition of a process matrix.}

To see that this statement is indeed equivalent to our own characterisation of $\L^\N$, let us first verify that all terms of this kind fulfil all the constraints of Eq.~\eqref{eq:constr_LN_app}. This is clearly the case for the identity $\id^\N$, since ${}_{[1-A_O^i]}\id^\N = 0$ for any party $A_i$. Consider then some generic Hilbert-Schmidt term $T_k$ of the form $\cdots \sigma_{\mu_k}^{A_I^k}\id^{A_O^k} \cdots$ (with $\sigma_{\mu_k} \neq \id$). Such a term satisfies ${}_{[1-A_O^k]}T_k = {}_{A_{IO}^k}T_k = 0$, so that for any $\X \subseteq \N, \X \neq \emptyset$ we have ${}_{\prod_{i \in \X}[1-A_O^i]A_{IO}^{\N \backslash \X}}T_k = 0$, whether $k \in \X$ or $k \in \N \backslash \X$. By linearity, any operator $W$ whose Hilbert-Schmidt decomposition~\eqref{eq:HS_decomp} only contains the identity or such terms $T_k$ thus satisfies all the constraints~\eqref{eq:constr_LN_app}.
Conversely, suppose that the Hilbert-Schmidt decomposition of $W$ contains a term $F$ (with a nonzero weight) that is ``forbidden'' according to Proposition~3.1 of Ref.~\cite{oreshkov16}, that is, a term such that for all parties $A_k$, there is either a nontrivial operator $\sigma_{\nu_k} \neq \id$ on $A_O^k$, or an identity operator on both $A_I^k$ and $A_O^k$ (and where there is at least one party for which the former is true). Consider then the nonempty subset $\X \subseteq \N$ of parties $A_i$ for which $\sigma_{\nu_i}^{A_O^i} \neq \id^{A_O^i}$ in $F$. For $i \in \X$, one thus has ${}_{[1-A_O^i]}F = F$, while for $j \in \N \backslash \X$, ${}_{A_{IO}^j}F = F$; all in all, ${}_{\prod_{i \in \X}[1-A_O^i]A_{IO}^{\N \backslash \X}}F = F$. By the linear independence of all Hilbert-Schmidt terms, $W$ then cannot satisfy ${}_{\prod_{i \in \X}[1-A_O^i]A_{IO}^{\N \backslash \X}}W = 0$, and thus violates the constraints of Eq.~\eqref{eq:constr_LN_app}.

\medskip

The process matrices that are compatible with the causal order $\K_1 \prec \K_2$, with $\K_1 \cup \K_2 = \N$, were likewise characterised in Ref.~\cite{oreshkov16} in terms of allowed terms in a Hilbert-Schmidt basis decomposition. The following terminology was used: the \emph{restriction} of a Hilbert-Schmidt term onto certain subsystems is the part of the term corresponding to the respective subsystems---for example, the restriction of the term $\sigma_{\mu_1}^{A_I^1} \sigma_{\nu_1}^{A_O^1}  \sigma_{\mu_2}^{A_I^2}  \sigma_{\nu_2}^{A_O^2} \sigma_{\mu_3}^{A_I^3}  \sigma_{\nu_3}^{A_O^3} \cdots$ onto the subsystems $A_I^2 \otimes A_O^2$ is just $\sigma_{\mu_2}^{A_I^2}  \sigma_{\nu_2}^{A_O^2}$.
According to Proposition~3.2 in Ref.~\cite{oreshkov16}, the (valid) process matrices that do not allow signalling from $\K_2$ to $\K_1$ are those that contain only Hilbert-Schmidt terms whose restriction to the (incoming and outgoing systems of) parties in $\K_2$ are of the allowed type for a $|\K_2|$-partite process matrix for those parties---that is, terms with either the identity operator $\id^{A_{I}^{k_2}}\id^{A_{O}^{k_2}}$ for all parties in $\K_2$, or for which there is some party $A_{k_2} \in \K_2$ with a nontrivial generalised Pauli operator $\sigma_{\mu_{k_2}} \neq \id$ on $A_I^{k_2}$ and the identity operator on $A_O^{k_2}$.

To see that this proposition is indeed equivalent to our characterisation of $\L^{\K_1 \prec \K_2}$ given by Eqs.~\eqref{eq:constr_LK1K2} or Eq.~\eqref{eq:constr_LK1K2_v2}, note that the restriction of a Hilbert-Schmidt term $T$ to $\K_2$ is precisely obtained, up to a multiplicative factor (which may be $0$), by taking the partial trace $T_{|M_{\K_1}} \coloneqq \Tr_{\K_1} [M_{\K_1} \otimes \id^{\K_2} \cdot T]$, for any $M_{\K_1} \in A_{IO}^{\K_1}$. Hence, imposing that all Hilbert-Schmidt terms $T$ in the decomposition of $W = \sum_T w_T \, T$ have restrictions to $\K_2$ that are in $\L^{\K_2}$, as in the characterisation of Ref.~\cite{oreshkov16} just recalled, is equivalent to imposing that for any $M_{\K_1} \in A_{IO}^{\K_1}$, $W_{|M_{\K_1}} = \sum_T w_T \, T_{|M_{\K_1}}$ only have terms $T_{|M_{\K_1}} \in \L^{\K_2}$, i.e., that $W_{|M_{\K_1}}$ itself is in $\L^{\K_2}$, as imposed in Eq.~\eqref{eq:constr_LK1K2_v2}.
Note that Proposition~3.2 in Ref.~\cite{oreshkov16} pre-supposed that the process matrix under consideration was valid. If this is not pre-supposed, one must in addition impose, according to the previous characterisation, that for Hilbert-Schmidt terms in the decomposition of $W$ whose restriction to $\K_2$ is the identity operator $\id^{\K_2}$, there must either also be the identity operator for all parties in $\K_1$, or there must be some party $A_{k_1} \in \K_1$ with a restriction $\sigma_{\mu_{k_1}}^{A_I^{k_1}}\id^{A_O^{k_1}} \neq \id^{A_I^{k_1}}\id^{A_O^{k_1}}$---in other words, one must impose that $\Tr_{\K_2}W \in \L^{\K_1}$, so that one also recovers the first constraint of Eq.~\eqref{eq:constr_LK1K2_v2}.

\section{Characterisation of causally separable process matrices}
\label{app:characterisation_CS_Ws}

\renewcommand{\thelemma}{B\arabic{lemma}}
\setcounter{lemma}{0}
\renewcommand{\theproposition}{B\arabic{proposition}}
\setcounter{proposition}{0}

In this appendix we prove the propositions that allow us to characterise causally separable process matrices in terms of simple conditions. 
We start by proving the first part of Proposition~\ref{prop:comp_3_defs}, namely that in the particular tripartite scenario with $d_{C_O} = 1$, Ara\'{u}jo \emph{et al.}'s definition of causal separability (Definition~\ref{def:AB+-CS}) is equivalent to OG's notion of extensible causal separability (Definition~\ref{def:OG-ECS}), and thus also to our Definition~\ref{def:our_def-CS}. 
Then we provide the proofs for the characterisation of general tripartite causally separable process matrices (Proposition~\ref{prop:CS_charact_3}) as well as for the necessary condition (Proposition~\ref{prop:CS_necessary}) and the sufficient condition (Proposition~\ref{prop:CS_sufficient}) in the general $N$-partite case.
Note that all the special cases follow from Propositions~\ref{prop:CS_necessary} and~\ref{prop:CS_sufficient}, and we could just give the proofs of those two general propositions. 
However, for pedagogical reasons we start with the simpler proofs, which may entail some repetition in the arguments, but allows for greater clarity in presenting the core ideas.

All of the proofs below (of increasing complexity) make use of the same type of argument to prove the necessity of the respective conditions. This argument is based on the ``teleportation technique'' that follows from the lemma below.
Before stating it, let us introduce some further notation. 
For two Hilbert spaces $\HS^{X}$, $\HS^{X'}$ with the same dimension $d$, and denoting by $\{\ket{i}^{X^{(\prime)}}\}_{i=1}^d$ an orthonormal basis of either $\HS^{X}$ or $\HS^{X'}$, we will consider the maximally entangled state $\ket{\Phi^+}^{X/X'} \coloneqq \frac{1}{\sqrt{d}} \sum_{i} \ket{i}^X \otimes \ket{i}^{X'}$.
We also recall that for a given matrix $W \in A_{IO}^\N$, we denote by $W^{A_{IO}^k \to A_{I'}^{k'}}$ the matrix in $(\bigotimes_{j\in \N \backslash k} A_{IO}^j) \otimes A_{I'}^{k'}$ that has formally the same form as $W$, except that party $A_k$'s system $A_{IO}^k$ is now attributed to an extension $A_{I'}^{k'}$ of party $A_{k'}$'s incoming space.
Formally (recalling Eq.~\eqref{eq:W_A_telep_to_A'}),
\begin{equation}
W^{A_{IO}^k \to A_{I'}^{k'}} \coloneqq \sum_{i,j} \Tr_k \! \Big[ \!\ket{i}\!\!\bra{j}^{A_{IO}^k} \otimes \id^{\N \backslash k} \cdot W \Big] \otimes \ket{j}\!\!\bra{i}^{A_{I'}^{k'}}, \label{eq:W_A_telep_to_A'_app}
\end{equation}
where $\{\ket{i}\}$ is an orthonormal basis of $\HS^{A_I^k} \otimes \HS^{A_O^k}$.

\begin{lemma}[``Teleportation technique''] \label{lemma:teleportation}
Consider a process matrix $W \in \L^\N$, to which one attaches a maximally entangled state $\ketbra{\Phi^+}^{A_{I''}^k/A_{I'}^{k'}}$ shared by parties $A_k$ and $A_{k'}$, with dimensions $d_{A_{I''}^k} = d_{A_{I'}^{k'}} = d_{A_{IO}^k}$, and possibly some other ancillary state $\tilde\rho$ in some further extension $A_{I'}^{\N \backslash k'} \otimes A_{I''}^{k'}$. Consider then the case where party $A_k$ applies the CP map represented by the positive semidefinite CJ matrix $M_k = \ketbra{\Phi^+}^{A_{IO}^k/A_{I''}^k} \otimes \id^{A_{I'}^k}$. The resulting conditional matrix for the other $N$-1 parties is then
\begin{align}
& (W \otimes \ketbra{\Phi^+}^{A_{I''}^k/A_{I'}^{k'}} \otimes \tilde\rho)_{|M_k = \ketbra{\Phi^+}^{A_{IO}^k/A_{I''}^k} \otimes \id^{A_{I'}^k}} \notag \\[1mm]
& \coloneqq \Tr_k \Big[ \ketbra{\Phi^+}^{A_{IO}^k/A_{I''}^k} \otimes \id^{A_{I'}^k} \otimes \id^{\N \backslash k} \notag \\[-3mm]
& \hspace{35mm} \cdot W \otimes \ketbra{\Phi^+}^{A_{I''}^k/A_{I'}^{k'}} \otimes \tilde\rho \Big] \notag \\
& \, = \frac{1}{(d_{A_{IO}^k})^2} \ W^{A_{IO}^k \to A_{I'}^{k'}}  \, \otimes \, \Tr_k[\tilde\rho].
\end{align}
\end{lemma}

\begin{widetext}
\begin{proof}
For clarity, let us write explicitly as superscripts the spaces in which the various operators act. We have:
\begin{align}
& (W^{A_{IO}^\N} \otimes \ketbra{\Phi^+}^{A_{I''}^k/A_{I'}^{k'}} \otimes \tilde\rho^{A_{I'}^{\N \backslash k'}A_{I''}^{k'}})_{|M_k = \ketbra{\Phi^+}^{A_{IO}^k/A_{I''}^k} \otimes \id^{A_{I'}^k}} \notag \\
& \ = \Tr_{A_{II'I''O}^k} \Big[ \ketbra{\Phi^+}^{A_{IO}^k/A_{I''}^k} \otimes \id^{A_{I'}^k} \otimes \id^{A_{II'O}^{\N \backslash k}A_{I''}^{k'}} \cdot W^{A_{IO}^\N} \otimes \ketbra{\Phi^+}^{A_{I''}^k/A_{I'}^{k'}} \otimes \tilde\rho^{A_{I'}^{\N \backslash k'}A_{I''}^{k'}} \Big] \notag \\
& \ = {\textstyle \frac{1}{(d_{A_{IO}^k})^2}} \sum_{i,i',j,j'} \Tr_{A_{II''O}^k} \Big[ \ket{i}\!\!\bra{i'}^{A_{IO}^k} \otimes \ket{i}\!\!\bra{i'}^{A_{I''}^k} \otimes \id^{A_{IO}^{\N \backslash k}} \otimes \id^{A_{I'}^{k'}} \cdot W^{A_{IO}^\N} \otimes \ket{j}\!\!\bra{j'}^{A_{I''}^k} \otimes \ket{j}\!\!\bra{j'}^{A_{I'}^{k'}} \Big] \otimes \Tr_{A_{I'}^k} \Big[  \tilde\rho^{A_{I'}^{\N \backslash k'}A_{I''}^{k'}} \Big] \notag \\
& \ = {\textstyle \frac{1}{(d_{A_{IO}^k})^2}} \sum_{i,j} \Tr_{A_{IO}^k} \Big[ \ket{i}\!\!\bra{j}^{A_{IO}^k} \otimes \id^{A_{IO}^{\N \backslash k}} \cdot W^{A_{IO}^\N} \Big] \otimes \ket{j}\!\!\bra{i}^{A_{I'}^{k'}} \otimes \Tr_{A_{I'}^k} \Big[  \tilde\rho^{A_{I'}^{\N \backslash k'}A_{I''}^{k'}} \Big] \notag \\
& \ = {\textstyle \frac{1}{(d_{A_{IO}^k})^2}} \ W^{A_{IO}^k \to A_{I'}^{k'}} \, \otimes \, \Tr_k[\tilde\rho] \,.
\end{align}
\end{proof}
\end{widetext}

We shall also use the following facts in (some of) the proofs below:

\begin{proposition} \label{prop:factor_rho}
Without loss of generality, each $W_{(k)}^\rho$ in Definition~\eqref{def:our_def-CS} can be taken to be of the form $W_{(k)} \otimes \rho$. Eq.~\eqref{eq:our_def-CS} then implies the direct decomposition $W = \sum_{k \in \N} q_k \, W_{(k)}$, with each $W_{(k)} \in A_{IO}^\N$ being a process matrix compatible with party $A_k$ acting first (and such that for any CP map $M_k \in A_{II'O}^k$, $(W_{(k)}\otimes\rho)_{|M_k}$ is causally separable).
\end{proposition}

\begin{proof}
If $\rho$ is pure, then from the extremality of pure states it follows that $W_{(k)}^\rho = W_{(k)} \otimes \rho$. If $\rho$ is mixed one can first purify it by introducing an additional incoming system for some arbitrary party, obtain the appropriate decomposition~\eqref{eq:our_def-CS} for its purification, and then trace out the additional incoming space just introduced to reach the same conclusion. As $W_{(k)}^\rho$ is compatible with $A_k$ acting first and $W_{(k)} = \Tr_{A_{I'}^\N} W_{(k)}^\rho$, then $W_{(k)}$ itself must also be compatible with $A_k$ acting first (see remarks in Appendix~\ref{app:trace_out}).
\end{proof}

\begin{proposition} \label{prop:traceout}
In a scenario where the parties' incoming spaces are decomposed as $A_I^\N \otimes A_{I'}^\N$ (possibly with some trivial spaces $A_I^k$ or $A_{I'}^k$), if a process matrix $W \in A_{II'O}^\N$ is causally separable, then so is $\Tr_{I'} W \in A_{IO}^\N$ (with $\Tr_{I'} \coloneqq \Tr_{A_{I'}^\N}$).
\end{proposition}

\begin{proof}
For $N=1$ party any process matrix is by definition causally separable, so that the claim is trivial.

Suppose the claim holds true in the $(N{-}1)$-partite case.
If $W \in A_{II'O}^\N$ is causally separable then by Definition~\ref{def:our_def-CS}, for any extension $A_{I''}^\N$ of the parties' incoming systems and any ancillary quantum state $\rho \in A_{I''}^\N$, $W \otimes \rho$ has a decomposition of the form $W \otimes \rho = \sum_k q_k W_{(k)}^\rho$ with each $W_{(k)}^\rho$ a valid process matrix compatible with party $A_k$ first, and such that for any possible CP map $M_k' \in A_{II'I''O}^k$ applied by party $A_k$, the conditional $(N{-}1)$-partite process matrix $(W_{(k)}^\rho)_{|M_k'} \coloneqq \Tr_k [M_k' \otimes \id^{\N \backslash k} \, \cdot \, W_{(k)}^\rho]$ is itself causally separable.

One then has $\Tr_{I'} W \otimes \rho = \sum_k q_k (\Tr_{I'} W_{(k)}^\rho)$, with $\Tr_{I'} W_{(k)}^\rho$ a valid process matrix compatible with party $A_k$ first (see remarks in Appendix~\ref{app:trace_out}). For any possible CP map $M_k \in A_{II''O}^k$ applied by party $A_k$, one has $(\Tr_{I'} W_{(k)}^\rho)_{|M_k} \coloneqq \Tr_k [M_k \otimes \id^{\N \backslash k} \, \cdot \, (\Tr_{I'} W_{(k)}^\rho)] = \Tr_{I'} \Tr_k [M_k \otimes \id^{A_{I'}^k} \otimes \id^{\N \backslash k} \, \cdot \, W_{(k)}^\rho] = \Tr_{I'} \big[ (W_{(k)}^\rho)_{|M_k'=M_k \otimes \id^{A_{I'}^k}} \big]$. As stated above, $(W_{(k)}^\rho)_{|M_k'=M_k \otimes \id^{A_{I'}^k}}$ is causally separable, and by the induction hypothesis so is $\Tr_{I'} \big[ (W_{(k)}^\rho)_{|M_k'=M_k \otimes \id^{A_{I'}^k}} \big] = (\Tr_{I'} W_{(k)}^\rho)_{|M_k}$. We thus have a valid causally separable decomposition of $\Tr_{I'} W \otimes \rho$ for any extension $\rho$, which proves that $\Tr_{I'} W$ is causally separable, and which thus proves, by induction, the claim of Proposition~\ref{prop:traceout}.
\end{proof}

Again, this property is quite intuitive: clearly, discarding some parts of the incoming systems cannot induce some causal nonseparability where there was none previously. As for the similar statements for valid process matrices and for process matrices compatible with a fixed causal order discussed in Appendix~\ref{app:trace_out}, the converse is not necessarily true: if $\Tr_{I'} W$ is a causally separable process matrix, then $W$ may not necessarily be causally separable%
\footnote{As a counterexample, consider some causally nonseparable bipartite process matrix $W \in A_{I'O} \otimes B_{IO}$. The process matrix $\Tr_{A_{I'}}W \in A_O \otimes B_{IO}$ is then compatible with the order $A \prec B$ and thus causally separable (see Appendix~\ref{app:fixed_order_part_cases}), although $W$ is not.}%
---unless $W$ is of the product form $W = W_0^{A_{IO}^\N} \otimes \rho^{A_{I'}^\N}$, in which case by our Definition~\ref{def:our_def-CS} if $W_0$ is a causally separable process matrix then so is $W = W_0 \otimes \rho$.

\subsection{Tripartite causally separable process matrices}
\label{app:charact_CS_3}

\subsubsection{Particular tripartite case with $d_{C_O} = 1$}
\label{app:charact_CS_3_dCO1}

Let us start by considering the tripartite case where party $C$ has no outgoing system (or equivalently, a trivial outgoing system, i.e., $d_{C_O} = 1$). The following proposition directly implies (after proper re-normalisation with appropriate weights $q$, $1{-}q$) the first part of Proposition~\ref{prop:comp_3_defs}, namely the equivalence in that case between Ara\'ujo \emph{et al.}'s causal separability and OG's extensible causal separability (which, we recall, is what we simply call causal separability here).

\begin{proposition}[Characterisation of tripartite causally separable process matrices with $d_{C_O} = 1$] \label{prop:CS_charact_3_no_CO}
In a tripartite scenario where party $C$ has no outgoing system, a matrix $W \in A_{IO} \otimes B_{IO} \otimes C_I$ is a valid tripartite causally separable process matrix (as per Definition~\ref{def:our_def-CS}) if and only if it can be decomposed as
\begin{equation}
W = W_{(A,B,C)} + W_{(B,A,C)} \label{eq:decomp_CS_W_3_no_CO}
\end{equation}
where, for each permutation $(X,Y)$ of the two parties $A$ and $B$, $W_{(X,Y,C)}$ is a positive semidefinite matrix satisfying
\begin{equation}
{}_{[1-X_O]Y_{IO}C_I} W_{(X,Y,C)} = 0 , \ {}_{[1-Y_O]C_I}W_{(X,Y,C)} = 0 \label{eq:decomp_CS_W_3_no_CO_constr}
\end{equation}
(i.e., $W_{(X,Y,C)}$ is a valid process matrix compatible with the causal order $X \prec Y \prec C$).
\end{proposition}

\begin{proof}
Consider a causally separable process matrix $W \in A_{IO} B_{IO} C_I$.
Let us then introduce an extension $A_{I'} \otimes C_{I''}$ of parties $A$ and $C$'s incoming spaces, of dimensions $d_{A_{I'}} = d_{C_{I''}} = d_{C_I}$, and consider attaching to $W$ the maximally entangled ancillary state $\rho = \ketbra{\Phi^+}^{C_{I''}/A_{I'}}$.

As $W$ is assumed to be causally separable, according to Definition~\ref{def:our_def-CS} and Proposition~\ref{prop:factor_rho} it must be decomposable as
\begin{equation}
W = W_{(A)} + W_{(B)} + W_{(C)} \,, \label{eq:decomp_W_A_B_C}
\end{equation}
where each term $W_{(X)} \in A_{IO} B_{IO} C_I$ is a (nonnormalised) process matrix compatible with party $X$ acting first, and such that whatever CP map that party applies to their share of $W_{(X)} \otimes \rho$, the resulting conditional process matrix for the other two parties is causally separable.

As $W_{(A)}$ is compatible both with $A$ first and with $C$ last (since $d_{C_O} = 1$, see Appendix~\ref{app:fixed_order_part_cases}), it is compatible with the fixed causal order $A \prec B \prec C$; formally, it satisfies ${}_{[1-A_O]B_{IO}C_I} W_{(A)} = {}_{[1-B_O]C_I} W_{(A)} = 0$, see Eq.~\eqref{eq:constr_fixed_order_N}. Similarly, $W_{(B)}$ is compatible with the order $B \prec A \prec C$.

Consider now the term $W_{(C)}$. Letting party $C$ act first on $W_{(C)} \otimes \rho = W_{(C)} \otimes \ketbra{\Phi^+}^{C_{I''}/A_{I'}}$ and project his incoming systems onto the maximally entangled state $\ket{\Phi^+}^{C_I/C_I''}$, according to Lemma~\ref{lemma:teleportation} (with a trivial extra ancillary state $\tilde\rho$), parties $A$ and $B$ are then left with the conditional process matrix
\begin{equation}
(W_{(C)} \otimes \rho)_{|M_C=\ketbra{\Phi^+}^{C_I/C_I''}} \ \propto \ W_{(C)}^{C_I \to A_{I'}} \,.
\end{equation}
By assumption this conditional process matrix must be a (bipartite) causally separable process matrix: according to Proposition~\ref{prop:CS_charact_2}, there must therefore exist a decomposition for $W_{(C)}^{C_I \to A_{I'}}$ of the form
\begin{equation}
W_{(C)}^{C_I \to A_{I'}} = W_{(C),(A,B)}^{C_I \to A_{I'}} + W_{(C),(B,A)}^{C_I \to A_{I'}} \label{eq:decomp_WC_CItoAI}
\end{equation}
with $W_{(C),(A,B)}^{C_I \to A_{I'}}, W_{(C),(B,A)}^{C_I \to A_{I'}} \in A_{II'O}B_{IO}$ two process matrices compatible with the fixed orders $A \prec B$ and $B \prec A$, respectively---i.e., satisfying
\begin{align}
{}_{[1-A_O]B_{IO}} W_{(C),(A,B)}^{C_I \to A_{I'}} = {}_{[1-B_O]} W_{(C),(A,B)}^{C_I \to A_{I'}} &= 0 \,, \notag \\
{}_{[1-B_O]A_{II'O}} W_{(C),(B,A)}^{C_I \to A_{I'}} = {}_{[1-A_O]} W_{(C),(B,A)}^{C_I \to A_{I'}} &= 0 \,. \label{eq:constr_WC_CItoAI_AtoB_BtoA}
\end{align}

Recall now that $W_{(C)}^{C_I \to A_{I'}}$ is formally the same matrix as $W_{(C)}$, except that system $C_I$ is replaced by $A_{I'}$.
Changing back $A_{I'}$ into $C_I$ in Eq.~\eqref{eq:decomp_WC_CItoAI}, we obtain the decomposition 
\begin{equation}
W_{(C)} = W_{(C),(A,B)} + W_{(C),(B,A)} \label{eq:decomp_W_C}
\end{equation}
with two positive semidefinite matrices $W_{(C),(A,B)}$, $W_{(C),(B,A)} \in A_{IO}B_{IO}C_I$ satisfying
\begin{align}
 {}_{[1-A_O]B_{IO}} W_{(C),(A,B)} = {}_{[1-B_O]} W_{(C),(A,B)} &= 0 \,, \notag \\
 {}_{[1-B_O]A_{IO}C_I} W_{(C),(B,A)} = {}_{[1-A_O]} W_{(C),(B,A)} &= 0 \,, \label{eq:constr_WC_AtoB_BtoA}
\end{align}
as implied by Eq.~\eqref{eq:constr_WC_CItoAI_AtoB_BtoA} after replacing $A_{I'}$ by $C_I$.
These constraints further imply that
\begin{align}
 \hspace{-2mm} {}_{[1-A_O]B_{IO}C_I} W_{(C),(A,B)} = {}_{[1-B_O]C_I} W_{(C),(A,B)} &= 0 \,, \notag \\
 \hspace{-2mm} {}_{[1-B_O]A_{IO}C_I} W_{(C),(B,A)} = {}_{[1-A_O]C_I} W_{(C),(B,A)} &= 0 \,,  
\end{align}
i.e., that $W_{(C),(A,B)}$ and $W_{(C),(B,A)}$ are process matrices compatible with the fixed causal orders $A \prec B \prec C$ and $B \prec A \prec C$, respectively (see Eq.~\eqref{eq:constr_fixed_order_N}).

From Eqs.~\eqref{eq:decomp_W_A_B_C} and~\eqref{eq:decomp_W_C}, and by defining $W_{(A,B,C)} \coloneqq W_{(A)} + W_{(C),(A,B)} \ge 0$ and $W_{(B,A,C)} \coloneqq W_{(B)} + W_{(C),(B,A)} \ge 0$, we thus find that $W$ indeed has a decomposition of the form of Eq.~\eqref{eq:decomp_CS_W_3_no_CO}, with each term satisfying the constraints of Eq.~\eqref{eq:decomp_CS_W_3_no_CO_constr}.

\medskip

Conversely, any process matrix $W$ that can be decomposed as in Eq.~\eqref{eq:decomp_CS_W_3_no_CO}, with process matrices $W_{(A,B,C)}$ and $W_{(B,A,C)}$ satisfying the constraints of Eq.~\eqref{eq:decomp_CS_W_3_no_CO_constr}---i.e., being compatible with the causal orders $A \prec B \prec C$ and $B \prec A \prec C$---is clearly causally separable, which concludes the proof of Proposition~\ref{prop:CS_charact_3_no_CO}.
\end{proof}

\subsubsection{General tripartite causally separable process matrices}
\label{app:charact_CS_3_general}

We now turn to proving Proposition~\ref{prop:CS_charact_3}, which characterises causal separability in the general tripartite scenario where all three parties have nontrivial incoming and outgoing systems.

\begin{proof}
Consider a causally separable process matrix $W \in A_{IO}B_{IO}C_{IO}$. Let us introduce here an extension $A_{I'} \otimes A_{I''} \otimes B_{I'} \otimes B_{I''} \otimes C_{I'} \otimes C_{I''}$ of all three parties' incoming spaces, with dimensions $d_{A_{I''}} = d_{B_{I'}} = d_{A_{IO}}$, $d_{B_{I''}} = d_{C_{I'}} = d_{B_{IO}}$ and $d_{C_{I''}} = d_{A_{I'}} = d_{C_{IO}}$, and let us attach to $W$ the state $\rho = \ketbra{\Phi^+}^{A_{I''}/B_{I'}} \otimes \ketbra{\Phi^+}^{B_{I''}/C_{I'}} \otimes \ketbra{\Phi^+}^{C_{I''}/A_{I'}}$.

According to Definition~\ref{def:our_def-CS} and Proposition~\ref{prop:factor_rho}, $W$ must be decomposable as
\begin{equation}
W = W_{(A)} + W_{(B)} + W_{(C)} \,, \label{eq:decomp_Wrho_3}
\end{equation}
where each term $W_{(X)} \in A_{IO} B_{IO} C_{IO}$ is a process matrix compatible with party $X$ acting first---so that it satisfies in particular ${}_{[1-X_O]Y_{IO}Z_{IO}}W_{(X)} = 0$ (with $X \neq Y \neq Z$, see Eq.~\eqref {eq:constr_L_k_first_app}), as in Eq.~\eqref{eq:decomp_CS_W_3_constrWX}---and such that whatever that party does on $W_{(X)} \otimes \rho$, the resulting conditional process matrix for the other two parties is causally separable.

Consider the first term in Eq.~\eqref{eq:decomp_Wrho_3}. Letting party $A$ act first on $W_{(A)} \otimes \rho$ and perform the operation described by the CJ operator $M_A = \ketbra{\Phi^+}^{A_{IO}/A_{I''}} \otimes \id^{A_{I'}} \ge 0$, we find, using Lemma~\ref{lemma:teleportation} (with $\tilde\rho = \Tr_{A_{I''}B_{I'}}[\rho]$), that the remaining parties $B,C$ are left with the conditional process matrix
\begin{align}
(W_{(A)} \otimes \rho)_{|M_A = \ketbra{\Phi^+}^{A_{IO}/A_{I''}} \otimes \id^{A_{I'}}} \qquad\qquad \notag \\
\propto \ W_{(A)}^{A_{IO} \to B_{I'}} \otimes \Tr_{A_{I'I''}B_I'}[\rho]\,.
\end{align}
By assumption this conditional process matrix---and therefore $W_{(A)}^{A_{IO} \to B_{I'}}$ itself (according to Proposition~\ref{prop:traceout}, after tracing out $\Tr_{A_{I'I''}B_I'}[\rho]$ completely)---must be a (bipartite) causally separable process matrix: there must therefore exist a decomposition of the form
\begin{equation}
W_{(A)}^{A_{IO} \to B_{I'}} = W_{(A,B,C)}^{A_{IO} \to B_{I'}} + W_{(A,C,B)}^{A_{IO} \to B_{I'}}
\end{equation}
where $W_{(A,B,C)}^{A_{IO} \to B_{I'}}, W_{(A,C,B)}^{A_{IO} \to B_{I'}} \in B_{II'O}C_{IO}$ are bipartite processes compatible with the causal orders $B \prec C$ and $C \prec B$, respectively. After re-attributing the system $B_{I'}$ to $A_{IO}$, we obtain a decomposition for $W_{(A)}$,
\begin{equation}
W_{(A)} = W_{(A,B,C)} + W_{(A,C,B)},
\end{equation}
with the positive semidefinite matrices $W_{(A,B,C)}$, $W_{(A,C,B)} \in A_{IO}B_{IO}C_{IO}$ satisfying the following constraints, obtained (as we did in the previous subsection) after replacing $B_{I'}$ by $A_{IO}$ in the constraints satisfied by $W_{(A,B,C)}^{A_{IO} \to B_{I'}}$ and $W_{(A,C,B)}^{A_{IO} \to B_{I'}}$:
\begin{align}
 {}_{[1-B_O]C_{IO}} W_{(A,B,C)} = {}_{[1-C_O]} W_{(A,B,C)} &= 0 \,, \notag \\
 {}_{[1-C_O]A_{IO}B_{IO}} W_{(A,C,B)} = {}_{[1-B_O]} W_{(A,C,B)} &= 0 \,. \label{eq:WABC_WACB_constr}
\end{align}

Furthermore, since $W_{(A)}$ is compatible with $A$ acting first it satisfies in particular ${}_{[1-C_O]B_{IO}} W_{(A)} = 0$ (see Eq.~\eqref{eq:constr_L_k_first_app}), and because of Eq.~\eqref{eq:WABC_WACB_constr}, we also have ${}_{[1-C_O]B_{IO}} W_{(A,B,C)} = 0$. 
Given that $W_{(A,C,B)} = W_{(A)} - W_{(A,B,C)}$, we have ${}_{[1-C_O]B_{IO}} W_{(A,C,B)} = 0$ as well.%
\footnote{Note that this is the step where the tripartite proof does not generalise straightforwardly to $N \ge 4$ parties. 
In particular, we cannot use the same argument to prove that the constraints~\eqref{eq:explicit_NC_4partite} that appear in our necessary condition are satisfied without tracing out the $X_{IO}$ on the first and fourth lines (as one would need if the terms in Eq.~\eqref{eq:explicit_NC_4partite_2} were to satisfy Eq.~\eqref{eq:decomp_CS_suff_cond} and thus specify a decomposition satisfying also our sufficient condition for causal separability). 
One indeed obtains e.g.\ ${}_{[1-Z_O]Y_{IO}T_{IO}}W_{(X,Z)}^{^{\scriptscriptstyle [X \to Y]}} = {}_{[1-Z_O]Y_{IO}T_{IO}}\big(W_{(X)} - W_{(X,Y,Z,T)}^{^{\scriptscriptstyle [X \to Y]}} - W_{(X,Y,T,Z)}^{^{\scriptscriptstyle [X \to Y]}} - W_{(X,T,Y,Z)}^{^{\scriptscriptstyle [X \to Y]}} - W_{(X,T,Z,Y)}^{^{\scriptscriptstyle [X \to Y]}}\big) = -\big({}_{[1-Z_O]Y_{IO}T_{IO}}W_{(X,T,Z,Y)}^{^{\scriptscriptstyle [X \to Y]}}\big)$ 
which, \emph{a priori}, may still be nonzero.}
Together with Eq.~\eqref{eq:WABC_WACB_constr}, we thus find that all constraints of Eq.~\eqref{eq:decomp_CS_W_3_constrWXYZ} for $X=A$ are satisfied. One can similarly show that they are satisfied for $X = B,C$, which proves (since we noted before that Eq.~\eqref{eq:decomp_CS_W_3_constrWX} is also satisfied) that the decomposition of Proposition~\ref{prop:CS_charact_3} is indeed a necessary condition for any causally separable process matrix $W$.

\medskip

Conversely, suppose a matrix $W \in A_{IO} \otimes B_{IO} \otimes C_{IO}$ has a decomposition of the form~\eqref{eq:decomp_CS_W_3} that satisfies Eqs.~\eqref{eq:decomp_CS_W_3_constrWX}--\eqref{eq:decomp_CS_W_3_constrWXYZ}. Then as we noted right after Proposition~\ref{prop:CS_charact_3}, each term $W_{(X)}$ is a valid process matrix, compatible with party $X$ acting first.
For any CP map $M_X$ applied by party $X$ on its share of $W_{(X)}$, the resulting conditional process matrix for the other two parties $Y,Z$ is
\begin{align}
\hspace{-1mm} (W_{(X)})_{|M_X} \! &= \Tr_X[M_X \otimes \id^{Y_{IO}Z_{IO}} \cdot W_{(X)}] \notag \\
&= \Tr_X[M_X \!\otimes\! \id^{Y_{IO}Z_{IO}} \!\cdot\! (W_{(X,Y,Z)} \!+\! W_{(X,Z,Y)})] \notag \\
&= (W_{(X,Y,Z)})_{|M_X} + (W_{(X,Z,Y)})_{|M_X}
\end{align}
with $(W_{(X,Y,Z)})_{|M_X} \ge 0$ satisfying
\begin{align}
& {}_{[1-Y_O]Z_{IO}}[(W_{(X,Y,Z)})_{|M_X}] = [{}_{[1-Y_O]Z_{IO}}W_{(X,Y,Z)}]_{|M_X} = 0, \notag \\
& {}_{[1-Z_O]}[(W_{(X,Y,Z)})_{|M_X}] = [{}_{[1-Z_O]}W_{(X,Y,Z)}]_{|M_X} = 0
\end{align}
(and similarly for $(W_{(X,Z,Y)})_{|M_X}$), as follows from Eq.~\eqref{eq:decomp_CS_W_3_constrWXYZ}. This shows that $(W_{(X,Y,Z)})_{|M_X}$ and $(W_{(X,Z,Y)})_{|M_X}$ are valid bipartite process matrices compatible with the orders $Y \prec Z$ and $Z \prec Y$, respectively, so that $(W_{(X)})_{|M_X}$ is causally separable.
Note that for any ancillary state $\rho$, $W \otimes \rho$ also has a decomposition as in Eq.~\eqref{eq:decomp_CS_W_3}, obtained simply by attaching the ancillary state to every individual term in the decomposition of $W$. Therefore, the same reasoning as above applies, which implies that $W$ is causally separable.  
This thus shows that the decomposition of Proposition~\ref{prop:CS_charact_3} is also a sufficient condition for a matrix $W$ to represent a causally separable process matrix, which concludes the proof of that proposition.
\end{proof}

Let us mention here that Proposition~\ref{prop:CS_charact_3_no_CO}, for the particular tripartite case where $d_{C_O}=1$, could also be obtained as a corollary of the general tripartite case considered by Proposition~\ref{prop:CS_charact_3}.
Indeed, in the case where $d_{C_O} = 1$ the matrices $W_{(A)}$ and $W_{(B)}$ in Eq.~\eqref{eq:decomp_CS_W_3} are compatible with the fixed causal orders $A \prec B \prec C$ and $B \prec A \prec C$, respectively (as they are compatible with both $A$ or $B$ first, and $C$ last); furthermore, the matrix $W_{(C,A,B)}$ satisfies ${}_{[1-A_O]B_{IO}}W_{(C,A,B)} = {}_{[1-B_O]}W_{(C,A,B)} = 0$ and therefore ${}_{[1-A_O]B_{IO}C_I}W_{(C,A,B)} = {}_{[1-B_O]C_I}W_{(C,A,B)} = 0$, which implies that it is also compatible with $A \prec B \prec C$; and similarly, the matrix $W_{(C,B,A)}$ is also compatible with $B \prec A \prec C$. The decomposition of Eq.~\eqref{eq:decomp_CS_W_3} thus provides a decomposition in the form $W = \tilde W_{(A,B,C)} + \tilde W_{(B,A,C)}$ with $\tilde W_{(A,B,C)} = W_{(A)} + W_{(C,A,B)}$ and $\tilde W_{(B,A,C)} = W_{(B)} + W_{(C,B,A)}$ satisfying the constraints of Eq.~\eqref{eq:decomp_CS_W_3_no_CO_constr}.

\subsubsection{Particular tripartite case with $d_{A_I} = 1$}
\label{app:charact_CS_3_dAI1}

Another particular tripartite case of interest is one where one party, say now $A$, has no \emph{incoming} space (or a trivial one, with $d_{A_I} = 1$). The following characterisation is also obtained as a corollary of the general tripartite case above.

\begin{proposition}[Characterisation of tripartite causally separable process matrices with $d_{A_I} = 1$] \label{prop:CS_charact_3_no_AI}
In a tripartite scenario where party $A$ has no incoming system, a matrix $W \in A_{O} \otimes B_{IO} \otimes C_{IO}$ is a valid tripartite causally separable process matrix (as per Definition~\ref{def:our_def-CS}) if and only if
\begin{equation}
{}_{[1-A_O]B_{IO}C_{IO}}W = 0 \label{eq:decomp_CS_W_3_no_AI_constr0A}
\end{equation}
and $W$ can be decomposed as
\begin{equation}
W = W_{(A,B,C)} + W_{(A,C,B)} \label{eq:decomp_CS_W_3_no_AI}
\end{equation}
where, for each permutation $(X,Y)$ of the two parties $B$ and $C$, $W_{(A,X,Y)}$ is a positive semidefinite matrix satisfying
\begin{equation}
{}_{[1-X_O]Y_{IO}} W_{(A,X,Y)} = 0 , \ {}_{[1-Y_O]}W_{(A,X,Y)} = 0 \,. \label{eq:decomp_CS_W_3_no_AI_constr}
\end{equation}
\end{proposition}

Note already that contrary to the decomposition of Proposition~\ref{prop:CS_charact_3_no_CO}, the two summands $W_{(A,B,C)}$ and $W_{(A,C,B)}$ above are not necessarily valid process matrices: indeed, they are not required to satisfy ${}_{[1-A_O]B_{IO}C_{IO}}W_{(A,X,Y)} = 0$ (only their sum must satisfy Eq.~\eqref{eq:decomp_CS_W_3_no_AI_constr0A}). This allows for dynamical causal orders, where $A$ (incoherently) controls the causal order between the next parties $B$ and $C$.

\begin{proof}
According to Proposition~\ref{prop:CS_charact_4_no_AI_no_DO}, a causally separable process matrix $W \in A_{O} \otimes B_{IO} \otimes C_{IO}$ must have a decomposition of the form~\eqref{eq:decomp_CS_W_3} that satisfies the constraints~\eqref{eq:decomp_CS_W_3_constrWX}--\eqref{eq:decomp_CS_W_3_constrWXYZ}.

In particular, the constraints ${}_{[1-A_O]}W_{(B,C,A)}={}_{[1-C_O]A_O}W_{(B,C,A)}=0$ imply ${}_{[1-C_O]}W_{(B,C,A)}=0$. The constraints on $W_{(B,A,C)}$ and $W_{(B,C,A)}$ in turn imply that ${}_{[1-A_O]C_{IO}}W_{(B)}={}_{[1-C_O]}W_{(B)}=0$, which, together with ${}_{[1-B_O]A_OC_{IO}}W_{(B)}=0$, further imply that ${}_{[1-A_O]B_{IO}C_{IO}}W_{(B)}={}_{[1-B_O]C_{IO}}W_{(B)}=0$. Similarly, one also has ${}_{[1-A_O]B_{IO}C_{IO}}W_{(C)}={}_{[1-C_O]B_{IO}}W_{(C)}={}_{[1-B_O]}W_{(C)}=0$.

It then follows, since $W=W_{(A)}+W_{(B)}+W_{(C)}$ and ${}_{[1-A_O]B_{IO}C_{IO}}W_{(A)}=0$, that ${}_{[1-A_O]B_{IO}C_{IO}}W=0$ as well; furthermore, by defining $\tilde W_{(A,B,C)} = W_{(A,B,C)} + W_{(B)}$ and $\tilde W_{(A,C,B)} = W_{(A,C,B)} + W_{(C)}$, we obtain the decomposition $W = \tilde W_{(A,B,C)} + \tilde W_{(A,C,B)}$ of the form~\eqref{eq:decomp_CS_W_3_no_AI}, with $\tilde W_{(A,B,C)}$ and $\tilde W_{(A,C,B)}$ satisfying the constraints~\eqref{eq:decomp_CS_W_3_no_AI_constr}.

\medskip

Conversely, it is clear that the decomposition of Proposition~\ref{prop:CS_charact_3_no_AI} is a particular case of that of Proposition~\ref{prop:CS_charact_3} (with $W_{(B)} = W_{(C)} = 0$), so that any process matrix that can be decomposed as in Eq.~\eqref{eq:decomp_CS_W_3_no_AI} and satisfies Eqs.~\eqref{eq:decomp_CS_W_3_no_AI_constr0A} and~\eqref{eq:decomp_CS_W_3_no_AI_constr} is causally separable according to Proposition~\ref{prop:CS_charact_3}.
\end{proof}

\subsubsection{Allowed and forbidden terms in a Hilbert–Schmidt basis decomposition}

We note that Ref.~\cite{oreshkov16} already provided a characterisation of general tripartite causally separable process matrices---or ``extensibly'' causally separable process matrices, in their terminology. Let us prove here the equivalence with our own characterisation (Proposition~\ref{prop:CS_charact_3}) explicitly.

According to Proposition~3.3 in Ref.~\cite{oreshkov16}, every tripartite (extensibly) causally separable process matrix $W \in A_{IO} \otimes B_{IO} \otimes C_{IO}$ can be written in the form 
\begin{equation}\label{eq:dec_ECS_tripartite_oreshkov16}
W = W_{(A)} + W_{(B)} + W_{(C)} 
\end{equation}
where each $W_{(X)}$ contains only Hilbert-Schmidt terms (see Appendix~\ref{app:allowed_forbidden_terms}) that are allowed in a process matrix compatible with party $X$ acting first as per Proposition~3.2 in~\cite{oreshkov16}---i.e., $W_{(X)} \in \L^{X \prec \{Y,Z\}}$ in our language---and has the form
\begin{equation}\label{eq:dec_ECS_tripartite_oreshkov16_2}
W_{(X)} = \Omega_{(X,Y,Z)} \otimes \id^{Z_O} + \Omega_{(X,Z,Y)} \otimes \id^{Y_O}
\end{equation}
where $\Omega_{(X,Y,Z)} \in X_{IO} \otimes Y_{IO} \otimes Z_{I}$ and $\Omega_{(X,Z,Y)} \in X_{IO} \otimes Y_{I} \otimes Z_{IO}$ are positive semidefinite. (We changed here the notations of Ref.~\cite{oreshkov16} to match ours; note in particular that unlike in~\cite{oreshkov16}, we again ignore the normalisation constraints in the decomposition of $W$.)

Any such process matrix can thus be decomposed as in Proposition~\ref{prop:CS_charact_3}, with $W_{(X,Y,Z)} = \Omega_{(X,Y,Z)} \otimes \id^{Z_O} \ge 0$ (which also implies $W_{(X)} = W_{(X,Y,Z)} + W_{(X,Z,Y)} \ge 0$). 
As $W_{(X)} \in \L^{X \prec \{Y,Z\}}$, it satisfies in particular ${}_{[1-X_O]Y_{IO}Z_{IO}}W_{(X)} = 0$ (see Eq.~\eqref{eq:constr_L_k_first_app}), as in Eq.~\eqref{eq:decomp_CS_W_3_constrWX}. It is furthermore immediate to see that each $W_{(X,Y,Z)} = \Omega_{(X,Y,Z)} \otimes \id^{Z_O}$ satisfies the second constraint in~\eqref{eq:decomp_CS_W_3_constrWXYZ}, i.e., ${}_{[1-Z_O]}W_{(X,Y,Z)} = 0$. Finally, one has ${}_{[1-Y_O]Z_{IO}}W_{(X)} = 0$ (see again Eq.~\eqref{eq:constr_L_k_first_app}) and ${}_{[1-Y_O]}W_{(X,Z,Y)} = 0$, which implies ${}_{[1-Y_O]Z_{IO}}W_{(X,Y,Z)} = {}_{[1-Y_O]Z_{IO}}(W_{(X)}-W_{(X,Z,Y)}) = 0$, i.e., the first constraint in Eq.~\eqref{eq:decomp_CS_W_3_constrWXYZ}. Thus, any process matrix that satisfies the characterisation of Proposition 3.3 from Ref.~\cite{oreshkov16} also satisfies that of our Proposition~\ref{prop:CS_charact_3}.
Conversely, let $W$ be a process matrix that has a decomposition as in Proposition~\ref{prop:CS_charact_3}. As discussed after that proposition, the conditions~\eqref{eq:decomp_CS_W_3_constrWX}--\eqref{eq:decomp_CS_W_3_constrWXYZ} imply that $W_{(X)} \in \L^{X \prec \{Y,Z\}}$.
Furthermore, $W_{(X,Y,Z)} \ge 0$ and ${}_{[1-Z_O]}W_{(X,Y,Z)} = 0$ implies that $W_{(X,Y,Z)}$ is of the form $\Omega_{(X,Y,Z)} \otimes \id^{Z_O}$ with $\Omega_{(X,Y,Z)} \ge 0$, so that $W_{(X)} = W_{(X,Y,Z)} + W_{(X,Z,Y)}$ is of the form~\eqref{eq:dec_ECS_tripartite_oreshkov16_2}.
This indeed establishes the equivalence of Proposition~3.3 in Ref.~\cite{oreshkov16} with our Proposition~\ref{prop:CS_charact_3}.

\medskip

As emphasised before, the matrices $W_{(X,Y,Z)} = \Omega_{(X,Y,Z)} \otimes \id^{Z_O}$ need not be valid process matrices (the only requirement is that $\Omega_{(X,Y,Z)} \ge 0$). Both individual summands in Eq.~\eqref{eq:dec_ECS_tripartite_oreshkov16_2} can thus contain terms that are forbidden in a process matrix compatible with party $X$ acting first, as long as these terms cancel out in the sum. More precisely, in addition to the terms that are allowed in a process matrix with $X$ first, $W_{(X,Y,Z)}$ and $W_{(X,Z,Y)}$ can contain terms of the form $\sigma_\mu^{X_I} \sigma_\nu^{X_O} \id^{Y_{IO}Z_{IO}}$ with $\sigma_\nu^{X_O} \neq \id^{X_O}$ (i.e., $\nu \ge 1$). Any other term that is forbidden in a process matrix with $X$ first has a nontrivial $\sigma$ operator on either $Y_O$ or $Z_O$, and thus cannot appear in $\Omega_{(X,Y,Z)}$ or $\Omega_{(X,Z,Y)}$, and cannot be cancelled out in Eq.~\eqref{eq:dec_ECS_tripartite_oreshkov16_2}.
In the explicit example of Eq.~\eqref{eq:TrD_Wswicth} given at the end of Sec.~\ref{subsec:examples}, for instance, on can check indeed that $W_{(A,B,C)}$ and $W_{(A,C,B)}$ contain the Hilbert-Schmidt term $\hat{\textsc{z}}^{A_O^c}\id^{B_{IO}^tC_{IO}^t}$, which come with opposite signs and cancel out in the sum.

\subsection{General multipartite causally separable process matrices}
\label{app:charact_CS_N}

\subsubsection{Necessary condition for causal separability}
\label{app:charact_CS_N_necessary}

Let us now prove the necessary condition for general multipartite causal separability given by Proposition~\ref{prop:CS_necessary}.

\begin{proof}
Consider an $N$-partite causally separable process matrix $W \in A_{IO}^\N$. Let us introduce now, for each party $A_k$, an extension $\bigotimes_{k' \in \N \backslash k} \big(A_{I_{(\gets k')}'}^k\otimes A_{I_{(\to k')}''}^k\big)$ of their incoming space, with dimensions $d_{A_{I_{(\gets k')}'}^k} = d_{A_{IO}^{k'}}$ and $d_{A_{I_{(\to k')}''}^k} = d_{A_{IO}^k}$, and let us attach to $W$ the state $\rho = \bigotimes_{k,k' \in \N, k\neq k'} \big( \ketbra{\Phi^+}^{A_{I_{(\to k')}''}^k/A_{I_{(\gets k)}'}^{k'}} \otimes \ketbra{\Phi^+}^{A_{I_{(\to k)}''}^{k'}/A_{I_{(\gets k')}'}^k} \big)$---i.e., we provide each pair of parties with two maximally entangled states, which will allow us to use the teleportation technique in either direction.

According to Definition~\ref{def:our_def-CS} and Proposition~\ref{prop:factor_rho}, $W$ must be decomposable as
\begin{equation}
W = \sum_{k \in \N} W_{(k)} \,, \label{eq:decomp_Wrho_N}
\end{equation}
where each term $W_{(k)}$ is a process matrix compatible with party $A_k$ acting first, and such that whatever that party does on $W_{(k)} \otimes \rho$, the resulting conditional process matrix for the other $N-1$ parties is causally separable.

Consider, for a given $k \in \N$, the process matrix $W_{(k)} \otimes \rho$, and let party $A_k$ perform, for a given $k' \neq k$, the CP map $M_k = \ketbra{\Phi^+}^{A_{IO}^k/A_{I_{(\to k')}''}^k} \otimes \id^{A_{I_\text{rest}'}^k} \ge 0$, with $A_{I_\text{rest}'}^k = A_{I_{(\gets k')}'}^k \bigotimes_{k'' \in \N \backslash \{k,k'\}} \big(A_{I_{(\gets k'')}'}^k\otimes A_{I_{(\to k'')}''}^k\big)$.
The resulting conditional process matrix for the other $N-1$ parties is then, according to Lemma~\ref{lemma:teleportation} (with $\tilde\rho = \Tr_{A_{I_{(\to k')}''}^k A_{I_{(\gets k)}'}^{k'}} [\rho]$),
\begin{align}
(W_{(k)} \otimes \rho)_{|M_k} \, \propto \, W_{(k)}^{A_{IO}^k \to A_{I_{(\gets k)}'}^{k'}} \otimes \Tr_k \Tr_{A_{I_{(\gets k)}'}^{k'}}\![\rho]\,.
\end{align}
As this conditional process matrix must be causally separable, it then follows from Proposition~\ref{prop:traceout} that $W_{(k)}^{A_{IO}^k \to A_{I_{(\gets k)}'}^{k'}}$ itself must be causally separable, which concludes the proof of Proposition~\ref{prop:CS_necessary} (where $W_{(k)}^{A_{IO}^k \to A_{I_{(\gets k)}'}^{k'}}$ was simply denoted $W_{(k)}^{A_{IO}^k \to A_{I'}^{k'}}$).
\end{proof}

\subsubsection{Sufficient condition for causal separability}
\label{app:charact_CS_N_sufficient}

Here we shall prove the sufficient condition for general multipartite causal separability given by Proposition~\ref{prop:CS_sufficient}.
Let us however first prove the claim that was made just after that proposition, namely that if Eq.~\eqref{eq:decomp_CS_suff_cond} is satisfied for all $(k_1, \ldots, k_n)$, then one also has, for all $(k_1, \ldots, k_n)$ with $1 \le n < N$, that
\begin{align}
& \forall \ \X \subseteq \N \backslash \{k_1, \ldots, k_n\}, \X \neq \emptyset, \notag \\
& {}_{\prod_{i \in \X}[1-A_O^i]A_{IO}^{\N \backslash \{k_1, \ldots, k_n\} \backslash \X}}W_{(k_1, \ldots, k_n)} = 0. \label{eq:decomp_CS_suff_cond_2_app}
\end{align}

\begin{proof}
This can be seen (by induction) as follows. Assume that Eq.~\eqref{eq:decomp_CS_suff_cond} is satisfied for all $(k_1, \ldots, k_n)$.

Eq.~\eqref{eq:decomp_CS_suff_cond_2_app} is indeed satisfied for all $(k_1, \ldots, k_n)$ for $n = N-1$, as in that case $W_{(k_1, \ldots, k_{N-1})} = W_{(k_1, \ldots, k_{N-1},k_N)}$ satisfies ${}_{[1-A_O^{k_N}]}W_{(k_1, \ldots, k_{N-1})} = 0$ by assumption~\eqref{eq:decomp_CS_suff_cond}.

Suppose then that for a given value of $n \ge 2$, Eq.~\eqref{eq:decomp_CS_suff_cond_2_app} is satisfied for all $(k_1, \ldots, k_n)$. Consider a given ordered subset of parties $(k_1, \ldots, k_{n-1})$ and a given nonempty subset $\X \subseteq \N \backslash \{k_1, \ldots, k_{n-1}\}$, and define for ease of notation the linear function $f_{\{k_1, \ldots, k_{n-1}\}}^\X(W) \coloneqq {}_{\prod_{i \in \X}[1-A_O^i]A_{IO}^{\N \backslash \{k_1, \ldots, k_{n-1}\} \backslash \X}}W$.

Let then $k_n \in \N \backslash \{k_1, \ldots, k_{n-1}\}$. If $\X = \{k_n\}$ then $f_\X(W_{(k_1, \ldots, k_n)}) = {}_{[1-A_O^{k_n}]A_{IO}^{\N \backslash \{k_1, \ldots, k_{n-1}, k_n\}}}W_{(k_1, \ldots, k_n)} = 0$ according to Eq.~\eqref{eq:decomp_CS_suff_cond}.
If on the other hand $\X \neq \{k_n\}$, then depending on whether $k_n \in \X$ or $k_n \notin \X$, we have $f_\X(W_{(k_1, \ldots, k_n)}) = {}_{[1-A_O^{k_n}]}({}_{\prod_{i \in \X'}[1-A_O^i]A_{IO}^{\N \backslash \{k_1, \ldots, k_{n-1}, k_n\} \backslash \X'}}W_{(k_1, \ldots, k_n)})$ with $\X' = \X \backslash k_n$, or $f_\X(W_{(k_1, \ldots, k_n)}) = {}_{A_{IO}^{k_n}}({}_{\prod_{i \in \X}[1-A_O^i]A_{IO}^{\N \backslash \{k_1, \ldots, k_{n-1}, k_n\} \backslash \X}}W_{(k_1, \ldots, k_n)})$. In both cases, $\X^{(\prime)} \subseteq \N \backslash \{k_1, \ldots, k_n\}, \X^{(\prime)} \neq \emptyset$, so that by the induction hypothesis $f_\X(W_{(k_1, \ldots, k_n)}) = 0$, which thus holds for all $k_n \in \N \backslash \{k_1, \ldots, k_{n-1}\}$.
As $W_{(k_1, \ldots, k_{n-1})} = \sum_{k_n \in \N \backslash \{k_1, \ldots, k_{n-1}\}} W_{(k_1, \ldots, k_{n-1}, k_n)}$, we also have $f_\X(W_{(k_1, \ldots, k_{n-1})}) = 0$, which, by induction, concludes the proof of Eq.~\eqref{eq:decomp_CS_suff_cond_2_app} for all ordered subsets $(k_1, \ldots, k_n)$ with $1 \le n < N$.
\end{proof}

In particular for the case $n = 1$, Eq.~\eqref{eq:decomp_CS_suff_cond_2_app} together with Eq.~\eqref{eq:decomp_CS_suff_cond} imply that each matrix $W_{(k_1)}$ is a valid process matrix compatible with party $A_{k_1}$ acting first (see Eq.~\eqref{eq:constr_L_k_first_app}).

Let us now prove Proposition~\ref{prop:CS_sufficient} by induction.

\begin{proof}
Clearly, it trivially holds for $N = 1$ (in which case Eq.~\eqref{eq:decomp_CS_suff_cond} ensures in particular that $W$ is a valid process matrix). (Note also that for $N = 2$ and $3$, it reduces to the sufficient conditions of Propositions~\ref{prop:CS_charact_2} and~\ref{prop:CS_charact_3}, respectively.)

Suppose Proposition~\ref{prop:CS_sufficient} holds in the $(N{-}1)$-partite case, and consider a matrix $W \in A_{IO}^\N$ that can be decomposed as in Eq.~\eqref{eq:decomp_CS_W}, with all partial sums $W_{(k_1, \ldots, k_n)}$ satisfying Eq.~\eqref{eq:decomp_CS_suff_cond}. Then we have
\begin{equation}
W = \sum_{k_1 \in \N} W_{(k_1)}
\end{equation}
with (as noted above) each $W_{(k_1)}$ being a valid process matrix compatible with party $A_{k_1}$ acting first.

Consider a CP map $M_{k_1}$ applied by party $A_{k_1}$ on $W_{(k_1)} = \sum_{\pi \in \Pi_{(k_1)}} W_\pi$. The resulting conditional process matrix for the remaining $N-1$ parties is
\begin{align}
(W_{(k_1)})_{|M_{k_1}} \coloneqq& \Tr_{k_1} [M_{k_1} \otimes \id^{\N \backslash k_1} \cdot W_{(k_1)}] \notag \\
=& \sum_{\pi \in \Pi_{(k_1)}} (W_\pi)_{|M_{k_1}} \label{eq:decomp_Wk1_Mk1_v0}
\end{align}
with $(W_\pi)_{|M_{k_1}} \coloneqq \Tr_{k_1} [M_{k_1} \otimes \id^{\N \backslash k_1} \cdot W_\pi]$.
By denoting by $\Pi^{\N \backslash k_1}$ the set of permutations of $\N \backslash k_1$ (and by $\Pi^{\N \backslash k_1}_{(k_2,\ldots,k_n)}$ the set of those that start with $k_2,\ldots,k_n$), by writing any permutation $\pi$ of $\N$ that starts with $k_1$ as $\pi = (k_1,\pi')$ with $\pi' \in \Pi^{\N \backslash k_1}$, and by defining $[(W_{(k_1)})_{|M_{k_1}}]_{\pi'} \coloneqq (W_{(k_1,\pi')})_{|M_{k_1}}$, we can re-write Eq.~\eqref{eq:decomp_Wk1_Mk1_v0} as
\begin{align}
(W_{(k_1)})_{|M_{k_1}} = \sum_{\pi' \in \Pi^{\N \backslash k_1}} [(W_{(k_1)})_{|M_{k_1}}]_{\pi'}, \label{eq:decomp_Wk1_Mk1}
\end{align}
in a similar form to Eq.~\eqref{eq:decomp_CS_W}.
For $n = 2, \ldots, N$, and for any ordered subset of parties $(k_2, \ldots, k_n)$ of $\N \backslash \{k_1\}$, the partial sums
\begin{align}
[(W_{(k_1)})_{|M_{k_1}}]_{(k_2, \ldots, k_n)} \coloneqq& \sum_{\pi' \in \Pi^{\N \backslash k_1}_{(k_2,\ldots,k_n)}} [(W_{(k_1)})_{|M_{k_1}}]_{\pi'} \notag \\
=& \sum_{\pi \in \Pi_{(k_1,k_2,\ldots,k_n)}} (W_\pi)_{|M_{k_1}} \notag \\
=& \ (W_{(k_1,k_2, \ldots, k_n)})_{|M_{k_1}}
\end{align}
then satisfy
\begin{align}
{}_{[1-A_O^{k_n}]A_{IO}^{\N \backslash \{k_1\} \backslash \{k_2,\ldots,k_n\}}}[(W_{(k_1)})_{|M_{k_1}}]_{(k_2, \ldots, k_n)} \notag \\
= {}_{[1-A_O^{k_n}]A_{IO}^{\N \backslash \{k_1,k_2,\ldots,k_n\}}}[(W_{(k_1,k_2, \ldots, k_n)})_{|M_{k_1}}] & \notag \\
= ({}_{[1-A_O^{k_n}]A_{IO}^{\N \backslash \{k_1,k_2,\ldots,k_n\}}}W_{(k_1,k_2, \ldots, k_n)})_{|M_{k_1}} &= 0
\end{align}
by assumption~\eqref{eq:decomp_CS_suff_cond}.
Thus, Eq.~\eqref{eq:decomp_Wk1_Mk1} provides a decomposition of the $(N{-}1)$-partite process matrix $(W_{(k_1)})_{|M_{k_1}}$ of the same form as in Eq.~\eqref{eq:decomp_CS_W}, with positive semidefinite matrices $[(W_{(k_1)})_{|M_{k_1}}]_{\pi'}$ and with all partial sums satisfying the analogue constraints as those of Eq.~\eqref{eq:decomp_CS_suff_cond}. By the induction hypothesis, this implies that $(W_{(k_1)})_{|M_{k_1}}$ is causally separable.

Note that the exact same reasoning also goes through if instead of $W$ we consider $W \otimes \rho$ with any ancillary state $\rho$.
Indeed, $W \otimes \rho$ also has a decomposition as in Eq.~\eqref{eq:decomp_CS_W} obtained simply by attaching the $\rho$ to every individual term in the decomposition of $W$. This shows that $W$ is causally separable, and by induction this proves that the decomposition of Proposition~\ref{prop:CS_sufficient} is indeed a sufficient condition for $W$ to be causally separable.
\end{proof}

For clarity and to get some better intuition on how it generalises the characterisation of Proposition~\ref{prop:CS_charact_3} for the tripartite case, let us write the sufficient condition of Proposition~\ref{prop:CS_sufficient} explicitly in the fourpartite case:

\begin{widetext}
\begin{proposition}[Sufficient condition for fourpartite causally separable process matrices] \label{prop:CS_sufficient_4}
If a matrix $W \in A_{IO} \otimes B_{IO} \otimes C_{IO} \otimes D_{IO}$ can be decomposed as
\begin{equation}
\begin{array}{rcccccccccccc}
W \! &\!=\!& && W_{(A)} && &\!+\!& W_{(B)} &\!+\!& W_{(C)} &\!+\!& W_{(D)} \\
&& \multicolumn{5}{c}{\overbrace{\hspace{102mm}}} && \overbrace{\hspace{8mm}} && \overbrace{\hspace{8mm}} && \overbrace{\hspace{8mm}} \\[-1mm]
&\!=\!& W_{(A,B)} &\!+\!& W_{(A,C)} &\!+\!& W_{(A,D)} &\!+\!& \cdots &\!+\!& \cdots &\!+\!& \cdots \\
&& \overbrace{\hspace{38mm}} && \overbrace{\hspace{38mm}} && \overbrace{\hspace{38mm}} && \resizebox{2.5mm}{5pt}{$\overbrace{}$} \, \resizebox{2.5mm}{5pt}{$\overbrace{}$} \, \resizebox{2.5mm}{5pt}{$\overbrace{}$} && \resizebox{2.5mm}{5pt}{$\overbrace{}$} \, \resizebox{2.5mm}{5pt}{$\overbrace{}$} \, \resizebox{2.5mm}{5pt}{$\overbrace{}$} && \resizebox{2.5mm}{5pt}{$\overbrace{}$} \, \resizebox{2.5mm}{5pt}{$\overbrace{}$} \, \resizebox{2.5mm}{5pt}{$\overbrace{}$} \\[-1mm]
&\!=\!& W_{(A,B,C,D)} \!+\! W_{(A,B,D,C)} &\!+\!& W_{(A,C,B,D)} \!+\! W_{(A,C,D,B)} &\!+\!& W_{(A,D,B,C)} \!+\! W_{(A,D,C,B)} &\!+\!& \cdots &\!+\!& \cdots &\!+\!& \cdots
\end{array} \label{eq:decomp_CS_W_4}
\end{equation}
with, for each permutation of the four parties $(X,Y,Z,T)$, positive semidefinite matrices $W_{(X,Y,Z,T)}$, $W_{(X,Y)} \coloneqq W_{(X,Y,Z,T)} + W_{(X,Y,T,Z)}$ and $W_{(X)} \coloneqq W_{(X,Y)} + W_{(X,Z)} + W_{(X,T)}$ satisfying
\begin{gather}
{}_{[1-X_O]Y_{IO}Z_{IO}T_{IO}}W_{(X)} = 0 \,, \label{eq:decomp_CS_W_4_constrX} \\
{}_{[1-Y_O]Z_{IO}T_{IO}}W_{(X,Y)} = 0 \,, \label{eq:decomp_CS_W_4_constrXY} \\
{}_{[1-Z_O]T_{IO}}W_{(X,Y,Z,T)} = 0 \,, \quad {}_{[1-T_O]}W_{(X,Y,Z,T)} = 0 \,, \label{eq:decomp_CS_W_4_constrXYZT}
\end{gather}
then $W$ is a valid fourpartite causally separable process matrix (as per our Definition~\ref{def:our_def-CS}).
\end{proposition}
\end{widetext}

It follows from Eqs.~\eqref{eq:decomp_CS_W_4_constrXY}--\eqref{eq:decomp_CS_W_4_constrXYZT} that for each party $X$, $W_{(X)}$ also satisfies ${}_{[1-Y_O]Z_{IO}T_{IO}}W_{(X)} = {}_{[1-Y_O][1-Z_O]T_{IO}}W_{(X)} = {}_{[1-Y_O][1-Z_O][1-T_O]}W_{(X)} = 0$ for all $X \neq Y \neq Z \neq T$ (see Eq.~\eqref{eq:decomp_CS_suff_cond_2_app}). This, together with Eq.~\eqref{eq:decomp_CS_W_4_constrX} and the fact that $W_{(X)} \ge 0$, implies that $W_{(X)}$ is a valid process matrix, compatible with party $X$ acting first (see Eq.~\eqref{eq:constr_L_k_first_app}).

Similarly, it follows from Eq.~\eqref{eq:decomp_CS_W_4_constrXYZT} that for each pair of parties $X,Y$, $W_{(X,Y)}$ also satisfies ${}_{[1-Z_O]T_{IO}}W_{(X,Y)} = {}_{[1-Z_O][1-T_O]}W_{(X,Y)} = 0$ for all $X \neq Y \neq Z \neq T$. This, together with Eq.~\eqref{eq:decomp_CS_W_4_constrXY} and the fact that $W_{(X,Y)} \ge 0$, implies that whatever CP map $M_X$ party $X$ applies, the conditional process matrix $(W_{(X,Y)})_{|M_X}$ is a valid tripartite process matrix for parties $Y,Z,T$, compatible with party $Y$ acting first.

Finally, Eq.~\eqref{eq:decomp_CS_W_4_constrXYZT} implies that whatever CP maps $M_X$, $M_Y$ parties $X$ and $Y$ apply, the conditional matrix $(W_{(X,Y,Z,T)})_{|M_X \otimes M_Y}$ is a valid bipartite process matrix for parties $Z,T$, compatible with party $Z$ acting first.

\subsection{Fourpartite causally separable process matrices in the particular case with $d_{D_O} = 1$}
\label{app:particular_4partite}

Consider now a fourpartite situation where party $D$ has no outgoing system (or a trivial one, with $d_{D_O} = 1$). It turns out that in such a case our sufficient condition above is also necessary, and it simplifies as follows (note the similarity with Proposition~\ref{prop:CS_charact_3}).

\begin{widetext}
\begin{proposition}[Characterisation of fourpartite causally separable process matrices with $d_{D_O} = 1$] \label{prop:CS_charact_4_no_DO}
In a fourpartite scenario where party $D$ has no outgoing system, a matrix $W \in A_{IO} \otimes B_{IO} \otimes C_{IO} \otimes D_I$ is a valid fourpartite causally separable process matrix (as per Definition~\ref{def:our_def-CS}) if and only if it can be decomposed as
\begin{equation}
\begin{array}{rcccccc}
W &=& W_{(A)} &+& W_{(B)} &+& W_{(C)} \\[1mm]
&=& \overbrace{W_{(A,B,C,D)} + W_{(A,C,B,D)}} &+& \overbrace{W_{(B,A,C,D)} + W_{(B,C,A,D)}} &+& \overbrace{W_{(C,A,B,D)} + W_{(C,B,A,D)}}
\end{array} \label{eq:decomp_CS_W_4_no_DO}
\end{equation}
where, for each permutation $(X,Y,Z)$ of the three parties $A$, $B$ and $C$, $W_{(X,Y,Z,D)}$ and $W_{(X)} \coloneqq W_{(X,Y,Z,D)} + W_{(X,Z,Y,D)}$ are positive semidefinite matrices satisfying
\begin{gather}
 {}_{[1-X_O]Y_{IO}Z_{IO}D_I} W_{(X)} = 0 \,, \label{eq:decomp_CS_W_4_no_DO_constrWX} \\[1mm]
 {}_{[1-Y_O]Z_{IO}D_I} W_{(X,Y,Z,D)} = 0 \,, \quad {}_{[1-Z_O]D_I}W_{(X,Y,Z,D)} = 0 \,. \label{eq:decomp_CS_W_4_no_DO_constrWXYZD}
\end{gather}
\end{proposition}
\end{widetext}

\begin{proof}
According to the necessary condition of Proposition~\ref{prop:CS_necessary}, a causally separable process matrix $W \in A_{IO} \otimes B_{IO} \otimes C_{IO} \otimes D_I$ must have a decomposition of the form $W = W_{(A)} + W_{(B)} + W_{(C)} + W_{(D)}$ where each $W_{(X)}$ is a process matrix compatible with $X$ first, such that for any $Y \neq X$, $W_{(X)}^{X_{IO} \to Y_{I'}}$ is causally separable.

Consider first $X=A$, and note already that as $W_{(A)}$ is compatible with $A$ first, one has, from Eq.~\eqref{eq:constr_L_k_first_app},
\begin{equation}
{}_{[1-A_O]B_{IO}C_{IO}D_I} W_{(A)} = 0.
\end{equation}
Taking now $Y=B$, we have that $W_{(A)}^{A_{IO} \to B_{I'}} \in B_{II'O} \otimes C_{IO} \otimes D_I$ is a tripartite causally separable process matrix in a scenario where one party ($D$) has no outgoing space. Using the characterisation of Proposition~\ref{prop:CS_charact_3_no_CO}, and re-attributing the system $B_{I'}$ back to $A_{IO}$ (as we did, e.g., in the proof of Proposition~\ref{prop:CS_charact_3_no_CO}), we obtain that $W_{(A)}$ must have a decomposition of the form
\begin{equation}
W_{(A)} = W_{(A,B,C,D)} + W_{(A,C,B,D)}
\end{equation}
with $W_{(A,B,C,D)}, W_{(A,C,B,D)} \ge 0$ satisfying
\begin{align}
{}_{[1-B_O]C_{IO}D_I} W_{(A,B,C,D)} = {}_{[1-C_O]D_I} W_{(A,B,C,D)} = 0, \notag \\
{}_{[1-C_O]A_{IO}B_{IO}D_I} W_{(A,C,B,D)} = {}_{[1-B_O]D_I} W_{(A,C,B,D)} = 0.
\end{align}
The first line further implies that ${}_{[1-C_O]B_{IO}D_I} W_{(A,B,C,D)} = 0$; noting that ${}_{[1-C_O]B_{IO}D_I} W_{(A)} = 0$ as well (as $W_{(A)}$ is compatible with $A$ first, see again Eq.~\eqref{eq:constr_L_k_first_app}) and that $W_{(A,C,B,D)} = W_{(A)} - W_{(A,B,C,D)}$, we also have
\begin{equation}
{}_{[1-C_O]B_{IO}D_I} W_{(A,C,B,D)} = 0.
\end{equation}
Hence, the term $W_{(A)}$ can be decomposed as in Eq.~\eqref{eq:decomp_CS_W_4_no_DO}, with the corresponding constraints being satisfied. The same holds, in a similar way, for the terms $W_{(B)}$ and $W_{(C)}$.

Consider now $W_{(D)}$. Taking e.g.\ $Y=A$, we have that $W_{(D)}^{D_I \to A_{I'}} \in A_{II'O} \otimes B_{IO} \otimes C_{IO}$ is a tripartite causally separable process matrix, which must have a decomposition as in Proposition~\ref{prop:CS_charact_3}. After relabelling the matrices $W_{(X,Y,Z)}$ from Eq.~\eqref{eq:decomp_CS_W_3} in the decomposition thus obtained to $W_{(X,Y,Z,D)}^{D_I \to A_{I'}}$, re-attributing the system $A_{I'}$ back to $D_I$ and applying the map ${}_{D_I}\cdot$ to all constraints of Eqs.~\eqref{eq:decomp_CS_W_3_constrWX}--\eqref{eq:decomp_CS_W_3_constrWXYZ}, we find that $W_{(D)}$ also has a decomposition as in Eq.~\eqref{eq:decomp_CS_W_4_no_DO} that satisfies the constraints~\eqref{eq:decomp_CS_W_4_no_DO_constrWX}--\eqref{eq:decomp_CS_W_4_no_DO_constrWXYZD}.

Altogether, $W$ is thus a combination of terms that have a decomposition as in Proposition~\ref{prop:CS_charact_4_no_DO}; combining these decompositions, it directly follows that $W$ itself has a decomposition of the form of Eq.~\eqref{eq:decomp_CS_W_4_no_DO} that satisfies the required constraints.

\medskip

Conversely, it is easy to see that if a matrix $W$ has a decomposition of the form~\eqref{eq:decomp_CS_W_4_no_DO}, then it is also of the form~\eqref{eq:decomp_CS_W_4} (where all terms $W_{(X,Y,Z,T)}$ with $D \neq T$ are $0$, and thus only the terms $W_{(X,Y,Z,D)} = W_{(X,Y)}$ remain). Furthermore, if the decomposition satisfies the constraints of Eqs.~\eqref{eq:decomp_CS_W_4_no_DO_constrWX}--\eqref{eq:decomp_CS_W_4_no_DO_constrWXYZD}, then it also satisfies those of Eqs.~\eqref{eq:decomp_CS_W_4_constrX}--\eqref{eq:decomp_CS_W_4_constrXYZT}. According to Proposition~\ref{prop:CS_sufficient}, this implies that such a process matrix $W$ is causally separable.
\end{proof}

One can further simplify the characterisation above in the particular fourpartite case where, in addition to one party ($D$) having no outgoing system, one also has a party ($A$) with no incoming system. We then obtain the following:

\begin{proposition}[Characterisation of fourpartite causally separable process matrices with $d_{A_I} = 1$ and $d_{D_O} = 1$] \label{prop:CS_charact_4_no_AI_no_DO}
In a fourpartite scenario where party $A$ has no incoming system and party $D$ has no outgoing system, a matrix $W \in A_O \otimes B_{IO} \otimes C_{IO} \otimes D_I$ is a valid fourpartite causally separable process matrix (as per Definition~\ref{def:our_def-CS}) if and only if
\begin{equation}
{}_{[1-A_O]B_{IO}C_{IO}D_I}W = 0 \label{eq:decomp_CS_W_4_no_AI_no_DO_constr0A}
\end{equation}
and $W$ can be decomposed as
\begin{equation}
W = W_{(A,B,C,D)} + W_{(A,C,B,D)} \label{eq:decomp_CS_W_4_no_AI_no_DO}
\end{equation}
where, for each permutation $(X,Y)$ of the two parties $B$ and $C$, $W_{(A,X,Y,D)}$ is a positive semidefinite matrix satisfying
\begin{equation}
{}_{[1-X_O]Y_{IO}D_I} W_{(A,X,Y,D)} = 0 , \ {}_{[1-Y_O]D_I}W_{(A,X,Y,D)} = 0 \,. \label{eq:decomp_CS_W_4_no_AI_no_DO_constr}
\end{equation}
\end{proposition}

We emphasise again that the two summands $W_{(A,B,C,D)}$ and $W_{(A,C,B,D)}$ above are not necessarily valid process matrices, thus allowing for dynamical causal orders.
We omit the proof of Proposition~\ref{prop:CS_charact_4_no_AI_no_DO} here, as it follows that of Proposition~\ref{prop:CS_charact_3_no_AI} very closely.
We note, as an aside, that both Propositions~\ref{prop:CS_charact_3_no_CO} and~\ref{prop:CS_charact_3_no_AI} could be obtained as corollaries of Proposition~\ref{prop:CS_charact_4_no_AI_no_DO} after removing one party. Namely, by imposing $d_{A_O} = 1$ above (which, in particular, makes Eq.~\eqref{eq:decomp_CS_W_4_no_AI_no_DO_constr0A} trivial), ignoring $A$ and relabelling $(B,C,D) \to (A,B,C)$ we obtain Proposition~\ref{prop:CS_charact_3_no_CO}; by imposing $d_{D_I} = 1$ instead in Proposition~\ref{prop:CS_charact_4_no_AI_no_DO} we directly obtain Proposition~\ref{prop:CS_charact_3_no_AI}.

\medskip

To conclude this section, we further note that Propositions~\ref{prop:CS_charact_4_no_DO} and~\ref{prop:CS_charact_4_no_AI_no_DO} generalise straightforwardly to cases with more parties $D, E, \ldots$ that have no outgoing spaces (by simply replacing $D_I$ by $D_IE_I\cdots$). Hence, we can give necessary and sufficient conditions for causal separability in any $N$-partite scenario in which at most 3 parties have nontrivial outgoing spaces.

\section{Explicit witness of causal nonseparability for $W^\text{act.}$}
\label{app:witness_Wactiv}

In this appendix we provide an explicit witness of causal nonseparability for the process matrix $W^\text{act.}$ introduced in Sec.~\ref{subsec:comparison}.

According to Eq.~\eqref{eq:Scone_N3_no_CO} in Appendix~\ref{app:witness_examples} below, in the tripartite scenario in which $W^\text{act.}$ is defined, where $d_{C_O} = 1$, the cone of causal witnesses can be characterised as
\begin{widetext}
\begin{align}
{\cal S} &= \big\{ S \in A_{IO}B_{IO}C_I \mid S = S_A^{(+)} + S_A^{(1)} + S_A^{(2)} \text{ with } S_A^{(+)} \ge 0, {}_{[1-A_O]B_{IO}C_I}S_A^{(1)} = S_A^{(1)}, {}_{[1-B_O]C_I}S_A^{(2)} = S_A^{(2)}, \notag \\
& \hspace{33mm} S = S_B^{(+)} + S_B^{(1)} + S_B^{(2)} \text{ with } S_B^{(+)} \ge 0, {}_{[1-B_O]A_{IO}C_I}S_B^{(1)} = S_B^{(1)}, {}_{[1-A_O]C_I}S_B^{(2)} = S_B^{(2)} \, \big\} \,. \label{eq:Scone_part_N3_bis}
\end{align}
Using the approach of Sec.~\ref{subsec:witnesses}, we obtained the following causal witness for $W^\text{act.}$, written, as in the definition~\eqref{eq:def_W_activ} of $W^\text{act.}$, in the order $C_IA_IB_IA_OB_O$:
\begin{align}
S^\text{act.} &= \frac14 \Big[ \id(\id\id{-}\hat{\textsc{z}}\hat{\textsc{z}})(\id\id{-}\hat{\textsc{z}}\hat{\textsc{z}}) - \frac23 \id (\hat{\textsc{x}}\hat{\textsc{x}}{+}\hat{\textsc{y}}\hat{\textsc{y}})(\id \hat{\textsc{z}}{+}\hat{\textsc{z}}\id) \notag \\[-1mm]
& \hspace{8mm} +\frac{1}{\sqrt{3}} \hat{\textsc{z}}(\id \hat{\textsc{z}}{-}\hat{\textsc{z}}\id)(\id\id{-}\hat{\textsc{z}}\hat{\textsc{z}}) +\frac{1}{\sqrt{3}} \hat{\textsc{x}}(\hat{\textsc{x}}\hat{\textsc{y}}{-}\hat{\textsc{y}}\hat{\textsc{x}})(\id \hat{\textsc{z}}{-}\hat{\textsc{z}}\id) + \frac13 \, \hat{\textsc{y}}(\hat{\textsc{x}}\hat{\textsc{x}}{+}\hat{\textsc{y}}\hat{\textsc{y}})(\id \hat{\textsc{z}}{-}\hat{\textsc{z}}\id) \Big] \, . \label{eq:def_S_activ}
\end{align}
\end{widetext}

To see that $S^\text{act.}$ indeed defines a valid causal witness, one can verify that it admits decompositions as in Eq.~\eqref{eq:Scone_part_N3_bis} above, with (still written in the same order) $S_A^{(1)} = 0,$ $S_A^{(2)} = \frac14 [-\frac43 \id (\hat{\textsc{x}}\hat{\textsc{x}}{+}\hat{\textsc{y}}\hat{\textsc{y}})\id \hat{\textsc{z}} ],$ $S_A^{(+)} = S^\text{act.} - S_A^{(2)}$ and similarly $S_B^{(1)} = 0,$ $S_B^{(2)} = \frac14 [-\frac43 \id (\hat{\textsc{x}}\hat{\textsc{x}}{+}\hat{\textsc{y}}\hat{\textsc{y}})\hat{\textsc{z}}\id ],$ $S_B^{(+)} = S^\text{act.} - S_B^{(2)}$. One can easily check that all constraints in Eq.~\eqref{eq:Scone_part_N3_bis} are satisfied.

With $S^\text{act.}$ thus defined, one finds $\Tr[S^\text{act.} \cdot W^\text{act.}] = -(\frac{4}{\sqrt{3}} - 2) < 0$, which proves that $W^\text{act.}$ is indeed causally nonseparable according to our Definition~\ref{def:our_def-CS}---or equivalently, to Ara\'ujo \emph{et al.}'s Definition~\ref{def:AB+-CS}, or ``extensibly causally nonseparable'' according to Definition~\ref{def:OG-ECS}---even though, as proven in Sec.~\ref{subsec:comparison}, it is causally separable according to OG's Definition~\ref{def:OG-CS}.

Since the causal witness $S^\text{act.}$ above was obtained with the SDP optimisation technique described in Sec.~\ref{subsec:witnesses}, it allows us to determine the robustness of $W^\text{act.}$ to white noise.
From Eq.~\eqref{eq:rSoptReln} we thus find that its random robustness is $r^* = - \Tr[S^\text{act.} \cdot W^\text{act.}] = \frac{4}{\sqrt{3}} - 2 \simeq 0.31$.

\subsubsection*{``Activation'' of causal nonseparability with $W^\text{act.}$}

It is instructive to see explicitly how causal nonseparability can be ``activated'' by attaching an entangled ancillary state to $W^\text{act.}$.

Recall that $W^\text{act.}$ is compatible with party $C$ acting first. As shown in Sec.~\ref{subsec:comparison}, it is such that for any CP map (or POVM element) $M_{\vec c}$ applied by $C$, the conditional bipartite process matrix $(W^\text{act.})_{|M_{\vec c}}$ is causally separable. This is precisely why $W^\text{act.}$ is considered to be causally separable according to OG's Definition~\ref{def:OG-CS}.

Consider now attaching an ancillary maximally entangled state $\rho = \ketbra{\Phi^+}^{A_{I'}/C_{I'}}$, shared by $A$ and $C$ with dimensions $d_{A_{I'}} = d_{C_{I'}} = d_{C_I}$, and letting $C$ project his two incoming systems onto $\ketbra{\Phi^+}^{C_I/C_{I'}}$. The resulting conditional process matrix $(W^\text{act.}\otimes \rho)_{|M_C = \ketbra{\Phi^+}}$ shared by $A$ and $B$ is then (up to normalisation) $(W^\text{act.})^{C_I \to A_{I'}}$, i.e., it is formally represented by the same matrix as $W^\text{act.}$, Eq.~\eqref{eq:def_W_activ}, with party $C$'s incoming system now given to party $A$ (see Lemma~\ref{lemma:teleportation} of Appendix~\ref{app:characterisation_CS_Ws}). One can verify that $(W^\text{act.}\otimes \rho)_{|M_C = \ketbra{\Phi^+}}$ thus obtained is causally nonseparable by constructing a (bipartite) causal witness using, for instance, the explicit characterisation of Eq.~\eqref{eq:Scone_N2} below, in a similar way to what we did for $W^\text{act.}$ above.

Note, however, that this argument is not sufficient to conclude that $W^\text{act.}$ is (extensibly) causally nonseparable: one indeed needs to prove that \emph{for any} possible decomposition of the form $W^\text{act.}\otimes \rho = W_{(A)}^\rho + W_{(B)}^\rho + W_{(C)}^\rho$ with each $W_{(X)}^\rho$ compatible with party $X$ acting first, there exist CP maps $M_A$, $M_B$, or $M_C$, that make either $(W_{(A)}^\rho)_{|M_A}$, $(W_{(B)}^\rho)_{|M_B}$ or $(W_{(C)}^\rho)_{|M_C}$ causally nonseparable.%
\footnote{Indeed, a process matrix compatible with $C$ first (in short, of the form $W = W_{(C)}$), and such that for some CP map $M_C$ the conditional process matrix $W_{|M_C}$ is causally nonseparable, may still be causally separable if it also has another, causally separable, decomposition of the form $W = W_{(A)}' + W_{(B)}' + W_{(C)}'$.
An example is for instance $W = W_0 \otimes \ketbra{\Phi^+}^{A_{I'}/C_{I'}}$ with (written again in the order $C_IA_IB_IA_OB_O$)
\begin{align}
W_0 = \frac18 \Big[ & (\id\id\id + \id \hat{\textsc{z}} \hat{\textsc{z}} + \hat{\textsc{z}} \id \hat{\textsc{z}} + \hat{\textsc{z}} \hat{\textsc{z}} \id) \id \id \notag \\[-1mm]
& + {\textstyle \frac{1}{\sqrt{2}}} ( \hat{\textsc{x}} \hat{\textsc{x}} \hat{\textsc{x}} - \hat{\textsc{x}} \hat{\textsc{y}} \hat{\textsc{y}} - \hat{\textsc{y}} \hat{\textsc{x}} \hat{\textsc{y}} - \hat{\textsc{y}} \hat{\textsc{y}} \hat{\textsc{x}} ) \hat{\textsc{x}} \id \notag \\[-1mm]
& + {\textstyle \frac{1}{\sqrt{2}}} ( \hat{\textsc{y}} \hat{\textsc{y}} \hat{\textsc{y}} - \hat{\textsc{y}} \hat{\textsc{x}} \hat{\textsc{x}} - \hat{\textsc{x}} \hat{\textsc{y}} \hat{\textsc{x}} - \hat{\textsc{x}} \hat{\textsc{x}} \hat{\textsc{y}} ) \id \hat{\textsc{y}}  \Big] \, .
\end{align}
One can check that $W \in \L^{C \prec \{A,B\}}$ and that with $M_C = \ketbra{\Phi^+}^{C_I/C_{I'}}$, the bipartite conditional process matrix $W_{|M_C}$ is causally nonseparable---even though $W$ is also compatible with the fixed order $A \prec B \prec C$ (and is hence causally separable). (Here the ancillary entangled state attached to $W_0$ and the CP map $M_C$ allow party $C$ to ``teleport'' their incoming system in $W_0$ to $A$; the same observation holds if $C$ teleports his system to $B$ instead.)
\\
A similar observation can be made at the level of correlations: a tripartite correlation $P(a,b,c|x,y,z)$ compatible with $C$ first and such that the bipartite conditional correlation $P_{z,c}(a,b|x,y) \coloneqq P(a,b|x,y,z,c)$ is noncausal for some $z,c$ may in general still be causal. An example with binary inputs and outputs $0,1$ is $P(a,b,c|x,y,z) \coloneqq \frac12 \delta_{b,x} \delta_{c,a \oplus y}$, which is indeed compatible with $C$ first (as $P(c|x,y,z) = \frac12$ does not depend on $x,y$) and is such that conditioned on $C$'s output $c = 0$, the resulting conditional bipartite correlation $P_{z,c=0}(a,b|x,y) = \delta_{b,x} \delta_{a,y}$ shared by $A$ and $B$ is noncausal (it violates the ``Guess Your Neighbour's Input'' inequality~\cite{branciard16} maximally). Nevertheless, $P$ is clearly also compatible with the fixed order $A \prec B \prec C$, and is hence causal.
\\
Note that the argument given by OG in Ref.~\cite{oreshkov16} to show activation of causal nonseparability consisted precisely in proving that, after attaching an ancillary state, the correlations generated by a given tripartite process matrix were compatible with $C$ first and such that the bipartite conditional correlation $P_{z,c}(a,b|x,y)$ was noncausal. In that case, however, $C$ was performing a deterministic operation (i.e., $c$ could only take a single fixed value), so this argument was enough, in their case, to prove that the tripartite correlation under consideration was indeed noncausal~\cite{abbott16,oreshkov18a}.
}
Our construction of a causal witness for $W^\text{act.}$ confirms nonetheless that this must indeed be the case, which allows us conclude, using OG's terminology, that the entangled ancillary state $\rho$ introduced here indeed ``activates'' the causal nonseparability of $W^\text{act.}$.

\section{Equivalence between Oreshkov and Giarmatzi's extensible causal (non)separability and our definition of causal (non)separability}
\label{app:equiv_OGECS_our_def}

\renewcommand{\theproposition}{D\arabic{proposition}}
\setcounter{proposition}{0}

In this appendix we prove that OG's Definition~\ref{def:OG-ECS} of extensible causal (non)separability and our Definition~\ref{def:our_def-CS} of multipartite causal (non)separability are equivalent. 

\begin{proof}
Let $W$ be an $N$-partite process matrix that is causally separable as per our Definition~\ref{def:our_def-CS}. The conditional $(N{-}1)$-partite process matrices $(W_{(k)}^\rho)_{|M_k} = \Tr_k [M_k \otimes \id^{\N \backslash k} \, \cdot \, W_{(k)}^\rho]$ in Definition~\ref{def:our_def-CS} are again causally separable (as per our definition), and thus fulfil in particular Definition~\ref{def:OG-CS}. Therefore, $W \otimes \rho$ is causally separable as per OG's Definition~\ref{def:OG-CS} (OG-CS) for any $A_{I'}^\N$ and any ancillary quantum state $\rho \in A_{I'}^\N$. That is, $W$ is extensibly causally separable as per OG's Definition~\ref{def:OG-ECS} (OG-ECS).

\medskip

The proof of the converse is more involved.
The idea is to consider, for an $N$-partite OG-ECS process matrix $W$ and two ancillary quantum states $\rho'$ and $\rho''$, the extended process matrices $W \otimes \rho'$ and $W \otimes \rho' \otimes \rho''$, which are both OG-CS. By comparing the corresponding decompositions we will show that the conditional $(N{-}1)$-partite process matrices obtained from the decomposition of $W \otimes \rho'$ are not only OG-CS, but also OG-ECS. From there, one can conclude by induction that $W$ then also satisfies our Definition~\ref{def:our_def-CS}.

The difficulty here is that the causally separable decomposition of $W \otimes \rho$ (for $\rho = \rho'$ or $\rho = \rho' \otimes \rho''$ in our case here) depends, \emph{a priori}, on $\rho$. The following proposition states, however, that there exists a decomposition of $W$ that provides a unique causally separable decomposition of $W \otimes \rho$ for any $\rho$.

\begin{proposition} \label{prop:convexDec}
Any $N$-partite extensibly causally separable (OG-ECS) process matrix $W$, as per Definition~\ref{def:OG-ECS}, can be decomposed as
\begin{equation}
 W = \sum_{k \in \N} q_k \, W_{(k)} , \label{eq:convexDec}
\end{equation}
with $q_k\ge 0$, $\sum_k q_k = 1$, and where for each $k$, $W_{(k)}$ is a process matrix compatible with party $A_k$ acting first, and is such that for any extension  $A_{I'}^\N$, any ancillary quantum state $\rho \in A_{I'}^\N$ and any possible CP map $M_k \in A_{II'O}^k$ applied by party $A_k$, the conditional $(N{-}1)$-partite process matrix $(W_{(k)} \otimes \rho)_{|M_k} \coloneqq \Tr_k [M_k \otimes \id^{\N \backslash k} \, \cdot \, W_{(k)} \otimes \rho]$ is causally separable (OG-CS) as per Definition~\ref{def:OG-CS}.
\end{proposition}

\begin{proof}
Consider an $N$-partite OG-ECS process matrix $W$. By Definition~\ref{def:OG-ECS}, for any extension $A_{I'}^\N$ and any state $\rho \in A_{I'}^\N$, the extended $N$-partite process matrix $W \otimes \rho$ must have a decomposition of the form
\begin{equation}\label{eq:convexDec1}
 W \otimes \rho = \sum_{k \in \N} q_k \, W_{(k)}^\rho,
\end{equation}
where for each $k$, $W_{(k)}^\rho$ is a process matrix compatible with party $A_k$ acting first, and such that whatever that party does, the resulting conditional process matrix for the other $(N{-}1)$ parties is OG-CS. 

By an argument similar to that of Proposition~\ref{prop:factor_rho}, it is easy to see that $W_{(k)}^\rho$ can without loss of generality be taken to be of the form $W_{(k)}^\rho = W_{(k)} \otimes \rho$.
We emphasise again that the convex decomposition of $W = \sum_{k \in \N} q_k \, W_{(k)}$ that then follows from Eq.~\eqref{eq:convexDec1} could \emph{a priori} depend on the ancillary state $\rho$. We will however now show that for all extensions and ancillary quantum states one can choose the same decomposition of $W$.

First, note that for any finite set of extensions $\{A_{I_1}^\N, A_{I_2}^\N ,\ldots, A_{I_n}^\N \}$ and ancillary quantum states  $\{ \rho_1 \in A_{I_1}^\N, \rho_2 \in A_{I_2}^\N ,\ldots, \rho_n \in A_{I_n}^\N \}$ under consideration, one can indeed choose the same decomposition---consider the ancillary state $\rho_1 \otimes \cdots \otimes \rho_n \in A_{I_1}^\N \otimes \cdots \otimes A_{I_n}^\N$, and the corresponding decomposition
\begin{equation}\label{eq:decTensorProduct}
 W \otimes \rho_1 \otimes \cdots \otimes \rho_n = \sum_{k \in \N} q_k \, W_{(k)} \otimes \rho_1 \otimes \cdots \otimes \rho_n, 
\end{equation}
with each $W_{(k)} \otimes \rho_1 \otimes \cdots \otimes \rho_n$---and therefore, each $W_{(k)}$---a process matrix compatible with party $A_k$ acting first, and such that for any operation $M_k$ applied by $A_k$ the resulting conditional process matrix $(W_{(k)} \otimes \rho_1 \otimes \cdots \otimes \rho_n)_{|M_k}$ for the other $(N{-}1)$ parties is OG-CS. 
Proposition~\ref{prop:traceout} now implies that these conditional process matrices remain causally separable when tracing out all but one ancillary states in the tensor product. Therefore, the decomposition of $W$ obtained from Eq.~\eqref{eq:decTensorProduct} can be chosen for any of the individual $\rho_j \in A_{I_j}^\N$.

Next, one uses the following result from basic topology:
 \begin{customthm}{2.36 in Ref.~\cite{rudin76}}
    If $\{K_\alpha\}$ is a collection of compact subsets of a metric space $X$ such that the intersection of every finite subcollection of $\{K_\alpha \}$ is nonempty, then $\bigcap K_\alpha$ is nonempty.
    \end{customthm}
    
    Here, let the index set be the set of all possible ancillary quantum states (of any dimension), and the set $K_\rho$, indexed by some quantum state $\rho$, be the set of possible causally separable decompositions of $W$ corresponding to the ancillary state $\rho$. 
	The finite intersection property follows from the observation above---for any finite set of ancillary states $\{\rho_1, \ldots, \rho_n \}$, there exists a common decomposition, that is, the intersection $K_{\rho_1} \cap \cdots \cap K_{\rho_n}$ is nonempty. 
	As the conditions of the above Theorem are satisfied,%
\footnote{More precisely, let $X$ be the space of $N$-tuples of Hermitian matrices $\vec{W} = (W_{(1)},\ldots,W_{(N)})$, equipped with the standard Euclidean metric, and let the sets $K_\rho$ be defined as
\begin{align}
\label{eq:char_K_rho}
& K_\rho \coloneqq \Big\{(W_{(1)},\ldots,W_{(N)})\,\,\Big|\, \Big( {\textstyle \sum_{k=1}^N} W_{(k)} = W \Big) \ \text{and} \notag\\[-1mm]
& \hspace{12mm} \Big( W_{(k)} \in \L^{A_k \prec (\N \backslash A_k)} \ \forall k \Big) \ \text{and} \ \left(W_{(k)} \ge 0 \ \forall k\right) \notag\\[-1mm]
& \hspace{22mm} \text{and} \ \Big( (W_{(k)} \otimes \rho)_{|M_k} \ \text{is OG-CS} \ \forall k, M_k \Big) \Big\} . 
\end{align}
It follows from the positivity of the $W_{(k)}$'s and the normalisation of $W$ that the sets $K_\rho$ are bounded. One can further easily convince oneself that the sets characterised by the four individual conditions in Eq.~\eqref{eq:char_K_rho} are closed, and thus, as it is the intersection of closed sets, that $K_\rho$ is closed. 
The sets $K_\rho$ being bounded and closed, it follows from the Heine-Borel theorem that they are compact, as required for Theorem 2.36 in Ref.~\cite{rudin76} to be applicable here.}
it guarantees that the intersection $\bigcap K_\rho$ over all quantum states $\rho$ is nonempty.
That is, there exists indeed a convex decomposition of $W$,
\begin{equation}
 W = \sum_{k \in \N} q_k \, W_{(k)} , 
\end{equation}
with $q_k\ge 0$, $\sum_k q_k = 1$, and where for each $k$, $W_{(k)}$ is a process matrix compatible with party $A_k$ acting first, and is such that for any extension $A_{I'}^\N$, any ancillary quantum state $\rho \in A_{I'}^\N$ and any possible CP map $M_k \in A_{II'O}^k$ applied by party $A_k$, the conditional $(N{-}1)$-partite process matrix $(W_{(k)} \otimes \rho)_{|M_k}$ is OG-CS.
\end{proof} 

One can now prove by induction that any OG-ECS process matrix is causally separable according to our Definition~\ref{def:our_def-CS}. In the single-partite case ($N{=}1$), the claim is trivial. Suppose, for $N \ge 2$, that the claim holds true in the $(N{-}1)$-partite case.
Let then $W$ be an $N$-partite OG-ECS process matrix.
According to Proposition~\ref{prop:convexDec}, $W$ has a decomposition of the form~\eqref{eq:convexDec}, such that for any $k$, for any arbitrary extensions $A_{I'}^\N$, $A_{I''}^{\N \backslash k}$, any ancillary quantum states $\rho' \in A_{I'}^\N$, $\rho'' \in A_{I''}^{\N \backslash k}$, and any CP map $M_k \in A_{II'O}^k$ applied by party $A_k$, the conditional $(N{-}1)$-partite process matrices $(W_{(k)} \otimes \rho')_{|M_k}$ and $(W_{(k)} \otimes \rho' \otimes \rho'')_{|M_k} = (W_{(k)} \otimes \rho')_{|M_k} \otimes \rho''$ are OG-CS. That is, for any $A_{I'}^\N$, $\rho'$ and $M_k$, $(W_{(k)} \otimes \rho')_{|M_k}$ is OG-ECS---and therefore, by the induction hypothesis, it is causally separable according to our Definition~\ref{def:our_def-CS}.
Summing up, we thus have, for any $A_{I'}^\N$ and any $\rho' \in A_{I'}^\N$, a decomposition of the form $W \otimes \rho' = \sum_k q_k W_{(k)} \otimes \rho'$ such that for any $M_k \in A_{II'O}^k$, $(W_{(k)} \otimes \rho')_{|M_k}$ is causally separable.
This means that $W$ itself is causally separable as per our Definition~\ref{def:our_def-CS}, which concludes the proof.
\end{proof}

\section{Causally separable process matrices can only generate causal correlations}
\label{app:causal_correl}

\renewcommand{\thedefinition}{E\arabic{definition}}
\setcounter{definition}{0}

Here we show explicitly that causally separable process matrices, according to our Definition~\ref{def:our_def-CS}, can only generate so-called ``causal correlations'' (even when attaching ancillary entangled states).

Let us first recall the definition of $N$-partite causal correlations given in Ref.~\cite{abbott16} (which is equivalent to that first introduced in Ref.~\cite{oreshkov16}):

\begin{definition}[$N$-partite causal correlations] \label{def:causal_correl}
For \mbox{$N=1$}, any correlation $P(a_1|x_1)$ is causal. For $N \ge 2$, an $N$-partite correlation $P(\vec a|\vec x)$ is said to be \emph{causal} if and only if it can be decomposed as 
\begin{equation}
 P(\vec a|\vec x) = \sum_{k \in \N} q_k \, P_k(a_k|x_k) \, P_{k,x_k,a_k}(\vec a_{\N \backslash k} | \vec x_{\N \backslash k}), \label{eq:causal_correl}
\end{equation}
with $q_k\ge 0$, $\sum_k q_k = 1$, where (for each $k$) $P_k(a_k|x_k)$ is a single-partite (and hence causal) correlation and (for each $k,x_k,a_k$) $P_{k,x_k,a_k}(\vec a_{\N \backslash k} | \vec x_{\N \backslash k})$ is a causal $(N{-}1)$-partite correlation.
\end{definition}

By this definition, for $N = 1$, any correlation---and in particular, any correlation generated by a (trivially) causally separable single-partite process matrix---is causal.

Assume, for $N > 1$, that any correlation generated by a $(N{-}1)$-partite causally separable process matrix is causal. Consider then an $N$-partite process matrix $W \in A_{IO}^\N$, some ancillary state $\rho \in A_{I'}^\N$, and some CP maps $M_{a_k|x_k} \in A_{II'O}^\N$ (for any $k,x_k,a_k$), which all together generate the probability distribution
\begin{equation}
P(\vec a | \vec x) = \Tr[ M_{a_1|x_1} \otimes \cdots \otimes M_{a_N|x_N} \cdot W \otimes \rho],
\end{equation}
as in Eq.~\eqref{eq:born_rule_Npartite}.
Assuming that $W$ is causally separable, according to Definition~\ref{def:our_def-CS} $W \otimes \rho$ can be decomposed as in Eq.~\eqref{eq:our_def-CS}, which allows us to write
\begin{align}
P(\vec a | \vec x) &= \sum_{k \in \N} q_k \, \Tr[ M_{a_1|x_1} \otimes \cdots \otimes M_{a_N|x_N} \cdot W_{(k)}^\rho] . \label{eq:P_proof_causal}
\end{align}

Here $W_{(k)}^\rho$ is compatible with party $A_k$ acting first, so that for any set of CPTP maps $M_{x_{k'}}$ with $k' \neq k$,
\begin{align}
& \Tr[ M_{a_k|x_k} \bigotimes_{k' \in \N \backslash k}\!\! M_{x_{k'}} \cdot W_{(k)}^\rho] \notag \\
& = \Tr[ M_{a_k|x_k} \!\!\!\bigotimes_{k' \in \N \backslash k}\!\!\!\! {\textstyle \frac{\id}{d_{A_O^{k'}}}} \cdot W_{(k)}^\rho] = \frac{1}{d_{A_O^{\N \backslash k}}} \!\Tr[ (W_{(k)}^\rho)_{|M_{a_k|x_k}} ], \label{eq:Pk_proof_causal}
\end{align}
which does not depend on the choice of CPTP maps $M_{x_{k'}}$, and defines a probability distribution $P_k(a_k|x_k)$ for party $A_k$.

The conditional process matrix $(W_{(k)}^\rho)_{|M_{a_k|x_k}}$ for parties in $\N \backslash k$ can be renormalised (when nonzero) by defining $(\widetilde W_{(k)}^\rho)_{|M_{a_k|x_k}} \coloneqq \frac{1}{P_k(a_k|x_k)} (W_{(k)}^\rho)_{|M_{a_k|x_k}}$, so that $(\widetilde W_{(k)}^\rho)_{|M_{a_k|x_k}}$ is now a properly normalised process matrix (according to Eq.~\eqref{eq:Pk_proof_causal} above that defines $P_k(a_k|x_k)$, we indeed have $\Tr[ (\widetilde W_{(k)}^\rho)_{|M_{a_k|x_k}} ] = d_{A_O^{\N \backslash k}}$, as required by Eq.~\eqref{eq:constr_norm_W}). We can then write Eq.~\eqref{eq:P_proof_causal} as
\begin{align}
P(\vec a | \vec x) &= \sum_{k \in \N} q_k \, \Tr[ \bigotimes_{k' \in \N \backslash k} M_{a_{k'}|x_{k'}} \cdot (W_{(k)}^\rho)_{|M_{a_k|x_k}}] \notag \\
&= \sum_{k \in \N} q_k \, P_k(a_k|x_k) \, P_{k,x_k,a_k}(\vec a_{\N \backslash k} | \vec x_{\N \backslash k}) \label{eq:causal_correl_proof}
\end{align}
with
\begin{align}
P_{k,x_k,a_k}(\vec a_{\N \backslash k} | \vec x_{\N \backslash k}) \coloneqq \Tr[ \!\! \bigotimes_{k' \in \N \backslash k} \!\!\!\! M_{a_{k'}|x_{k'}} \!\cdot\! (\widetilde W_{(k)}^\rho)_{|M_{a_k|x_k}}] .
\end{align}

Now, by assumption and according to Definition~\ref{def:our_def-CS} $(\widetilde W_{(k)}^\rho)_{|M_{a_k|x_k}}$ must a causally separable process matrix; by the induction hypothesis it can only generate causal correlations, which implies that $P_{k,x_k,a_k}(\vec a_{\N \backslash k} | \vec x_{\N \backslash k})$ is causal. Eq.~\eqref{eq:causal_correl_proof} thus provides a causal decomposition of $P(\vec a | \vec x)$ as in Eq.~\eqref{eq:causal_correl} of Definition~\ref{def:causal_correl}, which proves that the correlation $P(\vec a | \vec x)$ obtained from the $N$-partite causally separable process matrix $W$ is causal, and which, by induction, concludes the proof.

\section{Relationship between our necessary and sufficient conditions for causal separability} 
 
\subsection{A necessary but not sufficient condition}
\label{app:gap_nec_suff}

In our recursive necessary condition of Proposition~\ref{prop:CS_necessary} for general multipartite causal separability, we require the $(N{-}1)$-partite process matrices $W_{(k)}^{A_{IO}^k \to A_{I'}^{k'}}$ to be causally separable for each $k' \neq k$. 
In the tripartite case, it is not necessary to impose this explicitly, since considering the teleportation of $A_k$'s systems to some arbitrary $A_{k'}$ yields necessary conditions that already coincide with the sufficient conditions for tripartite causal separability (see the proof in Appendix~\ref{app:charact_CS_3_general}). In the general case, however, considering the teleportation to just one or some of the parties yields weaker necessary conditions that may not be sufficient. In this appendix we present an explicit fourpartite example. 

We consider the fourpartite scenario where $A$ has a trivial incoming space ($d_{A_I} = 1$) and $D$ has a trivial outgoing space ($d_{D_O} = 1$), and define the following matrix in $A_O \otimes B_{IO} \otimes C_{IO} \otimes D_{I}$:
\begin{align}
W^\text{gap} &\coloneqq \frac18 \Big[\id^{\otimes 6} + \frac{1}{\sqrt{2}} \, \hat{\textsc{z}} (\id\hat{\textsc{z}}\hat{\textsc{z}}\id + \hat{\textsc{z}}\id\hat{\textsc{x}}\hat{\textsc{z}})\id\Big] . 
\end{align}
It is easy to verify that $W^\text{gap}$ satisfies Eq.~\eqref{eq:constr_L_k_first} for $A_k=A$, i.e., that it is a valid process matrix compatible with party $A$ acting first (note its similarity with the original process matrix of Oreshkov, Costa and Brukner~\cite{oreshkov12}). Furthermore, it satisfies ${}_{[1-B_O]A_OC_{IO}D_I}W^\text{gap} = {}_{[1-C_O]A_OD_I}W^\text{gap} = {}_{[1-D_O]}W^\text{gap} = 0$ (as well as ${}_{[1-C_O]A_OB_{IO}D_I}W^\text{gap} = {}_{[1-B_O]A_OD_I}W^\text{gap} = 0$).
Thus, $W^\text{gap}$ can be decomposed as in Eqs.~\eqref{eq:decomp_CS_WA_4}--\eqref{eq:explicit_NC_4partite_2} with $Y=D$ and a single term in the decomposition, $W^\text{gap} = W_{(A)} = W_{(A,B)}^{[A \to D]} = W_{(A,B,C,D)}^{[A \to D]}$ (or $W^\text{gap} = W_{(A)} = W_{(A,C)}^{[A \to D]} = W_{(A,C,B,D)}^{[A \to D]}$) satisfying Eq.~\eqref{eq:explicit_NC_4partite}.
In other words, the tripartite conditional process matrix that we obtain by teleporting $A_O$ to $D_{I'}$ is causally separable (it is compatible with both fixed causal orders $B \prec C \prec D$ and $C \prec B \prec D$).

However, this is not the case when teleporting $A_O$ to $B_{I'}$, or to $C_{I'}$. $W^\text{gap}$ indeed cannot be decomposed as in Eq.~\eqref{eq:explicit_NC_4partite_2} with $Y = B$ or $C$, and is thus causally nonseparable. This can be certified by the causal witness (obtained as described in Sec.~\ref{subsec:witnesses}, with the characterisation of Eq.~\eqref{eq:Scone_N4_no_AI_no_DO} in Appendix~\ref{app:witness_examples})
\begin{align}
S^\text{gap} &\coloneqq \frac18 \Big[\id^{\otimes 6} - \hat{\textsc{z}} (\id\hat{\textsc{z}}\hat{\textsc{z}}\id + \hat{\textsc{z}}\id\hat{\textsc{x}}\hat{\textsc{z}})\id\Big] , 
\end{align}
for which we obtain $\Tr[W^\text{gap}\cdot S^\text{gap}] = 1-\sqrt{2} < 0$.

This shows that, in the general multipartite case, there is indeed a gap between the necessary conditions obtained by teleporting to just some of the parties and those obtained by teleporting to each of the parties, and that the former are not sufficient.

(Note that in the example above $D_I$ did not play any role, as we always had $\id^{D_I}$ in all terms. We could in fact consider the case where $D_I$ is also trivial, $d_{D_I}=1$. We kept here a nontrivial system $D_I$ to clarify the fact that $W^\text{gap}$ was defined in a fourpartite scenario, and that party $D$ does play a role in the argument.)

\subsection{Numerically investigating the (in)equivalence of our necessary and sufficient conditions}
\label{app:numerical_search}

In order to investigate whether the (full version of the) necessary condition in Proposition~\ref{prop:CS_necessary} and the sufficient condition in Proposition~\ref{prop:CS_sufficient} differ in general, we conducted numerical testing to see whether we could find process matrices contained in the cone $\W_+^\sep$ but not in $\W_-^\sep$ (i.e.\ the outer and inner approximations of $\W^\sep$ arising from the necessary and sufficient conditions, respectively).
To this end, we considered the following general approach: we first generated a large number of random process matrices.
For each process matrix $W$, we then solved the primal SDP optimisation problem~\eqref{eq:primalSDP} over the cones $\W_\pm^\sep$ to obtain the corresponding random robustnesses $r^*_\pm$.
If we were to find $r^*_+ \neq r^*_-$ (up to numerical error; note that, since $\W^\sep_- \subseteq \W^\sep_+$, one always has $r^*_+ \le r^*_-$), this would imply the cones differ since one would have $W + r^*_+ \id^\circ \in \W^\sep_+$ but $W + r^*_+ \id^\circ \notin W^\sep_-$.

The size of the SDP problems associated with finding the random robustness of a process matrix meant that we could not solve these problems for the ``complete'' fourpartite scenario with qubit incoming and outgoing spaces for each party (recall that, for three parties, the conditions are already known to coincide).
We therefore considered the restricted scenario in which $d_{A_I}=1$ while the remaining Hilbert spaces are two-dimensional, so that $W$ is thus $(128{\times}128)$-dimensional.
We note that in any simpler scenario, the necessary and sufficient conditions can be be proven to coincide, making this the simplest case of interest.
Indeed, in Appendix~\ref{app:particular_4partite} we already showed that they coincide if one of the four parties has a trivial outgoing space.
If, on the other hand, a second party were to have a trivial incoming space (e.g., $d_{B_I}=1$), it is not difficult to show they again coincide by writing explicitly the necessary and sufficient conditions of Propositions~\ref{prop:CS_necessary} and~\ref{prop:CS_sufficient}, by using the fact that they simplify to Proposition~\ref{prop:CS_charact_3_no_AI} in a tripartite case where (at least) one party has a trivial incoming space, and by using the linearity of the subspaces appearing in the constraints.
We leave the explicit proof of this as an exercise for the reader.

To generate random process matrices, one could follow the hit-and-run approach of Ref.~\cite{feix16}.
Although this approach is guaranteed to sample process matrices uniformly, the high dimensionality of the space of valid process matrices (in this scenario it is $7597$-dimensional) renders this approach intractable.
Instead, forgoing uniformity, we generated matrices by randomly sampling Hermitian positive semidefinite matrices, projecting them onto the space $\L^\N$ of valid process matrices before adding white noise (i.e., $\id^\circ$) until the resulting matrix was again positive semidefinite.

We solved the SDP optimisation problems for the necessary and sufficient conditions for approximately 1000 randomly generated process matrices (including several hundreds in which an additional constraint, namely the symmetry of $W$ between permutations of the parties $B$, $C$ and $D$, was imposed).
These numerical tests failed to provide any potential counterexamples: in all cases we found $r^*_+ = r^*_-$ up to numerical precision.

However, since the space of valid process matrices is so high-dimensional and our sampling method non-uniform, we do not believe that our results on this number of samples provide enough evidence to reasonably conjecture that the necessary and sufficient conditions coincide in this scenario.

\section{Construction of witnesses of causal nonseparability through SDP}
\label{app:witnesses_SDP}

In this appendix we give some further details relating to the construction of witnesses of causal nonseparability through SDP.
Firstly, we discuss the duality of the two SDP problems given in Sec.~\ref{subsec:witnesses}, showing that they are indeed dual and that the Strong Duality Theorem is satisfied.
We then give some additional details on how the characterisations of causal separability can be explicitly translated into SDP constraints in order to find witnesses in practice, giving some explicit examples that both illustrate this and, at the same time, allow the results in Sec.~\ref{subsec:examples} to be readily verified.

\subsection{Duality of SDP problems}
\label{app:witness_duality}

Since both the set of causally separable process matrices $\mathcal{W}^\text{sep}$ and its dual $\mathcal{S}=(\mathcal{W}^\text{sep})^*$ (or the inner and outer approximations $\mathcal{W}_\pm^\text{sep}$ of $\mathcal{W}^\text{sep}$ arising from Propositions~\ref{prop:CS_necessary} and~\ref{prop:CS_sufficient} and their respective duals $\mathcal{S}_\pm=(\mathcal{W}_\pm^\text{sep})^*$, see Sec.~\ref{subsec:witnesses}) are convex cones, the problems of minimising the amount of white noise that must be added to make a process matrix causally separable and finding the witness of causal separability with the most negative value for a given process matrix can be formulated as SDP problems as in Eqs.~\eqref{eq:primalSDP} and~\eqref{eq:dualSDP}, respectively.
For these problems to be efficiently solvable with standard algorithmic techniques for SDP, however, one must show that they have no duality gap (i.e., no difference between the optimal values of an SDP problem and its dual).
Here, we will show that Eqs.~\eqref{eq:primalSDP} and~\eqref{eq:dualSDP} are indeed a primal-dual pair, and that the Strong Duality Theorem holds~\cite{nesterov94}, implying that that their optimal solutions indeed coincide and can therefore be efficiently obtained. This shows, in particular, that the solution to the SDP problem~\eqref{eq:dualSDP} is the optimal witness with respect to the random robustness.

Ref.~\cite{araujo15} showed the duality of two variations of the SDP problems~\eqref{eq:primalSDP} and~\eqref{eq:dualSDP} in the bipartite case: rather than consider the robustness to white noise of a process matrix, they considered the robustness of mixing a given $W$ with any valid process matrix. The optimal solutions to the corresponding SDP problems give the generalised robustness of $W$.
Nevertheless, their approach to proving duality, and the applicability of the Strong Duality Theorem, is easily adapted to (and even simpler for) the random robustness, and the bipartite and some restricted tripartite versions of Eqs.~\eqref{eq:primalSDP} and~\eqref{eq:dualSDP} were already given in Ref.~\cite{branciard16a}.
The same approach can be used in the more general multipartite case to show that these problems (considering the cones $\mathcal{W}^\text{sep}$ or $\mathcal{W}_\pm^\text{sep}$) satisfy the required properties.
Rather than repeating these (somewhat technical and lengthy) arguments, we instead refer the reader to Appendix~E of Ref.~\cite{araujo15} and prove explicitly only the main technical lemma needed to generalise their approach.

\medskip

First, as noted already in~\cite{araujo15,branciard16a}, it is sufficient just to consider the restriction $\mathcal{S}^\N\coloneqq \mathcal{S}\cap\mathcal{L}^\N$ of witnesses in $\mathcal{L}^\N$.
Indeed, for any $S^\perp$ in the orthogonal subspace $(\mathcal{L}^\N)^\perp$ of $\mathcal{L}^\N$ and any process matrix $W$ one has $\Tr[S^\perp \cdot W]=0$, and thus for any $S\in\mathcal{S}$ there exists $S'\in\mathcal{S}^\N$ such that $\Tr[S\cdot W]=\Tr[S'\cdot W]$ for all $W\in\mathcal{L}^\N$.
The formulations given in Eqs.~\eqref{eq:primalSDP} and~\eqref{eq:dualSDP} are only formally dual when $\mathcal{W}^\sep$ and $\mathcal{S}$ are considered as subsets of the vector space $\mathcal{L}^\N$ [or when $\mathcal{S}$ is replaced by $\mathcal{S}^\N$ in Eq.~\eqref{eq:dualSDP}].
However, the fact that the restriction to $\mathcal{S}^\N$ does not change the optimal value of the problem ensures that the optimal solutions coincide in the more general formulation.

The primary element of the proof in Ref.~\cite{araujo15} which needs to be generalised beyond two parties is the need to show that $\mathcal{W}^\sep$ has a nonempty interior (within $\L^\N$, cf.\ their Lemma~7; we also need to check that it is pointed, which is trivial).
To this end, it is sufficient to show that the white noise process matrix $\id^\circ$ is in the interior of $\mathcal{W}^\sep$, i.e., that there exists $\varepsilon>0$ such that for any $W \in \L^\N$ with $\lVert W \rVert_\text{HS} < \varepsilon$ (where $\lVert\cdot\rVert_\text{HS}$ is the Hilbert-Schmidt norm), one has $\id^\circ + W \in \mathcal{W}^\sep$.

Recalling from Appendix~\ref{app:allowed_forbidden_terms} the characterisation of $\mathcal{L}^\N$ in terms of ``allowed'' terms in a Hilbert-Schmidt basis decomposition, let us first note that any allowed Hilbert-Schmidt term $T_k$ which contains $\cdots \sigma_{\mu_k}^{A_I^k}\id^{A_O^k}\cdots$ (with $\sigma^{A_I^k}_{\mu_k}\neq \id$) for some $k\in \N$ is compatible with any fixed causal order where party $A_k$ comes last---i.e., that $T_k \in \L^{A_{\pi_k(1)}\prec\cdots\prec A_{\pi_k(N-1)}\prec A_k}$ for any permutation $\pi_k$ of parties such that $\pi_k(N)=k$ (the same also trivially holds for the allowed identity term $\id^\N$). Indeed, $_{[1-A_O^k]}T_k=0$ and $_{A_{IO}^k}T_k=0$, so that Eq.~\eqref{eq:constr_order_A1__AN} holds for any such order.
It follows that any $W\in\mathcal{L}^\N$ can be written as $W=\sum_{k=1}^N \Omega_k$, where each $\Omega_k\in\mathcal{L}^{A_{\pi_k(1)}\prec\cdots\prec A_{\pi_k(N-1)}\prec A_k}$ (for some arbitrary $\pi_k$ for each $k$); furthermore, the terms $\Omega_k$ can be taken to be orthogonal, so that $\lVert W \rVert_\text{HS}^2 = \sum_{k=1}^N \lVert \Omega_k \rVert_\text{HS}^2$.

Note that the $\Omega_k$'s may not, in general, be positive semidefinite. Nevertheless, if we take $W$ such that $\lVert W \rVert_\text{HS} < \varepsilon \coloneqq \frac{1}{N d_I}$, with $d_I \coloneqq \prod_{k \in \N} d_{A_I^k}$, then each $\lVert \Omega_k \rVert \leq \lVert \Omega_k \rVert_\text{HS} < \frac{1}{N d_I}$ (where $\lVert\cdot\rVert$ is now the spectral norm), so that $\frac{1}{N d_I}\id + \Omega_k \ge 0$. For any such $W$, we thus obtain a decomposition
\begin{equation}
	\id^\circ + W = \sum_{k=1}^N \Big( \frac{1}{N d_I}\id + \Omega_k \Big)
\end{equation}
with $\frac{1}{N d_I}\id + \Omega_k \in \P \cap \mathcal{L}^{A_{\pi_k(1)}\prec\cdots\prec A_{\pi_k(N-1)}\prec A_k}$, which proves that $\id^\circ + W$ is the sum of (valid) process matrices compatible with fixed causal orders, and hence is causally separable: $\id^\circ + W \in \mathcal{W}^\sep$, as desired.

With this verified, the approach of Ref.~\cite{araujo15} can be applied, with the appropriate modifications for the random robustness,%
\footnote{Namely, one can change Eqs.~(E.3)--(E.7) in Ref.~\cite{araujo15} to $E=\L^\N$, $\K=\mathcal{W}^\sep$, $\L = \{r \id^\circ\mid r \in \mathbb{R}\}$, $b=W$ and $c=\id/\prod_{k \in \N} d_{A_O^k}$ (using their notations for $E,\K,\L,b,c$) and then adapt the proof accordingly.}
to show that the required duality indeed holds and that the conditions of the Strong Duality Theorem are satisfied.

\subsection{Explicit SDP constraints and example constructions}
\label{app:witness_examples}

In order to characterise $\mathcal{S}$ more explicitly for a given scenario, as well as to solve both the primal and dual SDP problem using convex optimisation algorithms~\cite{cvx}, it is helpful to write $\mathcal{W}^\sep$ explicitly as intersections and Minkowski sums of convex cones corresponding to individual constraints on causally separable process matrices.
The duality relations~\eqref{eq:duality_relations} can then be exploited to describe $\mathcal{S}$.
Here we give some examples to illustrate this procedure.

\medskip

The simplest example is the bipartite scenario.
From the definition in Eq.~\eqref{def:caus_sep_bi} we see that $\mathcal{W}^\sep=\mathcal{W}^{A\prec B}+\mathcal{W}^{B\prec A}$, where $\mathcal{W}^{A\prec B}=\mathcal{P} \cap \mathcal{L}^{A\prec B}$ and similarly for $\mathcal{W}^{B\prec A}$.
Using Eq.~\eqref{eq:constr_order_A1__AN} to write $\mathcal{L}^{A\prec B}$ and $\mathcal{L}^{B\prec A}$ in terms of spaces defined by individual linear constraints, or directly referring to Proposition~\ref{prop:CS_charact_2}, we see that
\begin{align}
\W^\text{sep} &= {\cal P} \cap \L_{[1-A_O]B_{IO}} \cap \L_{[1-B_O]} \notag \\
& \quad + {\cal P} \cap \L_{[1-B_O]A_{IO}} \cap \L_{[1-A_O]} \,,
\end{align}
with $\L_{[1-A_O]B_{IO}} \coloneqq \{ W \in A_{IO}B_{IO} \mid {}_{[1-A_O]B_{IO}}W = 0 \}$, $\L_{[1-B_O]} \coloneqq \{ W \in A_{IO}B_{IO} \mid {}_{[1-B_O]}W = 0 \}$, etc.
It follows that
\begin{align}
{\cal S} = (\W^\text{sep})^* &= ({\cal P} + \L_{[1-A_O]B_{IO}}^\perp + \L_{[1-B_O]}^\perp) \notag \\
& \quad \cap ({\cal P} + \L_{[1-B_O]A_{IO}}^\perp + \L_{[1-A_O]}^\perp) \,, \label{eq:Scone_N2}
\end{align}
where we used the fact that ${\cal P}$ is self-dual, and where $\L_{[1-A_O]B_{IO}}^\perp = \{ S \in A_{IO}B_{IO} \mid {}_{[1-A_O]B_{IO}}S = S \}$ is the orthogonal subspace of $\L_{[1-A_O]B_{IO}}$, $\L_{[1-B_O]}^\perp = \{ S \in A_{IO}B_{IO} \mid {}_{[1-B_O]}S = S \}$ is the orthogonal subspace of $\L_{[1-B_O]}$, etc.

Note that a slightly different, but equivalent, characterisation was given for the bipartite scenario in Refs.~\cite{araujo15,branciard16a}.
Although their formulation is slightly simpler, we choose to give the above form as it shows more clearly the procedure of obtaining explicit SDP characterisations from the characterisations of causally separable process matrices given in the main text, and it generalises more directly to the multipartite scenario.

\medskip

The next simplest case is the tripartite scenario with $d_{C_O}=1$. 
In this case, causally separable process matrices are characterised by Proposition~\ref{prop:CS_charact_3_no_CO}, from which it follows that
\begin{align}
\W^\text{sep} &= {\cal P} \cap \L_{[1-A_O]B_{IO}C_I} \cap \L_{[1-B_O]C_I} \notag \\
& \quad + {\cal P} \cap \L_{[1-B_O]A_{IO}C_I} \cap \L_{[1-A_O]C_I} \,,
\end{align}
with similar notations for $\L_{[1-A_O]B_{IO}C_I}$, $\L_{[1-B_O]C_I}$, etc. as before.
Similarly to the bipartite case, this leads to
\begin{align}
{\cal S} &= ({\cal P} + \L_{[1-A_O]B_{IO}C_I}^\perp + \L_{[1-B_O]C_I}^\perp) \notag \\
& \quad \ \cap ({\cal P} + \L_{[1-B_O]A_{IO}C_I}^\perp + \L_{[1-A_O]C_I}^\perp) \,. \label{eq:Scone_N3_no_CO}
\end{align}
We note again that two slightly different, but once again equivalent, characterisations were given in Refs.~\cite{araujo15,branciard16a} for this particular tripartite case.

\medskip

In the tripartite scenario with $d_{A_I}=1$ instead (as, e.g., in the example of ``activation of causal nonseparability'' given by Oreshkov and Giarmatzi~\cite{oreshkov16}), Proposition~\ref{prop:CS_charact_3_no_AI} leads to
\begin{align}
\W^\text{sep} =& \, \L_{[1-A_O]B_{IO}C_{IO}} \notag \\
& \ \cap \big[ {\cal P} \cap \L_{[1-B_O]C_{IO}} \cap \L_{[1-C_O]} \notag \\
& \qquad + {\cal P} \cap \L_{[1-C_O]B_{IO}} \cap \L_{[1-B_O]} \big] \,.
\end{align}
It follows that
\begin{align}
{\cal S} =& \, \L_{[1-A_O]B_{IO}C_{IO}}^\perp \notag \\
& \ + ( {\cal P} + \L_{[1-B_O]C_{IO}}^\perp + \L_{[1-C_O]}^\perp ) \notag \\
& \qquad \cap ( {\cal P} + \L_{[1-C_O]B_{IO}}^\perp + \L_{[1-B_O]}^\perp ) \,. \label{eq:Scone_N3_no_AI}
\end{align}

\medskip

In the general tripartite case, the characterisation of Proposition~\ref{prop:CS_charact_3} shows that we can write $\mathcal{W}^\sep$ as
\begin{align}\label{eq:Wsep_N3}
\W^\text{sep} &= \L_{[1-A_O]B_{IO}C_{IO}} \cap \big( {\cal P} \cap \L_{[1-B_O]C_{IO}} \cap \L_{[1-C_O]} \notag \\ 
& \hspace{32mm} + {\cal P} \cap \L_{[1-C_O]B_{IO}} \cap \L_{[1-B_O]} \big) \notag \\
& \ + \L_{[1-B_O]A_{IO}C_{IO}} \cap \big( {\cal P} \cap \L_{[1-A_O]C_{IO}} \cap \L_{[1-C_O]} \notag \\ 
& \hspace{32mm} + {\cal P} \cap \L_{[1-C_O]A_{IO}} \cap \L_{[1-A_O]} \big) \notag \\
& \ + \L_{[1-C_O]A_{IO}B_{IO}} \cap \big( {\cal P} \cap \L_{[1-A_O]B_{IO}} \cap \L_{[1-B_O]} \notag \\ 
& \hspace{32mm} + {\cal P} \cap \L_{[1-B_O]A_{IO}} \cap \L_{[1-A_O]} \big) ,
\end{align}
from which it follows that cone of witnesses is
\begin{align}\label{eq:Scone_N3}
{\cal S} &= \Big( \L_{[1-A_O]B_{IO}C_{IO}}^\perp + ({\cal P} + \L_{[1-B_O]C_{IO}}^\perp + \L_{[1-C_O]}^\perp) \notag \\[-2mm]
& \hspace{33mm} \cap ({\cal P} + \L_{[1-C_O]B_{IO}}^\perp + \L_{[1-B_O]}^\perp) \Big) \notag \\[-1mm]
& \ \cap \Big( \L_{[1-B_O]A_{IO}C_{IO}}^\perp + ({\cal P} + \L_{[1-A_O]C_{IO}}^\perp + \L_{[1-C_O]}^\perp) \notag \\[-2mm]
& \hspace{33mm} \cap ({\cal P} + \L_{[1-C_O]A_{IO}}^\perp + \L_{[1-A_O]}^\perp) \Big) \notag \\[-1mm]
& \ \cap \Big( \L_{[1-C_O]A_{IO}B_{IO}}^\perp + ({\cal P} + \L_{[1-A_O]B_{IO}}^\perp + \L_{[1-B_O]}^\perp) \notag \\[-2mm]
& \hspace{33mm} \cap ({\cal P} + \L_{[1-B_O]A_{IO}}^\perp + \L_{[1-A_O]}^\perp) \Big).
\end{align}

\medskip

Finally, in the fourpartite scenario with $d_{A_I} = d_{D_O} = 1$ (as, e.g., in the version of the quantum switch in Eq.~\eqref{eq:4partiteswitch}), Proposition~\ref{prop:CS_charact_4_no_AI_no_DO} leads to
\begin{align}
\W^\text{sep} =& \, \L_{[1-A_O]B_{IO}C_{IO}D_I} \notag \\
& \ \cap \big[ {\cal P} \cap \L_{[1-B_O]C_{IO}D_I} \cap \L_{[1-C_O]D_I} \notag \\
& \qquad + {\cal P} \cap \L_{[1-C_O]B_{IO}D_I} \cap \L_{[1-B_O]D_I} \big]
\end{align}
and
\begin{align}
{\cal S} =& \, \L_{[1-A_O]B_{IO}C_{IO}D_I}^\perp \notag \\
& \ + ( {\cal P} + \L_{[1-B_O]C_{IO}D_I}^\perp + \L_{[1-C_O]D_I}^\perp ) \notag \\
& \qquad \cap ( {\cal P} + \L_{[1-C_O]B_{IO}D_I}^\perp + \L_{[1-B_O]D_I}^\perp ) \,. \label{eq:Scone_N4_no_AI_no_DO}
\end{align}

\medskip

Once again, for more general scenarios it remains an open question whether the necessary and sufficient conditions of Propositions~\ref{prop:CS_necessary} and~\ref{prop:CS_sufficient} coincide.
Nonetheless, the same approach here can be applied to our necessary condition, which defines the cone $\mathcal{W}_+^\sep$ that is an outer approximation of $\mathcal{W}^\sep$, to characterise a subset $\mathcal{S}_+$ of causal witnesses.
Solving the dual SDP problem~\eqref{eq:dualSDP} over this set allows one to find valid witnesses of causal nonseparability for a given process matrix $W$, even though (without proof that the necessary and sufficient conditions coincide) such a witness may not be optimal amongst the full set of causal witnesses $\mathcal{S}$.

\bibliography{bib_multi_causal_nonsep}

\begin{thebibliography}{49}%
\makeatletter
\providecommand \@ifxundefined [1]{%
 \@ifx{#1\undefined}
}%
\providecommand \@ifnum [1]{%
 \ifnum #1\expandafter \@firstoftwo
 \else \expandafter \@secondoftwo
 \fi
}%
\providecommand \@ifx [1]{%
 \ifx #1\expandafter \@firstoftwo
 \else \expandafter \@secondoftwo
 \fi
}%
\providecommand \natexlab [1]{#1}%
\providecommand \enquote  [1]{``#1''}%
\providecommand \bibnamefont  [1]{#1}%
\providecommand \bibfnamefont [1]{#1}%
\providecommand \citenamefont [1]{#1}%
\providecommand \href@noop [0]{\@secondoftwo}%
\providecommand \href [0]{\begingroup \@sanitize@url \@href}%
\providecommand \@href[1]{\@@startlink{#1}\@@href}%
\providecommand \@@href[1]{\endgroup#1\@@endlink}%
\providecommand \@sanitize@url [0]{\catcode `\\12\catcode `\$12\catcode
  `\&12\catcode `\#12\catcode `\^12\catcode `\_12\catcode `\%12\relax}%
\providecommand \@@startlink[1]{}%
\providecommand \@@endlink[0]{}%
\providecommand \url  [0]{\begingroup\@sanitize@url \@url }%
\providecommand \@url [1]{\endgroup\@href {#1}{\urlprefix }}%
\providecommand \urlprefix  [0]{URL }%
\providecommand \Eprint [0]{\href }%
\providecommand \doibase [0]{http://dx.doi.org/}%
\providecommand \selectlanguage [0]{\@gobble}%
\providecommand \bibinfo  [0]{\@secondoftwo}%
\providecommand \bibfield  [0]{\@secondoftwo}%
\providecommand \translation [1]{[#1]}%
\providecommand \BibitemOpen [0]{}%
\providecommand \bibitemStop [0]{}%
\providecommand \bibitemNoStop [0]{.\EOS\space}%
\providecommand \EOS [0]{\spacefactor3000\relax}%
\providecommand \BibitemShut  [1]{\csname bibitem#1\endcsname}%
\let\auto@bib@innerbib\@empty
\bibitem [{\citenamefont {Hardy}(2005)}]{hardy05}%
  \BibitemOpen
  \bibfield  {author} {\bibinfo {author} {\bibfnamefont {L.}~\bibnamefont
  {Hardy}},\ }\href@noop {} {\  (\bibinfo {year} {2005})},\ \Eprint
  {http://arxiv.org/abs/gr-qc/0509120} {arXiv:gr-qc/0509120} \BibitemShut
  {NoStop}%
\bibitem [{\citenamefont {Oreshkov}\ \emph {et~al.}(2012)\citenamefont
  {Oreshkov}, \citenamefont {Costa},\ and\ \citenamefont
  {Brukner}}]{oreshkov12}%
  \BibitemOpen
  \bibfield  {author} {\bibinfo {author} {\bibfnamefont {O.}~\bibnamefont
  {Oreshkov}}, \bibinfo {author} {\bibfnamefont {F.}~\bibnamefont {Costa}}, \
  and\ \bibinfo {author} {\bibfnamefont {{\v{C}}.}~\bibnamefont {Brukner}},\
  }\href {\doibase 10.1038/ncomms2076} {\bibfield  {journal} {\bibinfo
  {journal} {Nat. Commun.}\ }\textbf {\bibinfo {volume} {3}},\ \bibinfo {pages}
  {1092} (\bibinfo {year} {2012})},\ \Eprint {http://arxiv.org/abs/1105.4464}
  {arXiv:1105.4464 [quant-ph]} \BibitemShut {NoStop}%
\bibitem [{\citenamefont {Zych}\ \emph {et~al.}(2017)\citenamefont {Zych},
  \citenamefont {Costa}, \citenamefont {Pikovski},\ and\ \citenamefont
  {Brukner}}]{zych17}%
  \BibitemOpen
  \bibfield  {author} {\bibinfo {author} {\bibfnamefont {M.}~\bibnamefont
  {Zych}}, \bibinfo {author} {\bibfnamefont {F.}~\bibnamefont {Costa}},
  \bibinfo {author} {\bibfnamefont {I.}~\bibnamefont {Pikovski}}, \ and\
  \bibinfo {author} {\bibfnamefont {{\v{C}}.}~\bibnamefont {Brukner}},\
  }\href@noop {} {\  (\bibinfo {year} {2017})},\ \Eprint
  {http://arxiv.org/abs/1708.00248} {arXiv:1708.00248 [quant-ph]} \BibitemShut
  {NoStop}%
\bibitem [{\citenamefont {Chiribella}\ \emph {et~al.}(2013)\citenamefont
  {Chiribella}, \citenamefont {D'Ariano}, \citenamefont {Perinotti},\ and\
  \citenamefont {Valiron}}]{chiribella13}%
  \BibitemOpen
  \bibfield  {author} {\bibinfo {author} {\bibfnamefont {G.}~\bibnamefont
  {Chiribella}}, \bibinfo {author} {\bibfnamefont {G.~M.}\ \bibnamefont
  {D'Ariano}}, \bibinfo {author} {\bibfnamefont {P.}~\bibnamefont {Perinotti}},
  \ and\ \bibinfo {author} {\bibfnamefont {B.}~\bibnamefont {Valiron}},\ }\href
  {\doibase 10.1103/PhysRevA.88.022318} {\bibfield  {journal} {\bibinfo
  {journal} {Phys. Rev. A}\ }\textbf {\bibinfo {volume} {88}},\ \bibinfo
  {pages} {022318} (\bibinfo {year} {2013})},\ \Eprint
  {http://arxiv.org/abs/0912.0195} {arXiv:0912.0195 [quant-ph]} \BibitemShut
  {NoStop}%
\bibitem [{\citenamefont {Chiribella}(2012)}]{chiribella12}%
  \BibitemOpen
  \bibfield  {author} {\bibinfo {author} {\bibfnamefont {G.}~\bibnamefont
  {Chiribella}},\ }\href {\doibase 10.1103/PhysRevA.86.040301} {\bibfield
  {journal} {\bibinfo  {journal} {Phys. Rev. A}\ }\textbf {\bibinfo {volume}
  {86}},\ \bibinfo {pages} {040301} (\bibinfo {year} {2012})},\ \Eprint
  {http://arxiv.org/abs/1109.5154} {arXiv:1109.5154 [quant-ph]} \BibitemShut
  {NoStop}%
\bibitem [{\citenamefont {Colnaghi}\ \emph {et~al.}(2012)\citenamefont
  {Colnaghi}, \citenamefont {D'Ariano}, \citenamefont {Facchini},\ and\
  \citenamefont {Perinotti}}]{colnaghi12}%
  \BibitemOpen
  \bibfield  {author} {\bibinfo {author} {\bibfnamefont {T.}~\bibnamefont
  {Colnaghi}}, \bibinfo {author} {\bibfnamefont {G.~M.}\ \bibnamefont
  {D'Ariano}}, \bibinfo {author} {\bibfnamefont {S.}~\bibnamefont {Facchini}},
  \ and\ \bibinfo {author} {\bibfnamefont {P.}~\bibnamefont {Perinotti}},\
  }\href {\doibase 10.1016/j.physleta.2012.08.028} {\bibfield  {journal}
  {\bibinfo  {journal} {Phys. Lett. A}\ }\textbf {\bibinfo {volume} {376}},\
  \bibinfo {pages} {2940} (\bibinfo {year} {2012})},\ \Eprint
  {http://arxiv.org/abs/1109.5987} {arXiv:1109.5987 [quant-ph]} \BibitemShut
  {NoStop}%
\bibitem [{\citenamefont {{Ara{\'u}jo}}\ \emph {et~al.}(2014)\citenamefont
  {{Ara{\'u}jo}}, \citenamefont {{Costa}},\ and\ \citenamefont
  {{Brukner}}}]{araujo14}%
  \BibitemOpen
  \bibfield  {author} {\bibinfo {author} {\bibfnamefont {M.}~\bibnamefont
  {{Ara{\'u}jo}}}, \bibinfo {author} {\bibfnamefont {F.}~\bibnamefont
  {{Costa}}}, \ and\ \bibinfo {author} {\bibfnamefont {{\v C}.}~\bibnamefont
  {{Brukner}}},\ }\href {\doibase 10.1103/PhysRevLett.113.250402} {\bibfield
  {journal} {\bibinfo  {journal} {Phys. Rev. Lett.}\ }\textbf {\bibinfo
  {volume} {113}},\ \bibinfo {pages} {250402} (\bibinfo {year} {2014})},\
  \Eprint {http://arxiv.org/abs/1401.8127} {arXiv:1401.8127 [quant-ph]}
  \BibitemShut {NoStop}%
\bibitem [{\citenamefont {Facchini}\ and\ \citenamefont
  {Perdrix}(2015)}]{facchini15}%
  \BibitemOpen
  \bibfield  {author} {\bibinfo {author} {\bibfnamefont {S.}~\bibnamefont
  {Facchini}}\ and\ \bibinfo {author} {\bibfnamefont {S.}~\bibnamefont
  {Perdrix}},\ }in\ \href@noop {} {\emph {\bibinfo {booktitle} {Theory and
  Applications of Models of Computation}}},\ \bibinfo {editor} {edited by\
  \bibinfo {editor} {\bibfnamefont {R.}~\bibnamefont {Jain}}, \bibinfo {editor}
  {\bibfnamefont {S.}~\bibnamefont {Jain}}, \ and\ \bibinfo {editor}
  {\bibfnamefont {F.}~\bibnamefont {Stephan}}}\ (\bibinfo  {publisher}
  {Springer International Publishing},\ \bibinfo {address} {Cham},\ \bibinfo
  {year} {2015})\ pp.\ \bibinfo {pages} {324--331},\ \Eprint
  {http://arxiv.org/abs/1405.5205} {arXiv:1405.5205 [quant-ph]} \BibitemShut
  {NoStop}%
\bibitem [{\citenamefont {Feix}\ \emph {et~al.}(2015)\citenamefont {Feix},
  \citenamefont {Ara\'ujo},\ and\ \citenamefont {Brukner}}]{feix15}%
  \BibitemOpen
  \bibfield  {author} {\bibinfo {author} {\bibfnamefont {A.}~\bibnamefont
  {Feix}}, \bibinfo {author} {\bibfnamefont {M.}~\bibnamefont {Ara\'ujo}}, \
  and\ \bibinfo {author} {\bibfnamefont {{\v C}.}~\bibnamefont {Brukner}},\
  }\href {\doibase 10.1103/PhysRevA.92.052326} {\bibfield  {journal} {\bibinfo
  {journal} {Phys. Rev. A}\ }\textbf {\bibinfo {volume} {92}},\ \bibinfo
  {pages} {052326} (\bibinfo {year} {2015})},\ \Eprint
  {http://arxiv.org/abs/1508.07840} {arXiv:1508.07840 [quant-ph]} \BibitemShut
  {NoStop}%
\bibitem [{\citenamefont {Gu\'erin}\ \emph {et~al.}(2016)\citenamefont
  {Gu\'erin}, \citenamefont {Feix}, \citenamefont {Ara\'ujo},\ and\
  \citenamefont {Brukner}}]{guerin16}%
  \BibitemOpen
  \bibfield  {author} {\bibinfo {author} {\bibfnamefont {P.~A.}\ \bibnamefont
  {Gu\'erin}}, \bibinfo {author} {\bibfnamefont {A.}~\bibnamefont {Feix}},
  \bibinfo {author} {\bibfnamefont {M.}~\bibnamefont {Ara\'ujo}}, \ and\
  \bibinfo {author} {\bibfnamefont {{\v C}.}~\bibnamefont {Brukner}},\ }\href
  {\doibase 10.1103/PhysRevLett.117.100502} {\bibfield  {journal} {\bibinfo
  {journal} {Phys. Rev. Lett.}\ }\textbf {\bibinfo {volume} {117}},\ \bibinfo
  {pages} {100502} (\bibinfo {year} {2016})},\ \Eprint
  {http://arxiv.org/abs/1605.07372} {arXiv:1605.07372 [quant-ph]} \BibitemShut
  {NoStop}%
\bibitem [{\citenamefont {Ebler}\ \emph {et~al.}(2018)\citenamefont {Ebler},
  \citenamefont {Salek},\ and\ \citenamefont {Chiribella}}]{ebler18}%
  \BibitemOpen
  \bibfield  {author} {\bibinfo {author} {\bibfnamefont {D.}~\bibnamefont
  {Ebler}}, \bibinfo {author} {\bibfnamefont {S.}~\bibnamefont {Salek}}, \ and\
  \bibinfo {author} {\bibfnamefont {G.}~\bibnamefont {Chiribella}},\ }\href
  {\doibase 10.1103/PhysRevLett.120.120502} {\bibfield  {journal} {\bibinfo
  {journal} {Phys. Rev. Lett.}\ }\textbf {\bibinfo {volume} {120}},\ \bibinfo
  {pages} {120502} (\bibinfo {year} {2018})},\ \Eprint
  {http://arxiv.org/abs/1711.10165} {arXiv:1711.10165 [quant-ph]} \BibitemShut
  {NoStop}%
\bibitem [{\citenamefont {Salek}\ \emph {et~al.}(2018)\citenamefont {Salek},
  \citenamefont {Ebler},\ and\ \citenamefont {Chiribella}}]{salek18}%
  \BibitemOpen
  \bibfield  {author} {\bibinfo {author} {\bibfnamefont {S.}~\bibnamefont
  {Salek}}, \bibinfo {author} {\bibfnamefont {D.}~\bibnamefont {Ebler}}, \ and\
  \bibinfo {author} {\bibfnamefont {G.}~\bibnamefont {Chiribella}},\
  }\href@noop {} {\  (\bibinfo {year} {2018})},\ \Eprint
  {http://arxiv.org/abs/1809.06655} {arXiv:1809.06655 [quant-ph]} \BibitemShut
  {NoStop}%
\bibitem [{\citenamefont {Chiribella}\ \emph {et~al.}(2018)\citenamefont
  {Chiribella}, \citenamefont {Banik}, \citenamefont {Bhattacharya},
  \citenamefont {Guha}, \citenamefont {Alimuddin}, \citenamefont {Roy},
  \citenamefont {Saha}, \citenamefont {Agrawal},\ and\ \citenamefont
  {Kar}}]{chiribella18}%
  \BibitemOpen
  \bibfield  {author} {\bibinfo {author} {\bibfnamefont {G.}~\bibnamefont
  {Chiribella}}, \bibinfo {author} {\bibfnamefont {M.}~\bibnamefont {Banik}},
  \bibinfo {author} {\bibfnamefont {S.~S.}\ \bibnamefont {Bhattacharya}},
  \bibinfo {author} {\bibfnamefont {T.}~\bibnamefont {Guha}}, \bibinfo {author}
  {\bibfnamefont {M.}~\bibnamefont {Alimuddin}}, \bibinfo {author}
  {\bibfnamefont {A.}~\bibnamefont {Roy}}, \bibinfo {author} {\bibfnamefont
  {S.}~\bibnamefont {Saha}}, \bibinfo {author} {\bibfnamefont {S.}~\bibnamefont
  {Agrawal}}, \ and\ \bibinfo {author} {\bibfnamefont {G.}~\bibnamefont
  {Kar}},\ }\href@noop {} {\  (\bibinfo {year} {2018})},\ \Eprint
  {http://arxiv.org/abs/1810.10457} {arXiv:1810.10457 [quant-ph]} \BibitemShut
  {NoStop}%
\bibitem [{\citenamefont {Ara\'{u}jo}\ \emph {et~al.}(2015)\citenamefont
  {Ara\'{u}jo}, \citenamefont {Branciard}, \citenamefont {Costa}, \citenamefont
  {Feix}, \citenamefont {Giarmatzi},\ and\ \citenamefont {Brukner}}]{araujo15}%
  \BibitemOpen
  \bibfield  {author} {\bibinfo {author} {\bibfnamefont {M.}~\bibnamefont
  {Ara\'{u}jo}}, \bibinfo {author} {\bibfnamefont {C.}~\bibnamefont
  {Branciard}}, \bibinfo {author} {\bibfnamefont {F.}~\bibnamefont {Costa}},
  \bibinfo {author} {\bibfnamefont {A.}~\bibnamefont {Feix}}, \bibinfo {author}
  {\bibfnamefont {C.}~\bibnamefont {Giarmatzi}}, \ and\ \bibinfo {author}
  {\bibfnamefont {{\v{C}}.}~\bibnamefont {Brukner}},\ }\href {\doibase
  10.1088/1367-2630/17/10/102001} {\bibfield  {journal} {\bibinfo  {journal}
  {New J. Phys.}\ }\textbf {\bibinfo {volume} {17}},\ \bibinfo {pages} {102001}
  (\bibinfo {year} {2015})},\ \Eprint {http://arxiv.org/abs/1506.03776}
  {arXiv:1506.03776 [quant-ph]} \BibitemShut {NoStop}%
\bibitem [{\citenamefont {Branciard}(2016)}]{branciard16a}%
  \BibitemOpen
  \bibfield  {author} {\bibinfo {author} {\bibfnamefont {C.}~\bibnamefont
  {Branciard}},\ }\href {http://dx.doi.org/10.1038/srep26018} {\bibfield
  {journal} {\bibinfo  {journal} {Sci. Rep.}\ }\textbf {\bibinfo {volume}
  {6}},\ \bibinfo {pages} {26018} (\bibinfo {year} {2016})},\ \Eprint
  {http://arxiv.org/abs/1603.00043} {arXiv:1603.00043 [quant-ph]} \BibitemShut
  {NoStop}%
\bibitem [{\citenamefont {Baumeler}\ and\ \citenamefont
  {Wolf}(2014)}]{baumeler14}%
  \BibitemOpen
  \bibfield  {author} {\bibinfo {author} {\bibfnamefont {{\"{A}}.}~\bibnamefont
  {Baumeler}}\ and\ \bibinfo {author} {\bibfnamefont {S.}~\bibnamefont
  {Wolf}},\ }in\ \href {\doibase 10.1109/ISIT.2014.6874888} {\emph {\bibinfo
  {booktitle} {2014 IEEE Int. Symp. Inf. Theory (ISIT)}}}\ (\bibinfo {year}
  {2014})\ pp.\ \bibinfo {pages} {526--530},\ \Eprint
  {http://arxiv.org/abs/1312.5916} {arXiv:1312.5916 [quant-ph]} \BibitemShut
  {NoStop}%
\bibitem [{\citenamefont {Oreshkov}\ and\ \citenamefont
  {Giarmatzi}(2016)}]{oreshkov16}%
  \BibitemOpen
  \bibfield  {author} {\bibinfo {author} {\bibfnamefont {O.}~\bibnamefont
  {Oreshkov}}\ and\ \bibinfo {author} {\bibfnamefont {C.}~\bibnamefont
  {Giarmatzi}},\ }\href {\doibase 10.1088/1367-2630/18/9/093020} {\bibfield
  {journal} {\bibinfo  {journal} {New J. Phys.}\ }\textbf {\bibinfo {volume}
  {18}},\ \bibinfo {pages} {093020} (\bibinfo {year} {2016})},\ \Eprint
  {http://arxiv.org/abs/1506.05449} {arXiv:1506.05449 [quant-ph]} \BibitemShut
  {NoStop}%
\bibitem [{\citenamefont {Giacomini}\ \emph {et~al.}(2016)\citenamefont
  {Giacomini}, \citenamefont {Castro-Ruiz},\ and\ \citenamefont
  {Brukner}}]{giacomini16}%
  \BibitemOpen
  \bibfield  {author} {\bibinfo {author} {\bibfnamefont {F.}~\bibnamefont
  {Giacomini}}, \bibinfo {author} {\bibfnamefont {E.}~\bibnamefont
  {Castro-Ruiz}}, \ and\ \bibinfo {author} {\bibfnamefont {{\v
  C}.}~\bibnamefont {Brukner}},\ }\href
  {http://stacks.iop.org/1367-2630/18/i=11/a=113026} {\bibfield  {journal}
  {\bibinfo  {journal} {New J. Phys.}\ }\textbf {\bibinfo {volume} {18}},\
  \bibinfo {pages} {113026} (\bibinfo {year} {2016})},\ \Eprint
  {http://arxiv.org/abs/1510.06345} {arXiv:1510.06345 [quant-ph]} \BibitemShut
  {NoStop}%
\bibitem [{\citenamefont {Davies}\ and\ \citenamefont
  {Lewis}(1970)}]{davies70}%
  \BibitemOpen
  \bibfield  {author} {\bibinfo {author} {\bibfnamefont {E.~B.}\ \bibnamefont
  {Davies}}\ and\ \bibinfo {author} {\bibfnamefont {J.~T.}\ \bibnamefont
  {Lewis}},\ }\href {\doibase 10.1007/BF01647093} {\bibfield  {journal}
  {\bibinfo  {journal} {Commun. Math. Phys.}\ }\textbf {\bibinfo {volume}
  {17}},\ \bibinfo {pages} {239} (\bibinfo {year} {1970})}\BibitemShut
  {NoStop}%
\bibitem [{\citenamefont {Jamio\l{}kowski}(1972)}]{jamiolkowski72}%
  \BibitemOpen
  \bibfield  {author} {\bibinfo {author} {\bibfnamefont {A.}~\bibnamefont
  {Jamio\l{}kowski}},\ }\href {\doibase 10.1016/0034-4877(72)90011-0}
  {\bibfield  {journal} {\bibinfo  {journal} {Rep. Math. Phys.}\ }\textbf
  {\bibinfo {volume} {3}},\ \bibinfo {pages} {275} (\bibinfo {year}
  {1972})}\BibitemShut {NoStop}%
\bibitem [{\citenamefont {Choi}(1975)}]{choi75}%
  \BibitemOpen
  \bibfield  {author} {\bibinfo {author} {\bibfnamefont {M.-D.}\ \bibnamefont
  {Choi}},\ }\href {\doibase 10.1016/0024-3795(75)90075-0} {\bibfield
  {journal} {\bibinfo  {journal} {Linear Algebra Appl.}\ }\textbf {\bibinfo
  {volume} {10}},\ \bibinfo {pages} {285} (\bibinfo {year} {1975})}\BibitemShut
  {NoStop}%
\bibitem [{\citenamefont {Abbott}\ \emph {et~al.}(2016)\citenamefont {Abbott},
  \citenamefont {Giarmatzi}, \citenamefont {Costa},\ and\ \citenamefont
  {Branciard}}]{abbott16}%
  \BibitemOpen
  \bibfield  {author} {\bibinfo {author} {\bibfnamefont {A.~A.}\ \bibnamefont
  {Abbott}}, \bibinfo {author} {\bibfnamefont {C.}~\bibnamefont {Giarmatzi}},
  \bibinfo {author} {\bibfnamefont {F.}~\bibnamefont {Costa}}, \ and\ \bibinfo
  {author} {\bibfnamefont {C.}~\bibnamefont {Branciard}},\ }\href {\doibase
  10.1103/PhysRevA.94.032131} {\bibfield  {journal} {\bibinfo  {journal} {Phys.
  Rev. A}\ }\textbf {\bibinfo {volume} {94}},\ \bibinfo {pages} {032131}
  (\bibinfo {year} {2016})},\ \Eprint {http://arxiv.org/abs/1608.01528}
  {arXiv:1608.01528 [quant-ph]} \BibitemShut {NoStop}%
\bibitem [{\citenamefont {{Chiribella}}\ \emph {et~al.}(2008)\citenamefont
  {{Chiribella}}, \citenamefont {{D'Ariano}},\ and\ \citenamefont
  {{Perinotti}}}]{chiribella08}%
  \BibitemOpen
  \bibfield  {author} {\bibinfo {author} {\bibfnamefont {G.}~\bibnamefont
  {{Chiribella}}}, \bibinfo {author} {\bibfnamefont {G.~M.}\ \bibnamefont
  {{D'Ariano}}}, \ and\ \bibinfo {author} {\bibfnamefont {P.}~\bibnamefont
  {{Perinotti}}},\ }\href {\doibase 10.1209/0295-5075/83/30004} {\bibfield
  {journal} {\bibinfo  {journal} {EPL}\ }\textbf {\bibinfo {volume} {83}},\
  \bibinfo {pages} {30004} (\bibinfo {year} {2008})},\ \Eprint
  {http://arxiv.org/abs/0804.0180} {arXiv:0804.0180 [quant-ph]} \BibitemShut
  {NoStop}%
\bibitem [{\citenamefont {Procopio}\ \emph {et~al.}(2015)\citenamefont
  {Procopio}, \citenamefont {Moqanaki}, \citenamefont {Ara\'{u}jo},
  \citenamefont {Costa}, \citenamefont {Alonso~Calafell}, \citenamefont {Dowd},
  \citenamefont {Hamel}, \citenamefont {Rozema}, \citenamefont {Brukner},\ and\
  \citenamefont {Walther}}]{procopio15}%
  \BibitemOpen
  \bibfield  {author} {\bibinfo {author} {\bibfnamefont {L.~M.}\ \bibnamefont
  {Procopio}}, \bibinfo {author} {\bibfnamefont {A.}~\bibnamefont {Moqanaki}},
  \bibinfo {author} {\bibfnamefont {M.}~\bibnamefont {Ara\'{u}jo}}, \bibinfo
  {author} {\bibfnamefont {F.}~\bibnamefont {Costa}}, \bibinfo {author}
  {\bibfnamefont {I.}~\bibnamefont {Alonso~Calafell}}, \bibinfo {author}
  {\bibfnamefont {E.~G.}\ \bibnamefont {Dowd}}, \bibinfo {author}
  {\bibfnamefont {D.~R.}\ \bibnamefont {Hamel}}, \bibinfo {author}
  {\bibfnamefont {L.~A.}\ \bibnamefont {Rozema}}, \bibinfo {author}
  {\bibfnamefont {{\v C}.}~\bibnamefont {Brukner}}, \ and\ \bibinfo {author}
  {\bibfnamefont {P.}~\bibnamefont {Walther}},\ }\href {\doibase
  10.1038/ncomms8913} {\bibfield  {journal} {\bibinfo  {journal} {Nat.
  Commun.}\ }\textbf {\bibinfo {volume} {6}},\ \bibinfo {pages} {7913}
  (\bibinfo {year} {2015})},\ \Eprint {http://arxiv.org/abs/1412.4006}
  {arXiv:1412.4006 [quant-ph]} \BibitemShut {NoStop}%
\bibitem [{\citenamefont {Rubino}\ \emph
  {et~al.}(2017{\natexlab{a}})\citenamefont {Rubino}, \citenamefont {Rozema},
  \citenamefont {Feix}, \citenamefont {Ara{\'u}jo}, \citenamefont {Zeuner},
  \citenamefont {Procopio}, \citenamefont {Brukner},\ and\ \citenamefont
  {Walther}}]{rubino17}%
  \BibitemOpen
  \bibfield  {author} {\bibinfo {author} {\bibfnamefont {G.}~\bibnamefont
  {Rubino}}, \bibinfo {author} {\bibfnamefont {L.~A.}\ \bibnamefont {Rozema}},
  \bibinfo {author} {\bibfnamefont {A.}~\bibnamefont {Feix}}, \bibinfo {author}
  {\bibfnamefont {M.}~\bibnamefont {Ara{\'u}jo}}, \bibinfo {author}
  {\bibfnamefont {J.~M.}\ \bibnamefont {Zeuner}}, \bibinfo {author}
  {\bibfnamefont {L.~M.}\ \bibnamefont {Procopio}}, \bibinfo {author}
  {\bibfnamefont {{\v C}.}~\bibnamefont {Brukner}}, \ and\ \bibinfo {author}
  {\bibfnamefont {P.}~\bibnamefont {Walther}},\ }\href {\doibase
  10.1126/sciadv.1602589} {\bibfield  {journal} {\bibinfo  {journal} {Sci.
  Adv.}\ }\textbf {\bibinfo {volume} {3}},\ \bibinfo {pages} {e1602589}
  (\bibinfo {year} {2017}{\natexlab{a}})},\ \Eprint
  {http://arxiv.org/abs/1608.01683} {arXiv:1608.01683 [quant-ph]} \BibitemShut
  {NoStop}%
\bibitem [{\citenamefont {Rubino}\ \emph
  {et~al.}(2017{\natexlab{b}})\citenamefont {Rubino}, \citenamefont {Rozema},
  \citenamefont {Massa}, \citenamefont {Ara{\'u}jo}, \citenamefont {Zych},
  \citenamefont {Brukner},\ and\ \citenamefont {Walther}}]{rubino17a}%
  \BibitemOpen
  \bibfield  {author} {\bibinfo {author} {\bibfnamefont {G.}~\bibnamefont
  {Rubino}}, \bibinfo {author} {\bibfnamefont {L.~A.}\ \bibnamefont {Rozema}},
  \bibinfo {author} {\bibfnamefont {F.}~\bibnamefont {Massa}}, \bibinfo
  {author} {\bibfnamefont {M.}~\bibnamefont {Ara{\'u}jo}}, \bibinfo {author}
  {\bibfnamefont {M.}~\bibnamefont {Zych}}, \bibinfo {author} {\bibfnamefont
  {{\v C}.}~\bibnamefont {Brukner}}, \ and\ \bibinfo {author} {\bibfnamefont
  {P.}~\bibnamefont {Walther}},\ }\href@noop {} {\  (\bibinfo {year}
  {2017}{\natexlab{b}})},\ \Eprint {http://arxiv.org/abs/1712.06884}
  {arXiv:1712.06884 [quant-ph]} \BibitemShut {NoStop}%
\bibitem [{\citenamefont {Goswami}\ \emph {et~al.}(2018)\citenamefont
  {Goswami}, \citenamefont {Giarmatzi}, \citenamefont {Kewming}, \citenamefont
  {Costa}, \citenamefont {Branciard}, \citenamefont {Romero},\ and\
  \citenamefont {White}}]{goswami18}%
  \BibitemOpen
  \bibfield  {author} {\bibinfo {author} {\bibfnamefont {K.}~\bibnamefont
  {Goswami}}, \bibinfo {author} {\bibfnamefont {C.}~\bibnamefont {Giarmatzi}},
  \bibinfo {author} {\bibfnamefont {M.}~\bibnamefont {Kewming}}, \bibinfo
  {author} {\bibfnamefont {F.}~\bibnamefont {Costa}}, \bibinfo {author}
  {\bibfnamefont {C.}~\bibnamefont {Branciard}}, \bibinfo {author}
  {\bibfnamefont {J.}~\bibnamefont {Romero}}, \ and\ \bibinfo {author}
  {\bibfnamefont {A.~G.}\ \bibnamefont {White}},\ }\href {\doibase
  10.1103/PhysRevLett.121.090503} {\bibfield  {journal} {\bibinfo  {journal}
  {Phys. Rev. Lett.}\ }\textbf {\bibinfo {volume} {121}},\ \bibinfo {pages}
  {090503} (\bibinfo {year} {2018})},\ \Eprint
  {http://arxiv.org/abs/1803.04302} {arXiv:1803.04302 [quant-ph]} \BibitemShut
  {NoStop}%
\bibitem [{\citenamefont {Wei}\ \emph {et~al.}(2018)\citenamefont {Wei},
  \citenamefont {Tischler}, \citenamefont {Zhao}, \citenamefont {Li},
  \citenamefont {Arrazola}, \citenamefont {Liu}, \citenamefont {Zhang},
  \citenamefont {Li}, \citenamefont {You}, \citenamefont {Wang}, \citenamefont
  {Chen}, \citenamefont {Sanders}, \citenamefont {Zhang}, \citenamefont
  {Pryde}, \citenamefont {Xu},\ and\ \citenamefont {Pan}}]{wei18}%
  \BibitemOpen
  \bibfield  {author} {\bibinfo {author} {\bibfnamefont {K.}~\bibnamefont
  {Wei}}, \bibinfo {author} {\bibfnamefont {N.}~\bibnamefont {Tischler}},
  \bibinfo {author} {\bibfnamefont {S.-R.}\ \bibnamefont {Zhao}}, \bibinfo
  {author} {\bibfnamefont {Y.-H.}\ \bibnamefont {Li}}, \bibinfo {author}
  {\bibfnamefont {J.~M.}\ \bibnamefont {Arrazola}}, \bibinfo {author}
  {\bibfnamefont {Y.}~\bibnamefont {Liu}}, \bibinfo {author} {\bibfnamefont
  {W.}~\bibnamefont {Zhang}}, \bibinfo {author} {\bibfnamefont
  {H.}~\bibnamefont {Li}}, \bibinfo {author} {\bibfnamefont {L.}~\bibnamefont
  {You}}, \bibinfo {author} {\bibfnamefont {Z.}~\bibnamefont {Wang}}, \bibinfo
  {author} {\bibfnamefont {Y.-A.}\ \bibnamefont {Chen}}, \bibinfo {author}
  {\bibfnamefont {B.~C.}\ \bibnamefont {Sanders}}, \bibinfo {author}
  {\bibfnamefont {Q.}~\bibnamefont {Zhang}}, \bibinfo {author} {\bibfnamefont
  {G.~J.}\ \bibnamefont {Pryde}}, \bibinfo {author} {\bibfnamefont
  {F.}~\bibnamefont {Xu}}, \ and\ \bibinfo {author} {\bibfnamefont {J.-W.}\
  \bibnamefont {Pan}},\ }\href@noop {} {\  (\bibinfo {year} {2018})},\ \Eprint
  {http://arxiv.org/abs/1810.10238} {arXiv:1810.10238 [quant-ph]} \BibitemShut
  {NoStop}%
\bibitem [{\citenamefont {Oreshkov}(2018{\natexlab{a}})}]{oreshkov18}%
  \BibitemOpen
  \bibfield  {author} {\bibinfo {author} {\bibfnamefont {O.}~\bibnamefont
  {Oreshkov}},\ }\href@noop {} {\  (\bibinfo {year} {2018}{\natexlab{a}})},\
  \Eprint {http://arxiv.org/abs/1801.07594} {arXiv:1801.07594 [quant-ph]}
  \BibitemShut {NoStop}%
\bibitem [{\citenamefont {Gutoski}\ and\ \citenamefont
  {Watrous}(2006)}]{gutoski06}%
  \BibitemOpen
  \bibfield  {author} {\bibinfo {author} {\bibfnamefont {G.}~\bibnamefont
  {Gutoski}}\ and\ \bibinfo {author} {\bibfnamefont {J.}~\bibnamefont
  {Watrous}},\ }in\ \href@noop {} {\emph {\bibinfo {booktitle} {Proceedings of
  39th ACM STOC}}}\ (\bibinfo {year} {2006})\ pp.\ \bibinfo {pages}
  {565--574},\ \Eprint {http://arxiv.org/abs/quant-ph/0611234}
  {arXiv:quant-ph/0611234} \BibitemShut {NoStop}%
\bibitem [{\citenamefont {{Chiribella}}\ \emph {et~al.}(2009)\citenamefont
  {{Chiribella}}, \citenamefont {{D'Ariano}},\ and\ \citenamefont
  {{Perinotti}}}]{chiribella09}%
  \BibitemOpen
  \bibfield  {author} {\bibinfo {author} {\bibfnamefont {G.}~\bibnamefont
  {{Chiribella}}}, \bibinfo {author} {\bibfnamefont {G.~M.}\ \bibnamefont
  {{D'Ariano}}}, \ and\ \bibinfo {author} {\bibfnamefont {P.}~\bibnamefont
  {{Perinotti}}},\ }\href {\doibase 10.1103/PhysRevA.80.022339} {\bibfield
  {journal} {\bibinfo  {journal} {Phys. Rev. A}\ }\textbf {\bibinfo {volume}
  {80}},\ \bibinfo {eid} {022339} (\bibinfo {year} {2009})},\ \Eprint
  {http://arxiv.org/abs/0904.4483} {arXiv:0904.4483 [quant-ph]} \BibitemShut
  {NoStop}%
\bibitem [{\citenamefont {Oreshkov}(2016)}]{oreshkov16a}%
  \BibitemOpen
  \bibfield  {author} {\bibinfo {author} {\bibfnamefont {O.}~\bibnamefont
  {Oreshkov}},\ }\href@noop {} {}\bibinfo {howpublished} {private
  communication} (\bibinfo {year} {2016})\BibitemShut {NoStop}%
\bibitem [{\citenamefont {Wechs}\ \emph {et~al.}()\citenamefont {Wechs},
  \citenamefont {Dourdent}, \citenamefont {Abbott},\ and\ \citenamefont
  {Branciard}}]{wechs18}%
  \BibitemOpen
  \bibfield  {author} {\bibinfo {author} {\bibfnamefont {J.}~\bibnamefont
  {Wechs}}, \bibinfo {author} {\bibfnamefont {H.}~\bibnamefont {Dourdent}},
  \bibinfo {author} {\bibfnamefont {A.~A.}\ \bibnamefont {Abbott}}, \ and\
  \bibinfo {author} {\bibfnamefont {C.}~\bibnamefont {Branciard}},\ }\href@noop
  {} {}\bibinfo {howpublished} {in preparation}\BibitemShut {NoStop}%
\bibitem [{\citenamefont {Branciard}\ \emph {et~al.}(2016)\citenamefont
  {Branciard}, \citenamefont {Ara\'{u}jo}, \citenamefont {Feix}, \citenamefont
  {Costa},\ and\ \citenamefont {Brukner}}]{branciard16}%
  \BibitemOpen
  \bibfield  {author} {\bibinfo {author} {\bibfnamefont {C.}~\bibnamefont
  {Branciard}}, \bibinfo {author} {\bibfnamefont {M.}~\bibnamefont
  {Ara\'{u}jo}}, \bibinfo {author} {\bibfnamefont {A.}~\bibnamefont {Feix}},
  \bibinfo {author} {\bibfnamefont {F.}~\bibnamefont {Costa}}, \ and\ \bibinfo
  {author} {\bibfnamefont {{\v{C}}.}~\bibnamefont {Brukner}},\ }\href {\doibase
  10.1088/1367-2630/18/1/013008} {\bibfield  {journal} {\bibinfo  {journal}
  {New J. Phys.}\ }\textbf {\bibinfo {volume} {18}},\ \bibinfo {pages} {013008}
  (\bibinfo {year} {2016})},\ \Eprint {http://arxiv.org/abs/1508.01704}
  {arXiv:1508.01704 [quant-ph]} \BibitemShut {NoStop}%
\bibitem [{\citenamefont {Abbott}\ \emph {et~al.}(2017)\citenamefont {Abbott},
  \citenamefont {Wechs}, \citenamefont {Costa},\ and\ \citenamefont
  {Branciard}}]{abbott17}%
  \BibitemOpen
  \bibfield  {author} {\bibinfo {author} {\bibfnamefont {A.~A.}\ \bibnamefont
  {Abbott}}, \bibinfo {author} {\bibfnamefont {J.}~\bibnamefont {Wechs}},
  \bibinfo {author} {\bibfnamefont {F.}~\bibnamefont {Costa}}, \ and\ \bibinfo
  {author} {\bibfnamefont {C.}~\bibnamefont {Branciard}},\ }\href {\doibase
  10.22331/q-2017-12-14-39} {\bibfield  {journal} {\bibinfo  {journal}
  {{Quantum}}\ }\textbf {\bibinfo {volume} {1}},\ \bibinfo {pages} {39}
  (\bibinfo {year} {2017})},\ \Eprint {http://arxiv.org/abs/1708.07663}
  {arXiv:1708.07663 [quant-ph]} \BibitemShut {NoStop}%
\bibitem [{\citenamefont {Rockafellar}(1970)}]{rockafellar70}%
  \BibitemOpen
  \bibfield  {author} {\bibinfo {author} {\bibfnamefont {R.~T.}\ \bibnamefont
  {Rockafellar}},\ }\href@noop {} {\emph {\bibinfo {title} {Convex Analysis}}}\
  (\bibinfo  {publisher} {Princeton University Press},\ \bibinfo {year}
  {1970})\BibitemShut {NoStop}%
\bibitem [{\citenamefont {Ac{\'{i}}n}\ \emph {et~al.}(2001)\citenamefont
  {Ac{\'{i}}n}, \citenamefont {Bru\ss}, \citenamefont {Lewenstein},\ and\
  \citenamefont {Sanpera}}]{Acin01}%
  \BibitemOpen
  \bibfield  {author} {\bibinfo {author} {\bibfnamefont {A.}~\bibnamefont
  {Ac{\'{i}}n}}, \bibinfo {author} {\bibfnamefont {D.}~\bibnamefont {Bru\ss}},
  \bibinfo {author} {\bibfnamefont {M.}~\bibnamefont {Lewenstein}}, \ and\
  \bibinfo {author} {\bibfnamefont {A.}~\bibnamefont {Sanpera}},\ }\href
  {\doibase 10.1103/PhysRevLett.87.040401} {\bibfield  {journal} {\bibinfo
  {journal} {Phys. Rev. Lett.}\ }\textbf {\bibinfo {volume} {87}},\ \bibinfo
  {pages} {040401} (\bibinfo {year} {2001})},\ \Eprint
  {http://arxiv.org/abs/quant-ph/0103025} {arXiv:quant-ph/0103025} \BibitemShut
  {NoStop}%
\bibitem [{\citenamefont {Svetlichny}(1987)}]{svetlichny87}%
  \BibitemOpen
  \bibfield  {author} {\bibinfo {author} {\bibfnamefont {G.}~\bibnamefont
  {Svetlichny}},\ }\href {\doibase 10.1103/PhysRevD.35.3066} {\bibfield
  {journal} {\bibinfo  {journal} {Phys. Rev. D}\ }\textbf {\bibinfo {volume}
  {35}},\ \bibinfo {pages} {3066} (\bibinfo {year} {1987})}\BibitemShut
  {NoStop}%
\bibitem [{\citenamefont {Gallego}\ \emph {et~al.}(2012)\citenamefont
  {Gallego}, \citenamefont {W\"{u}rflinger}, \citenamefont {Ac{\'{i}}n},\ and\
  \citenamefont {Navascu\'{e}s}}]{gallego12}%
  \BibitemOpen
  \bibfield  {author} {\bibinfo {author} {\bibfnamefont {R.}~\bibnamefont
  {Gallego}}, \bibinfo {author} {\bibfnamefont {L.~E.}\ \bibnamefont
  {W\"{u}rflinger}}, \bibinfo {author} {\bibfnamefont {A.}~\bibnamefont
  {Ac{\'{i}}n}}, \ and\ \bibinfo {author} {\bibfnamefont {M.}~\bibnamefont
  {Navascu\'{e}s}},\ }\href {\doibase 10.1103/PhysRevLett.109.070401}
  {\bibfield  {journal} {\bibinfo  {journal} {Phys. Rev. Lett.}\ }\textbf
  {\bibinfo {volume} {109}},\ \bibinfo {pages} {070401} (\bibinfo {year}
  {2012})},\ \Eprint {http://arxiv.org/abs/1112.2647} {arXiv:1112.2647
  [quant-ph]} \BibitemShut {NoStop}%
\bibitem [{\citenamefont {Bancal}\ \emph {et~al.}(2013)\citenamefont {Bancal},
  \citenamefont {Barrett}, \citenamefont {Gisin},\ and\ \citenamefont
  {Pironio}}]{bancal13}%
  \BibitemOpen
  \bibfield  {author} {\bibinfo {author} {\bibfnamefont {J.-D.}\ \bibnamefont
  {Bancal}}, \bibinfo {author} {\bibfnamefont {J.}~\bibnamefont {Barrett}},
  \bibinfo {author} {\bibfnamefont {N.}~\bibnamefont {Gisin}}, \ and\ \bibinfo
  {author} {\bibfnamefont {S.}~\bibnamefont {Pironio}},\ }\href {\doibase
  10.1103/PhysRevA.88.014102} {\bibfield  {journal} {\bibinfo  {journal} {Phys.
  Rev. A}\ }\textbf {\bibinfo {volume} {88}},\ \bibinfo {pages} {014102}
  (\bibinfo {year} {2013})},\ \Eprint {http://arxiv.org/abs/1112.2626}
  {arXiv:1112.2626 [quant-ph]} \BibitemShut {NoStop}%
\bibitem [{\citenamefont {Jia}\ and\ \citenamefont {Sakharwade}(2018)}]{Jia18}%
  \BibitemOpen
  \bibfield  {author} {\bibinfo {author} {\bibfnamefont {D.}~\bibnamefont
  {Jia}}\ and\ \bibinfo {author} {\bibfnamefont {N.}~\bibnamefont
  {Sakharwade}},\ }\href {\doibase 10.1103/PhysRevA.97.032110} {\bibfield
  {journal} {\bibinfo  {journal} {Phys. Rev. A}\ }\textbf {\bibinfo {volume}
  {97}},\ \bibinfo {pages} {032110} (\bibinfo {year} {2018})},\ \Eprint
  {http://arxiv.org/abs/1706.05532} {arXiv:1706.05532 [quant-ph]} \BibitemShut
  {NoStop}%
\bibitem [{\citenamefont {Gu\'{e}rin}\ \emph {et~al.}(2018)\citenamefont
  {Gu\'{e}rin}, \citenamefont {Krumm}, \citenamefont {Budroni},\ and\
  \citenamefont {Brukner}}]{Guerin18}%
  \BibitemOpen
  \bibfield  {author} {\bibinfo {author} {\bibfnamefont {P.~A.}\ \bibnamefont
  {Gu\'{e}rin}}, \bibinfo {author} {\bibfnamefont {M.}~\bibnamefont {Krumm}},
  \bibinfo {author} {\bibfnamefont {C.}~\bibnamefont {Budroni}}, \ and\
  \bibinfo {author} {\bibfnamefont {{\v{C}}.}~\bibnamefont {Brukner}},\
  }\href@noop {} {\  (\bibinfo {year} {2018})},\ \Eprint
  {http://arxiv.org/abs/1806.10374} {arXiv:1806.10374 [quant-ph]} \BibitemShut
  {NoStop}%
\bibitem [{\citenamefont {Feix}\ \emph {et~al.}(2016)\citenamefont {Feix},
  \citenamefont {Ara{\'u}jo},\ and\ \citenamefont {Brukner}}]{feix16}%
  \BibitemOpen
  \bibfield  {author} {\bibinfo {author} {\bibfnamefont {A.}~\bibnamefont
  {Feix}}, \bibinfo {author} {\bibfnamefont {M.}~\bibnamefont {Ara{\'u}jo}}, \
  and\ \bibinfo {author} {\bibfnamefont {{\v C}.}~\bibnamefont {Brukner}},\
  }\href {http://stacks.iop.org/1367-2630/18/i=8/a=083040} {\bibfield
  {journal} {\bibinfo  {journal} {New. J. Phys.}\ }\textbf {\bibinfo {volume}
  {18}},\ \bibinfo {pages} {083040} (\bibinfo {year} {2016})},\ \Eprint
  {http://arxiv.org/abs/1604.03391} {arXiv:1604.03391 [quant-ph]} \BibitemShut
  {NoStop}%
\bibitem [{\citenamefont {Miklin}\ \emph {et~al.}(2017)\citenamefont {Miklin},
  \citenamefont {Abbott}, \citenamefont {Branciard}, \citenamefont {Chaves},\
  and\ \citenamefont {Budroni}}]{miklin17}%
  \BibitemOpen
  \bibfield  {author} {\bibinfo {author} {\bibfnamefont {N.}~\bibnamefont
  {Miklin}}, \bibinfo {author} {\bibfnamefont {A.~A.}\ \bibnamefont {Abbott}},
  \bibinfo {author} {\bibfnamefont {C.}~\bibnamefont {Branciard}}, \bibinfo
  {author} {\bibfnamefont {R.}~\bibnamefont {Chaves}}, \ and\ \bibinfo {author}
  {\bibfnamefont {C.}~\bibnamefont {Budroni}},\ }\href {\doibase
  10.1088/1367-2630/aa8f9f} {\bibfield  {journal} {\bibinfo  {journal} {New J.
  Phys.}\ }\textbf {\bibinfo {volume} {19}},\ \bibinfo {pages} {113041}
  (\bibinfo {year} {2017})},\ \Eprint {http://arxiv.org/abs/1706.10270}
  {arXiv:1706.10270 [quant-ph]} \BibitemShut {NoStop}%
\bibitem [{\citenamefont {Milz}\ \emph {et~al.}(2018)\citenamefont {Milz},
  \citenamefont {Pollock}, \citenamefont {Le}, \citenamefont {Chiribella},\
  and\ \citenamefont {Modi}}]{milz18}%
  \BibitemOpen
  \bibfield  {author} {\bibinfo {author} {\bibfnamefont {S.}~\bibnamefont
  {Milz}}, \bibinfo {author} {\bibfnamefont {F.~A.}\ \bibnamefont {Pollock}},
  \bibinfo {author} {\bibfnamefont {T.~P.}\ \bibnamefont {Le}}, \bibinfo
  {author} {\bibfnamefont {G.}~\bibnamefont {Chiribella}}, \ and\ \bibinfo
  {author} {\bibfnamefont {K.}~\bibnamefont {Modi}},\ }\href
  {http://stacks.iop.org/1367-2630/20/i=3/a=033033} {\bibfield  {journal}
  {\bibinfo  {journal} {New J. Phys.}\ }\textbf {\bibinfo {volume} {20}},\
  \bibinfo {pages} {033033} (\bibinfo {year} {2018})},\ \Eprint
  {http://arxiv.org/abs/1711.04065} {arXiv:1711.04065 [quant-ph]} \BibitemShut
  {NoStop}%
\bibitem [{\citenamefont {Oreshkov}(2018{\natexlab{b}})}]{oreshkov18a}%
  \BibitemOpen
  \bibfield  {author} {\bibinfo {author} {\bibfnamefont {O.}~\bibnamefont
  {Oreshkov}},\ }\href@noop {} {}\bibinfo {howpublished} {private
  communication} (\bibinfo {year} {2018}{\natexlab{b}})\BibitemShut {NoStop}%
\bibitem [{\citenamefont {Rudin}(1976)}]{rudin76}%
  \BibitemOpen
  \bibfield  {author} {\bibinfo {author} {\bibfnamefont {W.}~\bibnamefont
  {Rudin}},\ }\href@noop {} {\emph {\bibinfo {title} {Principles of
  Mathematical Analysis}}},\ \bibinfo {edition} {3rd}\ ed.\ (\bibinfo
  {publisher} {McGraw-Hill},\ \bibinfo {year} {1976})\BibitemShut {NoStop}%
\bibitem [{\citenamefont {Nesterov}\ and\ \citenamefont
  {Nemirovskii}(1994)}]{nesterov94}%
  \BibitemOpen
  \bibfield  {author} {\bibinfo {author} {\bibfnamefont {Y.}~\bibnamefont
  {Nesterov}}\ and\ \bibinfo {author} {\bibfnamefont {A.}~\bibnamefont
  {Nemirovskii}},\ }\href {\doibase 10.1137/1.9781611970791} {\emph {\bibinfo
  {title} {Interior-Point Polynomial Algorithms in Convex Programming}}},\
  \bibinfo {series} {Stud. Appl. Math}, Vol.~\bibinfo {volume} {32}\ (\bibinfo
  {publisher} {SIAM},\ \bibinfo {address} {Philadelphia},\ \bibinfo {year}
  {1994})\BibitemShut {NoStop}%
\bibitem [{\citenamefont {Grant}\ and\ \citenamefont {Boyd}(2014)}]{cvx}%
  \BibitemOpen
  \bibfield  {author} {\bibinfo {author} {\bibfnamefont {M.}~\bibnamefont
  {Grant}}\ and\ \bibinfo {author} {\bibfnamefont {S.}~\bibnamefont {Boyd}},\
  }\href {http://cvxr.com/cvx} {\enquote {\bibinfo {title} {{CVX}: Matlab
  software for disciplined convex programming, v2.1},}\ } (\bibinfo {year}
  {2014})\BibitemShut {NoStop}%
\end{thebibliography}%

\end{document}